\documentclass{article}

\usepackage{amsmath}
\usepackage{amsthm}
\usepackage{amssymb}
\allowdisplaybreaks[4]

\usepackage[top=2.5cm, bottom=2.5cm, left=2.4cm, right=2.4cm]{geometry}
\usepackage{indentfirst}
\usepackage{bm}
\usepackage{bbm}

\usepackage{tikz}
\usetikzlibrary{shapes.geometric, arrows}
\usepackage{booktabs}
\usepackage{multirow}
\usepackage{multicol}
\usepackage{makecell}
\usepackage{algorithm}
\usepackage{algpseudocode}

\usepackage[numbers]{natbib}

\usepackage{charter}

\usepackage{hyperref}
\hypersetup{
    colorlinks=true,
    linkcolor=blue,
    filecolor=magenta,
    urlcolor=cyan,
    citecolor = cyan,
        }

\newcommand{\Acal}{\mathcal{A}}
\newcommand{\Bcal}{\mathcal{B}}

\newcommand{\Ecal}{\mathcal{E}}
\newcommand{\Fcal}{\mathcal{F}}
\newcommand{\Gcal}{\mathcal{G}}

\newcommand{\Mcal}{\mathcal{M}}
\newcommand{\Ncal}{\mathcal{N}}

\newcommand{\Pcal}{\mathcal{P}}

\newcommand{\Scal}{\mathcal{S}}

\newcommand{\Vcal}{\mathcal{V}}

\newcommand{\fbf}{\mathbf{f}}
\newcommand{\gbf}{\mathbf{g}}

\newcommand{\sbf}{\mathbf{s}}

\newcommand{\ubf}{\mathbf{u}}
\newcommand{\vbf}{\mathbf{v}}
\newcommand{\wbf}{\mathbf{w}}
\newcommand{\xbf}{\mathbf{x}}
\newcommand{\ybf}{\mathbf{y}}

\newcommand{\Abf}{\mathbf{A}}

\newcommand{\Gbf}{\mathbf{G}}

\newcommand{\Ibf}{\mathbf{I}}

\newcommand{\Wbf}{\mathbf{W}}

\newcommand{\Ybf}{\mathbf{Y}}

\newcommand{\Ebb}{\mathbb{E}}
\newcommand{\Pbb}{\mathbb{P}}

\newcommand{\id}{\mathbb{I}}

\newcommand{\Var}{\mathrm{Var}}

\newcommand{\set}[1]{\{#1\}}

\newcommand{\Bigset}[1]{\Bigl\{#1\Bigr\}}

\newcommand{\norm}[1]{\left\Vert #1 \right\Vert}
\newcommand{\bignorm}[1]{\bigl\lVert #1\bigr\rVert}
\newcommand{\Bignorm}[1]{\Bigl\lVert #1\Bigr\rVert}
\newcommand{\biggnorm}[1]{\biggl\lVert#1\biggr\rVert}

\newcommand{\abs}[1]{|#1|}

\newcommand{\Bigabs}[1]{\Bigl| #1\Bigr|}
\newcommand{\biggabs}[1]{\biggl|#1\biggr|}

\newcommand{\size}[1]{|#1|}

\DeclareMathOperator*{\argmin}{arg\,min}

\newtheorem{theorem}{Theorem}[section]
\newtheorem{lemma}[theorem]{Lemma}
\newtheorem{prop}[theorem]{Proposition}

\newtheorem{condition}[theorem]{Condition}
\newtheorem{coro}[theorem]{Corollary}
\newtheorem{definition}[theorem]{Definition}
\newtheorem{rmk}[theorem]{Remark}

\newcommand{\revise}[1]{#1}

\usepackage{subcaption}
\usepackage{enumerate}
\usepackage[normalem]{ulem}
\newcommand{\btheta}{{\bm{\theta}}}

\renewcommand{\citep}{\cite}
\renewcommand{\citet}{\cite}

\title{Toward Exact Convergence in Byzantine-Robust Decentralized Learning: A Statistical Identification Approach}

\author{Siyuan Zhang\thanks{State Key Laboratory of Scientific and Engineering Computing, Academy of Mathematics and Systems Science, Chinese Academy of Sciences, and University of Chinese Academy of Sciences, China (e-mail: zhangsiyuan@amss.ac.cn)}, ~
Chengde Qian\thanks{School of Mathematical Sciences, Shanghai Jiao Tong University, China.  (e-mail: qianchd@gmail.com).},
~
Xin Liu\thanks{State Key Laboratory of Scientific and Engineering Computing, Academy of Mathematics and Systems Science, Chinese Academy of Sciences, and University of Chinese Academy of Sciences, China (e-mail: liuxin@lsec.cc.ac.cn)}
~ and  ~
Changliang Zou\thanks{School of Statistics and Data Sciences, Nankai University, China (e-mail: zoucl@nankai.edu.cn)}
}

\begin{document}
\maketitle

\begin{abstract}%
To defend against Byzantine attacks in decentralized learning, most existing methods rely on robust aggregation rules to mitigate the influence of malicious machines. However, these strategies inherently introduce bias, leading to inexact convergence with non-vanishing steady-state errors. In this paper, we propose a strategic shift from passive aggregation to active identification by introducing the Decentralized Rescaled Stochastic Gradient Descent with Byzantine Machine Identification (DRSGD-ByMI) framework. The core of our approach is an identification-based ``detect-then-optimize'' pipeline, where a p-value-free detection procedure is developed to accurately prune malicious nodes from the network. By leveraging sample-splitting score statistics, this identification mechanism achieves false discovery rate control without requiring restrictive distributional assumptions. We theoretically demonstrate that this precise identification allows the decentralized network to \revise{recover sufficient connectivity among the normal nodes}, \revise{thereby enabling DRSGD-ByMI to match, even in the presence of Byzantine machines, the same order-optimal convergence rate as standard decentralized stochastic first-order methods.} Numerical experiments validate our theoretical results and demonstrate the effectiveness of DRSGD-ByMI for decentralized robust learning problems.
\end{abstract}


\section{Introduction}

Decentralized learning has emerged as a critical paradigm for training global models across distributed machines, avoiding communication bottlenecks that occur in distributed systems with a central server.
Within a decentralized network topology, machines exchange messages only with their neighbors, which is communication-efficient and helps alleviate growing privacy concerns.
However, in practical large-scale systems, transmitted messages are highly susceptible to corruption due to various system uncertainties \citep{blanchard2017machine}.
Inevitably, a subset of machines will transmit unreliable or malicious messages.
These anomalous nodes are formally designated as \textit{Byzantine machines}, and the resulting disruptions are termed \textit{Byzantine attacks}.
Learning a reliable model in this regime is particularly challenging because the specific identities of the compromised machines remain unknown to the system.

A large number of studies have explored Byzantine-robust decentralized optimization.
Some studies have focused on deterministic settings \citep{yang-byrdie-2019, xu2018robust, fang-bridge-2022, gupta2021byzantine, ravi2019detection, kuwaranancharoen2020byzantine, kuwaranancharoen2024scalable, kuwaranancharoen2025geometric}, \revise{typically through} designing rules to make local updates insensitive to outlier neighbors.
Because computing exact \revise{local} gradients is often impractical, more attention has shifted to stochastic optimization \citep{guo2021byzantine, fang-byzantine-robust-2024, wu-byzantine-resilient-2023, yang-byzantine-robust-2024, he2022byzantine, peng2023byzantine, peng2021byzantine}.
In this area, several different strategies have been proposed.
For example, UBAR \citep{guo2021byzantine} improves robustness by using only neighbors with small parameter distances and lower loss values.
Centered-clipping methods \citep{he2022byzantine, yang-byzantine-robust-2024} limit the update size to defend against attacks, converging to a neighborhood of the stationary point.
However, the final error does not vanish and depends on the weights given to malicious peers and gradient differences.
Other works like \citet{peng2021byzantine} and \citet{peng2023byzantine} apply sign-based methods \citep{xu2018robust} to stochastic settings, achieving linear convergence with an error related to the number of Byzantine neighbors.
Additionally, BALANCE \citep{fang-byzantine-robust-2024} uses a decay rule to reduce the weight of suspicious nodes based on their distance.
\citet{wu-byzantine-resilient-2023} propose IOS, an iterative detection method analogous to the centralized FABA \citep{ijcai2019p670}, to locally compute robust centers.
Despite their diverse defensive mechanisms, these methods suffer from a shared limitation. Specifically, the interplay between Byzantine interference and the inherent bias of robust aggregation inevitably induces an optimization error. As a result, such approaches generally only guarantee inexact convergence, suffering from a non-vanishing steady-state error (see Table~\ref{tab:1} below).

\begin{table}[htbp]
  \centering
      \fontsize{8}{11}\selectfont
      \begin{tabular}{|c|c|c|c|c|}
      \hline
      {Existing works} &  {Convexity} & {Step-sizes} & {Aggregation} & {Convergence}  \\
      \hline
      \makecell[c]{\citep{peng2020byzantine},\\
      \citep{peng2021byzantine}}    & strongly convex & $\eta_{k}=\mathcal{O}(\frac{1}{k})$&  sign mapping & inexact + linear  \\
      \hline
      BRAVO \citep{peng2023byzantine}   &  strongly convex & constant &  sign mapping & inexact + linear   \\
      \hline
      UBAR \citep{guo2021byzantine}   & nonconvex & constant &  detection $+$ average  & No convergence theory  \\
      \hline
      SCCLIP \citep{he2022byzantine}    & nonconvex &  $\eta_{k}=\mathcal{O}(\frac{1}{\sqrt{K}})$ &  centered clipping & inexact + sublinear  \\
      \hline
      \makecell[c]{DSGD-RTC \\ \citep{yang-byzantine-robust-2024}}  & nonconvex &  $\eta_{k}=\mathcal{O}(\frac{1}{\sqrt{K}})$ &  centered clipping & inexact + sublinear  \\
      \hline
      IOS \citep{wu-byzantine-resilient-2023}  & nonconvex & $\eta_{k}=\mathcal{O}(\frac{1}{\sqrt{K}})$ &  detection $+$ (w)average & inexact + sublinear \\
     \hline
     BALANCE \citep{fang-byzantine-robust-2024}  & \makecell[c]{{strongly convex} \\  {nonconvex}} & constant & detection $+$ average & \makecell[c]{{inexact + linear} \\  {inexact + sublinear}}  \\
      \hline
     \revise{DRSGD-ByMI (this work)}  & nonconvex & $\eta_{k}=\mathcal{O}(\frac{1}{\sqrt{K}})$  & \makecell[c]{{any robust aggregation}\\ {satisfying Condition~\ref{asp:warm_up}}}  & {\textbf{exact} + sublinear}   \\
     \hline
     \end{tabular}
  \caption{Comparison of existing works in decentralized Byzantine-robust stochastic optimization. $K$ denotes total iterations.}
  \label{tab:1}
\end{table}

Building on these limitations, another promising direction is to integrate statistically principled Byzantine machine identification into decentralized optimization.
With statistically reliable identification and removal of Byzantine connections, exact convergence becomes achievable.
However, traditional outlier detection tests rely on p-values derived from asymptotic distributions (e.g., Hotelling’s $T^2$).
The accuracy of these tests degrades as the parameter dimension increases relative to the sample size, making them unreliable for distributed machine learning tasks characterized by high dimensionality.
Recently, in a centralized setup, \citet{qian2024bymi} proposed a sample-splitting procedure that used a general robust estimator to construct score statistics.
The approach is p-value-free and achieves finite-sample false discovery rate control.
Nevertheless, decentralized tasks are characterized by parameter heterogeneity across distributed machines, which prevents the direct adoption of centralized detection protocols and requires novel methodological designs.
In this direction, there are only a limited number of studies \citep{kailkhura2017byzantine, kailkhura2016data} that attempt to detect Byzantine machines through hypothesis testing.
However, these works are all restricted to scalar-valued problems and require prior conditions on the distribution in Byzantine machines (e.g., the Gaussian mixture condition in \citet{kailkhura2016data}).

Inspired by \citet{qian2024bymi}, we develop DRSGD-ByMI, a \textit{p-value-free} and \textit{dimension-insensitive} detect-and-optimize framework that integrates Byzantine machine identification with decentralized rescaled stochastic gradient descent.
The framework of DRSGD-ByMI is illustrated in Figure~\ref{fig:framework1} and the main contributions are summarized as follows.

\begin{itemize}
\item We design a novel framework which consists of three phases: warm-up, detection, and optimization.
It is flexible and can be seamlessly integrated with various existing decentralized Byzantine-robust algorithms during the warm-up phase.
Moreover, by incorporating a p-value-free Byzantine identification mechanism into the decentralized system and removing Byzantine edges, our framework allows for exact convergence by means of a decentralized rescaled stochastic gradient descent (DRSGD).

\item The proposed detection procedure guarantees finite-sample false discovery rate control and high-probability sure detection.
Based on these results, we prove that the resulting sub-graph of normal nodes forms a strongly connected component, enabling the optimization algorithm to operate effectively on the pruned graph.

\item The proposed algorithm achieves a high-probability non-asymptotic convergence rate of $\mathcal{O}(1/\sqrt{m_g K})$ in the nonconvex setting, where $K$ denotes the total number of iterations, $m_g$ denotes the number of normal machines.
\revise{This rate is order-optimal, matching that of standard Byzantine-free decentralized stochastic first-order methods \citep{lian2017can, yuan2022revisiting}.}
Numerical experiments demonstrate that DRSGD-ByMI not only accurately identifies Byzantine machines but also improves the robustness and accuracy of decentralized learning.

\end{itemize}

\begin{figure}[htbp]
    \centering
   \includegraphics[width=15cm, height= 8.4cm]{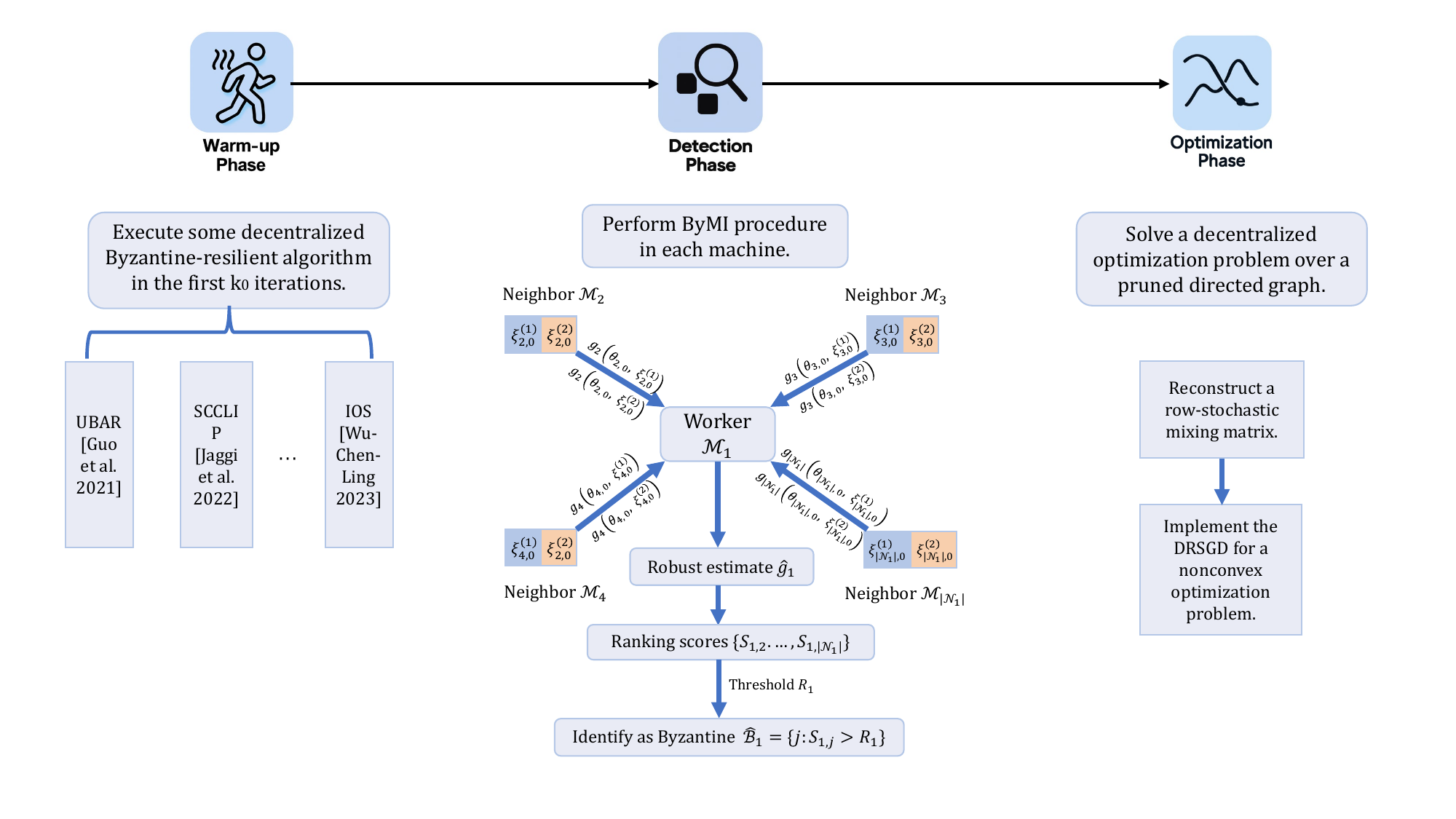}
   \caption{The framework of our proposed method: DRSGD-ByMI}
  \label{fig:framework1}
\end{figure}

\subsection{Organization}

The remainder of this paper is organized as follows.
Section~\ref{sec:pre} formulates the decentralized problem of interest and establishes the preliminaries on decentralized optimization and the Byzantine model, then formulates our ``detect-then-optimize'' problem setup.
Section~\ref{sec:frame} \revise{presents} the DRSGD-ByMI algorithmic framework.
Section~\ref{sec:gua} establishes statistical guarantees for identification and convergence guarantees for the optimization.
Finally, Section~\ref{sec:numeric} presents numerical results, followed by conclusions and future directions.

\subsection{Notation and Basic Terminology}

In this work, we denote scalars, vectors, and matrices by regular-font, bold lower-case, and bold upper-case letters, respectively.
The notations $``\langle \cdot, \cdot \rangle"$ and $``\| \cdot \|"$ represent the standard
inner product and norm of vectors, respectively.
We use $``\|\cdot \|_{2}"$ and $``\|\cdot \|_{\mathrm{F}}"$
particularly for the 2-norm and Frobenius norm of matrices, respectively.
For any semi-positive matrix ${\mathbf {\Omega}}$, $\|\cdot \|_{\bf \Omega}$ stands for the norm induced by $\mathbf {\Omega}$.
We also use $``\mathrm{Diag}(\cdot)"$ and $``\mathrm{diag}(\cdot)"$ to denote the diagonal matrix formed from an input matrix, and an input vector, respectively.
Moreover, the notation $\rho(\Abf)$ denotes the spectral radius of matrix $\Abf$, and $\lambda_2(\Abf)$ denotes the second-largest eigenvalue modulus of $\Abf$.

In a graph $\mathtt{G}=(\mathcal{V}, \mathcal{E})$, we use $\mathcal{N}_{i}$ to denote the neighbor set of node $i$.
When $\mathtt{G}$ is directed, we use $\mathcal{N}^{\text{in}}_{i}$ and $\mathcal{N}^{\text{out}}_{i}$ to denote the in-neighbor set and out-neighbor set of node $i$, respectively.
More specifically, $\mathcal{N}^{\text{in}}_{i}$ consists of the nodes that have directed arcs pointing to node $i$, while $\mathcal{N}^{\text{out}}_{i}$ consists of the nodes to which node $i$ has directed arcs.
Throughout the paper, we include node $i$ itself in these sets whenever self-loops are allowed by the mixing matrix.
For any subset $\Acal \subseteq \Vcal$, we use $\Ecal|_{\Acal}$ to denote the set of edges (or arcs) in $\Ecal$ whose two endpoints both belong to $\Acal$.
For a directed graph $\mathtt{G}$, a sub-graph $\mathtt{G}'$ is said to be strongly connected if every node in $\mathtt{G}'$ is reachable from every other node in $\mathtt{G}'$.
A sub-graph $\mathtt{G}'$ is called a strongly connected component if it is strongly connected and is not a proper sub-graph of any other strongly connected sub-graph of $\mathtt{G}$.

In general, we use $\{\btheta_{i,k}\}$ to represent the local optimization variable at the $k$-th iteration, and use $\bar{\btheta}_k$, $\tilde{\btheta}_k$, and $\widehat{\btheta}_{i, k}$ as the global average, global weighted average, and local robust average of the optimization variable $\{\btheta_{i,k}\}$ at the $k$-th iteration, respectively.
For any $a,b \in \mathbb{R},$ $a \vee b$ denotes the maximum of $a$ and $b$.
For any set $\Acal$, $\Acal^{c}$ denotes its complement, $\size{\Acal}$ denotes the number of elements in set $\Acal$, and $\mathbb{I}({\Acal})$ represents the indicator function on the set $\Acal$. \revise{For any two sets $\Acal$ and $\Bcal$, $\Acal \setminus \Bcal$ denotes the set difference $\{x \in \Acal: x \notin \Bcal\}$.} The notation ${\bm 1}_d$ stands for the all-ones vector in $\mathbb{R}^d$, ${\bm 1}_{\{i\}}$ stands for the vector whose $i$-th component is $1$ and all other components are $0$. \revise{For a vector $\vbf$, $[\vbf]_i$ denotes its $i$-th component.} We use $\sbf$ to denote a single sample from some distribution, and $\xi$ to denote a stochastic mini-batch sampled from some distribution.

\section{Problem Formulation and Preliminaries}\label{sec:pre}

Consider an undirected connected graph $\mathtt{G}=(\Vcal, \Ecal)$, where the node set $\Vcal:= \Gcal \cup \Bcal$ represents the indices of machines $\{\mathcal{M}_{i}: i\in[m]\}$, $\Gcal$ and $\Bcal$ respectively denote the node sets of normal machines and Byzantine machines, and the edge set $\Ecal$ represents the communication links between machines.
If $(i, j) \in \Ecal,$ machine $\Mcal_i$ and $\Mcal_j$ are neighbors and can communicate with each other.
Suppose that the total number of Byzantine machines is $\size{\Bcal} = \lfloor \varrho m \rfloor$, $\varrho \in (0, \frac{1}{2})$.
Meanwhile, denote $m_g := \size{\Gcal}$, and the set of edges between nodes in $\mathcal{G}$ as $\mathcal{E}|_{\mathcal{G}}$.
Note that each normal machine has no prior knowledge of either the number or the identities of its Byzantine neighbors.

In such a Byzantine environment, each node $i$ collects a local dataset $\Scal_i=\{\sbf^{(i)}_1, \ldots, \sbf_{N_i}^{(i)}\}$ with $N_i$ samples, and is assigned a local objective function $f_i(\btheta)$ associated with $\Scal_i$.
We aim to solve an empirical risk minimization problem over $\mathtt{G}$ in a decentralized manner,
\begin{equation}\label{obj:dcentral_all}
  \begin{aligned}
  \min_{{\bm \theta}_{i}\in \mathbb{R}^{d}, i\in \Gcal}  &~  \sum\limits_{i\in \Gcal} f_i(\btheta_i)\\
  \text{s. t.} &~ {\bm \theta}_{i} = {\bm \theta}_{j}, (i,j)\in \Ecal|_{\Gcal}.\\
  \end{aligned}
\end{equation}
Here, ${\btheta}_{i}$ is the local optimization variable of node $i$.
The local objective function $f_i(\btheta) := {N}_i^{-1} \sum_{u=1}^{N_i}\ell(\btheta; \sbf_{u}^{(i)})$, where $\ell$ is the empirical loss function that is continuously differentiable and possibly nonconvex.
We denote $f(\btheta) := {m}_g^{-1}\sum_{i\in \Gcal} f_i(\btheta)$.

Moreover, we characterize the Byzantine attacks via the Huber contamination model \citep{huber1964huber}.
Specifically, each sample $\sbf_{u}^{(i)}$ in the local dataset $\Scal_i$ is drawn from a distribution that depends on whether node $i$ is normal or Byzantine, i.e.,
\begin{equation}\label{eq:huber}
\begin{cases}
\sbf_{u}^{(i)} \sim  \Pcal, & \text{ if } i \in \mathcal{G},\\
\sbf_{u}^{(i)}  \sim \mathcal{Q}_i \neq \mathcal{P},  & \text{ if } i \in \mathcal{B},
\end{cases}
\end{equation}
where $\mathcal{Q}_i$ represents an arbitrary adversarial distribution.
That is to say, the data distributions across normal nodes are homogeneous, whereas the distributions at Byzantine nodes may deviate significantly from $\Pcal$.

To guarantee that problem \eqref{obj:dcentral_all} is well-posed, we impose the following connectivity condition on graph $\mathtt{G}$:

\begin{condition}\label{asp:scc_graph}
The graph $\mathtt{G}=(\mathcal{V}, \mathcal{E})$ is connected, and the sub-graph of normal machines $(\mathcal{G}, \mathcal{E}|_{\mathcal{G}})$ is connected.
\end{condition}

In the absence of Byzantine nodes, \eqref{obj:dcentral_all} is generally solved by the popular decentralized stochastic gradient descent (DSGD) algorithm \citep{lian2017can}.
At iteration $k$, each node $i$ independently samples a mini-batch $\xi_{i,k}$ from $\Scal_i$, and computes the stochastic gradient by $\gbf_i(\btheta_{i, k}; \xi_{i, k}):= \frac{1}{\size{\xi_{i, k}}} \sum_{\sbf \in \xi_{i, k}} \nabla \ell(\btheta_{i, k}; \sbf)$.
Let $(\Omega,\mathcal{F},\mathbb{P})$ denote the probability space associated with the per-iteration sampling randomness. Define $\mathcal{F}_{k} := \sigma(\{\xi_{i, t}: t\leq k-1, i \in \mathcal{V}\})$, and $\mathcal{F}_0 := \{\Omega, \varnothing\}$.
The local oracle at node $i$ satisfies
\begin{equation}
\mathbb{E}[\gbf_{i}(\btheta_{i,k}; \xi_{i, k}) | \mathcal{F}_{k}]= \nabla f_i(\btheta_{i, k}).
\end{equation}
Meanwhile, node $i$ communicates with its neighbors and aggregates their local optimization variables by $\sum_{j \in \Ncal_i} {\Wbf}(i,j){\bm  \theta}_{j, k}$, where ${\Wbf}$ is a symmetric, doubly stochastic mixing matrix satisfying ${\Wbf}(i,j)=0$ iff $i\neq j$ and $(i, j) \notin \Ecal$ \citep{xiao2006distributed}.
Finally, node $i$ performs a decentralized stochastic gradient descent step by
\begin{equation}
    \label{Eq_intro_DSGD0}
    {\bm \theta}_{i, k+1} = \sum\limits_{j \in {\Ncal_i}} {\Wbf}(i,j){\bm  \theta}_{j, k} - \eta_k \gbf_{i}(\btheta_{i,k}; \xi_{i, k}).
\end{equation}

If the graph $(\Gcal, \Ecal|_{\Gcal})$ is directed, each node can only aggregate the variables received from its in-neighbors.
The mixing matrix ${\Wbf}$ is restricted to be row-stochastic and asymmetric.
DGD-RS algorithm, proposed by \citet{mai2016distributed}, introduces auxiliary variables $\mathbf{y}_{i,k}$ to track the Perron left eigenvector $\vbf_1$ of the mixing matrix $\mathbf{W}$ corresponding to the eigenvalue $1$, and updates them by
\begin{equation}\label{eq:DGDRS}
\begin{cases}
{\bm \theta}_{i,k+1} & = \sum\limits_{j \in \Ncal_{i}^{\text{in}}} {\Wbf} (i,j) {\bm \theta}_{j,k} - \eta_k \frac{\nabla f_{i}(\btheta_{i,k})}{[\ybf_{i,k}]_i},\\
\ybf_{i,k+1} & = \sum\limits_{j \in \Ncal_{i}^{\text{in}}} {\Wbf}(i,j) \ybf_{j,k}, \\
\end{cases}
\end{equation}
where $\nabla f_{i}(\btheta_{i,k})$ is an exact local gradient at $\btheta_{i,k}$, $\ybf_{i,k}$ corrects the update directions across all nodes.
For strongly convex functions, \citet{mai2016distributed} demonstrates that \eqref{eq:DGDRS} converges to the optimal solution with diminishing step-size $\eta_k$.
Furthermore, with constant step-size, variants of the DGD-RS algorithm \citep{xi2018linear, shen2020heavy} can achieve linear convergence.
However, for nonconvex settings or stochastic scenarios, theoretical guarantees for row-stochastic-type approaches remain unexplored.

In the presence of Byzantine attacks, regardless of whether the communication graph is undirected or directed, the optimization variables of Byzantine nodes may drift far away due to corrupted gradients.
Through communication with neighbors, Byzantine nodes can further manipulate the optimization variables of the normal nodes via \eqref{Eq_intro_DSGD0} or \eqref{eq:DGDRS}, even forcing these variables to become zero or diverge arbitrarily.

\section{Algorithmic Framework}\label{sec:frame}

The algorithmic framework for DRSGD-ByMI consists of three phases: warm-up, detection, and optimization.
The warm-up phase provides stable initial optimization variables; the detection phase serves to identify suspicious Byzantine neighbors at each node and remove the corresponding in-arcs; the optimization phase carries out the decentralized learning task over the pruned graph.

Concretely, we randomly split each local dataset $\Scal_i$ into a warm-up set $\Scal^{\mathrm{W}}_i$ and an identification set $\Scal^{\mathrm{D}}_i$ of size $n$.
In the warm-up phase, we utilize  $\Scal^{\mathrm{W}}_i$ to generate stochastic oracle $\gbf_{i}(\btheta^{\mathrm{W}}_{i,k}, \xi_{i,k})$, where $\btheta^{\mathrm{W}}_{i, k}$ stands for the local optimization variable specific to the warm-up phase, $\xi_{i,k}$ is a mini-batch sampled from $\Scal^{\mathrm{W}}_i$.
Then, we employ an existing Byzantine-robust decentralized stochastic optimization algorithm to obtain stable initial optimization variables.
In the detection phase, we use the samples in $\Scal^{\mathrm{D}}_i$ to construct test statistics for identifying each node’s potentially Byzantine neighbors.
During the optimization phase, we sample from $\Scal_i$ to obtain stochastic gradients, and then perform decentralized stochastic optimization over the pruned directed graph.
In Algorithm~\ref{alg:DRSGD-ByMI}, we present the complete procedure of DRSGD-ByMI.

\begin{algorithm}[t]
\small
\caption{DRSGD-ByMI}
\label{alg:DRSGD-ByMI}
\begin{algorithmic}[1]
\Require  $\btheta_0, \ybf_{0} \in \mathbb{R}^{d}$ and a mixing matrix ${\Wbf}$.
\State{Employ a decentralized Byzantine-robust stochastic optimization algorithm over sets $\{\Scal^{\mathrm{W}}_i\}_{i \in \Vcal}$ to solve problem \eqref{obj:dcentral_all} during $k_0$ iterations, and obtain $(\btheta^{\mathrm{W}}_{1, k_0}, \ldots, \btheta^{\mathrm{W}}_{m, k_0})$.}  \hfill $\triangleleft$ Warm-up phase
\For{all $i\in [m]$ in parallel}
\State Set $k \gets 0$.
Initialize $\ybf_{i, k}= \ybf_{0},$ and $\btheta_{i, k}=\btheta_{i, k_0}^{\mathrm{W}}.$  \hfill $\triangleleft$ Detection phase
\State Independently and randomly split the identification set $\Scal^{\mathrm{D}}_i$ into two equal-sized sub-batches, $\xi_{i, k}^{(1)}$ and $\xi_{i, k}^{(2)}$.
\State Estimate the gradients $\gbf_{i} (\btheta_{i, k}; \xi_{i, k}^{(1)})$ and $\gbf_{i} (\btheta_{i, k}; \xi_{i, k}^{(2)})$ via~\eqref{eq:gradients}.
\State Send $\gbf_{i} (\btheta_{i, k}; \xi_{i, k}^{(1)})$,  $\gbf_{i} (\btheta_{i, k}; \xi_{i, k}^{(2)})$ and ${\bm \theta}_{i, k}$ and receive $\gbf_{j} (\btheta_{j, k}; \xi_{j, k}^{(1)})$,  $\gbf_{j} (\btheta_{j, k}; \xi_{j, k}^{(2)})$, ${\bm \theta}_{j, k}$ from its neighbors $j \in \Ncal_{i}$.
 \State Construct ranking score $S_{i,j}$ by \eqref{eq:score}  for each $j \in \Ncal_{i}$.
\State Choose threshold $R_i$ by \eqref{eq:thold}.
\State Identify the set of Byzantine machines as $\widehat{\mathcal{B}}_{i}:= \{j\in \Ncal_{i} : S_{i,j} \geq R_i \},$ and cut off the in-arcs from $\widehat{\mathcal{B}}_{i}$.
\State Construct row-stochastic mixing matrix ${\Wbf}$ by \eqref{eq:reassign}.  \hfill $\triangleleft$ Optimization phase
\While{\textit{stopping criterion is not satisfied}}
\State Update $\ybf_{i, k+1}$ and ${\bm \theta}_{i, k+1}$ by \eqref{eq:eupdate} and \eqref{eq:local_dsgdrs}.
\State Set $k\gets k+1$.
\EndWhile
\EndFor
\State \Return ${\bm \Theta}_k:= [{\btheta}_{1,k}^{\top}; \ldots; {\btheta}_{m, k}^{\top}]$.
\end{algorithmic}
\end{algorithm}

\subsection{Warm-up Phase}

We aim to provide an approximate solution of Problem \eqref{obj:dcentral_all} with a low consensus error over the normal nodes.
Mathematically, we require the following condition, Condition~\ref{asp:warm_up} to hold.

\begin{condition}\label{asp:warm_up}
Let $\{\btheta_{1, k_0}^{\mathrm{W}}, \ldots, \btheta_{m, k_0}^{\mathrm{W}}\}$ be the optimization variables obtained by the decentralized Byzantine-robust stochastic algorithm in the warm-up phase, where $k_0$ is the number of warm-up iterations.
Then with probability at least $1-\tau_{w}$, the following inequality holds:
\begin{equation}\label{eq:warmup}
\sum\limits_{j \in \mathcal{G}}\|\btheta^{\mathrm{W}}_{j, k_0}- \bar{\btheta}^{\mathrm{W}}_{k_0} \|^2= \mathcal{O}\biggl(\frac{1}{k_0^{1-\delta}}\biggr).
\end{equation}
Here, $\bar{\btheta}^{\mathrm{W}}_{k_0} = \frac{1}{m_g}\sum_{j\in \mathcal{G}}\btheta^{\mathrm{W}}_{j, k_0}$ is the average among normal nodes, $\delta\in(0,1)$, and $\tau_{w}$ decreases to $0$, as $k_0$ tends to $+\infty$.
\end{condition}

In fact, many existing Byzantine-robust stochastic optimization algorithms can provide optimization variables satisfying this condition.
For instance,
 SCCLIP \citep{he2022byzantine} demonstrated that
\begin{equation}\label{eq:SCCLIP}
  \frac{1}{m_g}\sum\limits_{j \in \mathcal{G}} \mathbb{E}[\|\btheta^{\mathrm{W}}_{j, k_0}-\bar{\btheta}^{\mathrm{W}}_{k_0} \|^2] \leq \mathcal{O}\biggl(\frac{\zeta^2}{(k_0+1)\lambda_2({\Wbf})^2}\biggr),
\end{equation}
with step-size $\eta= \mathcal{O}(\frac{1}{\sqrt{k_0}})$, where $\zeta$ upper-bounds the discrepancy between the exact local gradient at each normal
node and the global gradient, $\lambda_2({\Wbf})$ is the second largest norm of eigenvalue of $\Wbf$.

In addition, IOS \cite[Theorem 2]{wu-byzantine-resilient-2023} showed that
\begin{equation}\label{eq:IOS}
 \frac{1}{m_g}\sum\limits_{j \in \mathcal{G}}\mathbb{E}[\|\btheta^{\mathrm{W}}_{j, k_0}- \bar{\btheta}^{\mathrm{W}}_{k_0}  \|^2]\leq \mathcal{O}\biggl(\frac{\sigma^2+\zeta^2}{k_0 }\biggr),
\end{equation}
with step-size $\eta= \mathcal{O}(\frac{1}{\sqrt{k_0}})$, where $\sigma$ is an upper bound on the standard deviation of stochastic gradients at each normal node, $\zeta$ is defined as in \eqref{eq:SCCLIP}.

For convenience, we unify \eqref{eq:SCCLIP} and \eqref{eq:IOS} into the following form:

\begin{equation*}
 \frac{1}{m_g}\sum\limits_{j \in \mathcal{G}}\mathbb{E}[\|\btheta^{\mathrm{W}}_{j, k_0}- \bar{\btheta}^{\mathrm{W}}_{k_0}  \|^2]\leq \mathcal{O}\biggl(\frac{1}{k_0 }\biggr).
\end{equation*}
By Markov inequality, for any $\varepsilon>0$,
\begin{equation*}
\mathbb{P}\biggl( \frac{1}{m_g}\sum\limits_{j \in \mathcal{G}}[\|\btheta^{\mathrm{W}}_{j, k_0}- \bar{\btheta}^{\mathrm{W}}_{k_0} \|^2]\geq \varepsilon \biggr)   \leq \frac{1}{\varepsilon}{ \sum\limits_{j \in \mathcal{G}}\mathbb{E}[\|\btheta^{\mathrm{W}}_{j, k_0}- \bar{\btheta}^{\mathrm{W}}_{k_0} \|^2]}.
\end{equation*}
Set $\varepsilon:= \frac{1}{k_0^{1-\delta}}$, then we obtain
\begin{equation*}
\mathbb{P}\biggl( \frac{1}{m_g}\sum\limits_{j \in \mathcal{G}}[\|\btheta^{\mathrm{W}}_{j, k_0}- \bar{\btheta}^{\mathrm{W}}_{k_0}  \|^2]\leq \frac{1}{k_0^{1-\delta}} \biggr)   \geq 1- \mathcal{O}\biggl(\frac{1}{k_0^{\delta}}\biggr).
\end{equation*}

Condition~\ref{asp:warm_up} ensures that the aggregation rules in the adopted warm-up algorithm produce robust and nearly consensual estimates before entering the detection phase.
It can also be inferred that the maximal deviation of any individual node $j$ from the consensus mean is bounded by $\mathcal{O}(1/k_0^{(1-\delta)})$ with probability at least $1-\tau_{w}$.

\subsection{Detection Phase}

To address the inherent steady-state error in decentralized robust aggregation, we explicitly identify and isolate Byzantine machines in a statistically principled manner.
For each normal machine $\mathcal{M}_{i}$ ($i\in \mathcal{G}$), we define the set of its normal neighbors and Byzantine neighbors as $\Gcal_i:= \Gcal \cap \Ncal_i$ and $\Bcal_{i}:= \Bcal \cap \Ncal_i$, respectively.
Identifying its Byzantine neighbors $\mathcal{B}_i$ can be formulated as a local multiple hypothesis testing problem:
\begin{equation*}
\mathbb{H}_{0j}: j \in \mathcal{G}_i \quad \text{versus}\quad  \mathbb{H}_{1j}: j \in \mathcal{B}_i, \quad \text{ for each } j\in \Ncal_{i}.
\end{equation*}
Let $\widehat{\mathcal{B}}_{i}$ be the set of rejected (identified) neighbors.
We evaluate the identification performance using the False Discovery Proportion (FDP) and the Sure-Detection Probability ($\mathrm{P}_a$):
\begin{equation}\label{eq:fdptpr}
\text{FDP}(\widehat{\mathcal{B}}_i) := \frac{\size{\widehat{\mathcal{B}}_i \cap \mathcal{G}}}{\size{\widehat{\mathcal{B}}_i}\vee 1}, \quad \mathrm{P}_{a}(\widehat{\mathcal{B}}_i) := \mathbb{P}(\mathcal{B}_i \subseteq \widehat{\mathcal{B}}_i).
\end{equation}
Our goal is to achieve $\mathrm{P}_a \approx 1$ (to eliminate Byzantine bias) while controlling FDP (to maintain connectivity among normal nodes), thereby enabling exact convergence.
To this end, we propose a decentralized Byzantine machine identification (ByMI) procedure: we construct test statistics $S_{i,j}$ via sample splitting and calibrate the decision threshold using a data-driven empirical null distribution rather than an asymptotic approximation, thereby improving finite-sample reliability in high-dimensional settings. For notational convenience, we denote $(\btheta_{1, 0}, \ldots, \btheta_{m, 0}):= (\btheta_{1, k_0}^{\mathrm{W}}, \ldots, \btheta_{m, k_0}^{\mathrm{W}})$.
The detailed steps are as follows:

\textbf{Splitting:} For each machine $\mathcal{M}_i$, we independently and randomly split $\Scal^{\mathrm{D}}_i$ into two sub-batches, denoted by $\xi_{i, 0}^{(1)}$ and $\xi_{i, 0}^{(2)}$, each of size $n/2$.
Notice that the selection of the two sub-batches is independent of $\Scal^{\mathrm{W}}_i$ and the generation of $\btheta_{i, 0}$.
We compute the stochastic gradients at $\btheta_{i, 0}$ associated with $\xi_{i, 0}^{(1)}$ and $\xi_{i, 0}^{(2)}$, respectively.
\begin{equation}\label{eq:gradients}
\gbf_{i} (\btheta_{i, 0}; \xi_{i, 0}^{(h)}) :=\frac{1}{\size{\xi_{i, 0}^{(h)}}} \sum\limits_{\sbf \in \xi_{i, 0}^{(h)}} \nabla \ell(\btheta_{i, 0}; \sbf),\,\quad h=1,2.
\end{equation}

\textbf{Communication:} Each machine $\Mcal_i$ broadcasts $\set{\gbf_{i} (\btheta_{i, 0}; \xi_{i, 0}^{(h)})}_{h=1,2}$ and $\btheta_{i, 0}$ to its neighbors $\Ncal_{i}$ and receives $\set{\gbf_{j} (\btheta_{j, 0}; \xi_{j, 0}^{(h)})}_{h=1,2}$ and $\btheta_{j, 0}$ from each $j \in \Ncal_{i}$.

\textbf{Scoring:} Each machine $\Mcal_i$ obtains a robust mean estimate $\widehat\gbf_{i}$ by $\{\gbf_{j} (\btheta_{j, 0}; \xi_{j, 0}^{(1)})\}_{j\in \Ncal_i}$ and constructs a test score
\begin{equation}\label{eq:score}
  S_{i,j}=(\gbf_{j} (\btheta_{j, 0}; \xi_{j, 0}^{(1)})-\widehat\gbf_{i})^{\top}{\bm \Omega}_{i}(\gbf_{j} (\btheta_{j, 0}; \xi_{j, 0}^{(2)})-\widehat\gbf_{i}), j\in \Ncal_i.
\end{equation}
Here, the robust mean estimator can be computed using various methods, such as median-type algorithms \citep{su2019securing, yin2018byzantine} or dimension-agnostic algorithms \citep{diakonikolas2017being, zhu2022robust}. ${\bm \Omega}_i$ serves as a rough scale estimator for normalization or a projection matrix in conjunction with projection-based robust mean estimators.
More importantly, ${\bm \Omega}_i$ can depend only on the random variables in $\xi_{i, 0}^{(1)}$.
For $j \in \mathcal{G}$, $\gbf_{j} (\btheta_{j, 0}; \xi_{j, 0}^{(2)})-\widehat\gbf_{i}$ is approximately normally distributed due to the central limit theorem and the independence between $\xi_{j, 0}^{(1)}$ and $\xi_{j, 0}^{(2)}$.
This fact indicates that $S_{i,j}$ enjoys the (asymptotic) symmetry with mean zero.
In contrast, for $j \in \mathcal{B}$, the mean of $S_{i,j}$ is a large positive value, which depends on the difference between distributions $\mathcal{P}$ and $\mathcal{Q}_{j}$.

\textbf{Thresholding:} Each machine $\Mcal_i$  chooses the threshold $R_i$ by

\begin{equation}\label{eq:thold}
R_i:= \inf \biggl\{r>0: \frac{\size{\{j \in \Ncal_i: S_{i,j} \leq-r\}}}{\size{\{j \in \Ncal_i: S_{i,j} \geq r\}} \vee 1} \leq \alpha\biggr\},
\end{equation}
where $\alpha$ is a prespecified target significance level, and
identify the set of Byzantine machines as $\widehat{\mathcal{B}}_{i}:= \{j\in \Ncal_{i} : S_{i,j} \geq R_i \}.$ If the set is empty, we set $R_i:= +\infty$.
Intuitively,
$\size{\{ j : S_{i, j} \le -r \}}$ provides an upper bound for $\size{\{ j : S_{i, j} \le -r,\, \mathcal{M}_j \in \mathcal{G} \}},$ which in turn serves as a good approximation of
$\size{\{ j : S_{i, j} \ge r,\, \mathcal{M}_j \in \mathcal{G} \}}$, i.e., the number of false discoveries, according to the symmetry of \(S_{i, j}\) for all normal machines.
Consequently, the fraction in \eqref{eq:thold} constitutes an estimate of the FDP.

Compared with traditional mean tests, the test statistic $S_{i,j}$ does not rely on the $p$-values from the asymptotic distribution.
Moreover, conditioned  on $\{\xi_{j, 0}^{(1)}\}_{j \in \Ncal_i}$,  $S_{i,j}$ can be viewed as a univariate projection of $\gbf_{j} (\btheta_{j, 0}; \xi_{j, 0}^{(2)})-\widehat\gbf_{i}$, and therefore exhibits asymptotic symmetry, regardless of the gradient dimension $d$.

\subsection{Optimization Phase}

To begin with, we convert the original undirected graph $\mathtt{G}=(\mathcal{V}, \mathcal{E})$ into a directed graph by treating each undirected edge as two directed arcs with opposite directions.
Based on the results in the detection phase, each normal machine $\mathcal{M}_i$ removes the in-arcs from $\widehat{\mathcal{B}}_{i}$, while each Byzantine machine $\mathcal{M}_j$ arbitrarily deletes some in-arcs from its neighbors $\Ncal_j$.
After this, we get a pruned directed graph denoted by $\mathtt{G}'=(\mathcal{V}, \mathcal{E}')$.

An ideal scenario is that each normal machine successfully identifies all of its Byzantine neighbors, that is, $\mathcal{B}_i= \widehat{\mathcal{B}}_{i},$ for all $i \in \mathcal{G}$.
In this case, $(\mathcal{G}, \mathcal{E}'|_{\mathcal{G}})= (\mathcal{G}, \mathcal{E}|_{\mathcal{G}})$ constitutes the largest strongly connected component of $\mathtt{G}'$ and has no in-arcs from other components, in accordance with Condition~\ref{asp:scc_graph} and the Byzantine ratio condition $\varrho< \frac{1}{2}$.
Executing decentralized optimization over $\mathtt{G}'$ is essentially equivalent to solving a collection of decentralized subproblems defined on its strongly connected components that have no external in-arcs.
Among these subproblems, the one associated with the largest strongly connected component $(\mathcal{G}, \mathcal{E}'|_{\mathcal{G}})$ takes the form
\begin{equation}\label{obj:dcentral_homo}
  \begin{aligned}
  \min_{\btheta_{i}\in \mathbb{R}^{d}, i\in \mathcal{G}} &~ \sum\limits_{i \in \mathcal{G}} f_i(\btheta_i)\\
  \text{s. t.} &~ \btheta_{i} = \btheta_{j}, (i,j)\in \mathcal{E}'|_{\mathcal{G}}.\\
  \end{aligned}
\end{equation}
This problem coincides with the decentralized empirical risk optimization problem~\eqref{obj:dcentral_all}.

To approach this ideal scenario as closely as possible, we require that the detection phase accurately identifies all Byzantine neighbors (i.e., achieves a $\mathrm{P}_a$ of $100\%$) with high probability, while misidentifying only a small number of non-Byzantine neighbors (i.e., maintaining a low FDP or false discovery number) with high probability.
In Section \ref{sec:gua}, we will analyze the control of FDP and $\mathrm{P}_a$ and discuss how these criteria influence the probability that $(\mathcal{G}, \mathcal{E}'|_{\mathcal{G}})$ forms a strongly connected component without external in-arcs.

We now turn our attention to optimizing over the directed graph $\mathtt{G}'$.
Our first step is to locally adjust the weights to ensure that the mixing matrix is row-stochastic.
Let ${\Wbf}^{\text{old}}$ denote the original mixing matrix, and define the updated weights ${\Wbf}(i,j)$ at each $\mathcal{M}_i$ as follows:
\begin{equation} \label{eq:reassign}
  {\Wbf}(i,j) = \begin{cases}
  0, & j\in \widehat{\mathcal{B}}_{i},  \\
  \frac{{\Wbf}(i,j) ^{\text{old}}}{\sum\limits_{j \in \Ncal_{i} \setminus \widehat{\mathcal{B}}_{i}}{\Wbf}(i,j) ^{\text{old}}}, & j\in  \Ncal_{i}^{\text{in}}\setminus \widehat{\mathcal{B}}_{i},\\
  0, & \text{ otherwise.}\\
  \end{cases}
\end{equation}

Motivated by the approach in \citep{mai2016distributed}, we introduce an auxiliary variable  ${\ybf}_{i}$ for each node $i$, and update ${\ybf}_{i}$ by
\begin{equation}\label{eq:eupdate}
\ybf_{i, k+1}  =  \sum\limits_{j \in \Ncal_{i}^{\text{in}} } {\Wbf}(i,j)\ybf_{j, k}
\end{equation}
to track the dominant left eigenvector of ${\Wbf}$.
Notice that the entries of $\mathbf{y}_{i,k}$ corresponding to the nodes in the largest strongly connected component also track the dominant left eigenvector $\vbf_1$ of the submatrix of $\Wbf$ associated with this connected component.
Then, we use $\frac{1}{[{\ybf}_{i, k}]_i}$ to scale the stochastic gradient descent step in the $k$-th local update,
 \begin{equation}\label{eq:local_dsgdrs}
  \begin{aligned}
    {\bm \theta}_{i, k+1} & =  \sum\limits_{j \in \Ncal_{i}^{\text{in}} }{\Wbf}(i,j){\bm \theta}_{j, k} -  \eta_{k} \frac{\gbf_i({\bm \theta}_{i, k}; \xi_{i, k})}{[\ybf_{i, k+1}]_{i}},
    \end{aligned}
  \end{equation}
where $\gbf_i(\btheta_{i, k}; \xi_{i, k}):= \frac{1}{\size{\xi_{i, k}}} \sum_{\sbf \in \xi_{i, k}} \nabla \ell(\btheta_{i, k}; \sbf)$ is a mini-batch stochastic gradient, defined in the same way as in Section \ref{sec:pre}.
Without loss of generality, we assume the first $m_g$ nodes are normal nodes, and define $\tilde{{\bm \theta}}_{k}:=[{\bm \theta}_{1,k}, \ldots, {\bm \theta}_{m_g,k}]\vbf_1$ as the weighted average over the largest strongly connected component.
The update formula for $\{\tilde{{\bm \theta}}_{k}\}$ is derived from \eqref{eq:local_dsgdrs},
\begin{equation*}
  \tilde{{\bm \theta}}_{k+1}  =   \tilde{{\bm \theta}}_{k} - \eta_k \sum\limits_{i=1}^{m_g} \frac{[\vbf_1]_i}{[\ybf_{i, k+1}]_{i}}  \gbf_i(\btheta_{i, k}; \xi_{i, k}).
\end{equation*}
Here, $[\ybf_{i, k+1}]_{i}$ serves to counterbalance the effect of
$[\vbf_1]_i$. Thus, the rescaled stochastic gradient is expected to move $\{\tilde{{\bm \theta}}_{k}\}$ toward a stationary point of the problem \eqref{obj:dcentral_homo}.

\section{Theoretical Guarantees}\label{sec:gua}

This section provides statistical guarantees and convergence analyses of our proposed method DRSGD-ByMI.
We summarize the relationship between the conditions and theories across the different phases in Figure~\ref{fig:framework2}. A key departure from prior literature is our analytical focus: the challenge shifts from bounding the bias of robust rules to ensuring the preservation of a strongly connected component among honest nodes. We first analyze the identification accuracy and then demonstrate how the pruned network maintains sufficient connectivity to support exact convergence.

\begin{figure}[htbp]
    \centering
   \includegraphics[width=\textwidth]{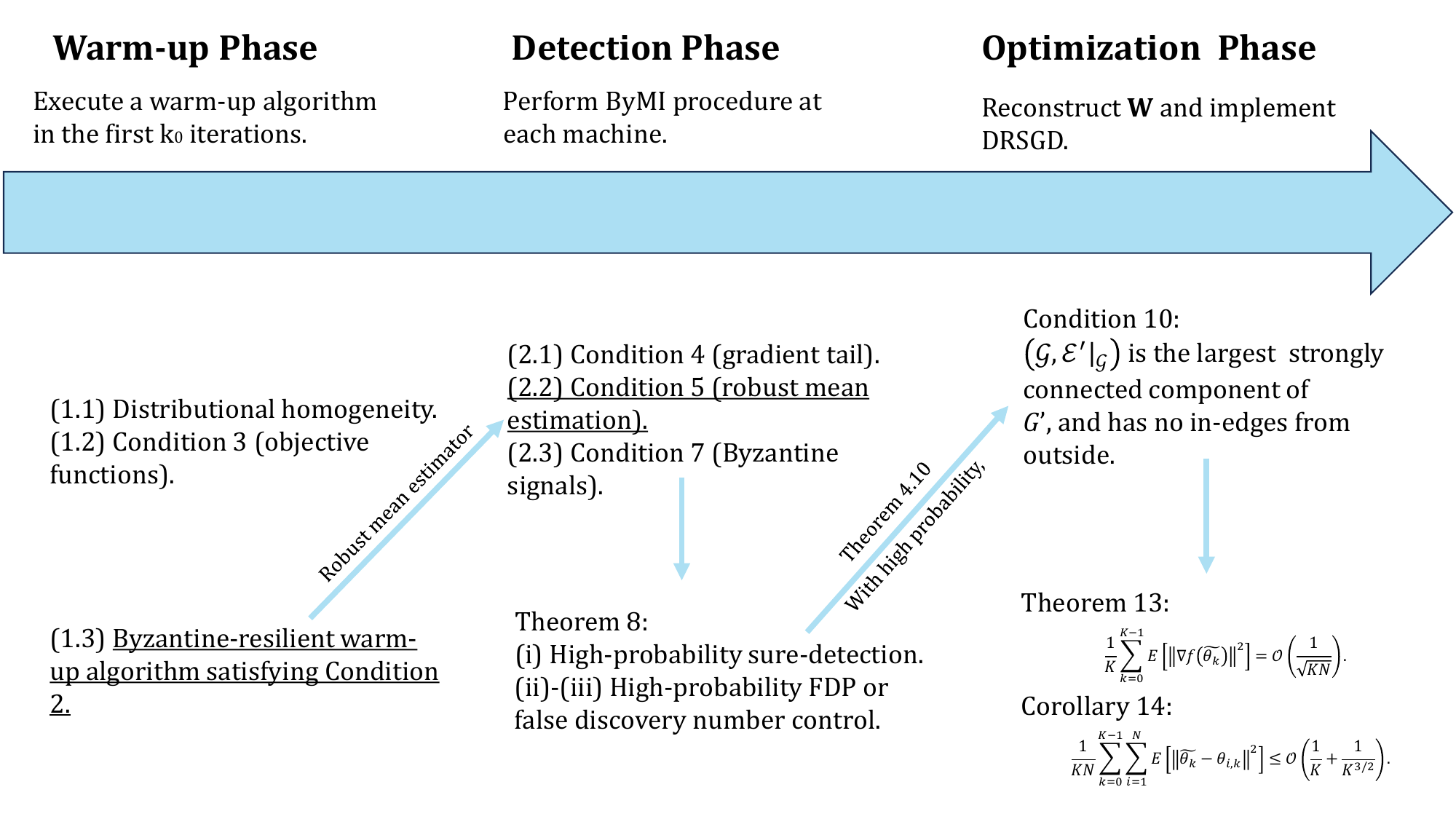}
   \caption{The condition–theorem framework of DRSGD-ByMI.}
  \label{fig:framework2}
\end{figure}

\subsection{Conditions on Objective Functions}

To begin with, let $(\Omega,\mathcal{F},\mathbb{P})$ denote the probability space associated with the per-iteration sampling randomness in the optimization phase. Note that it is independent of (i) the randomness for splitting $\{\Scal^{\mathrm{W}}_i\}_{i \in \Vcal}$ and $\{\Scal^{\mathrm{D}}_i\}_{i \in \Vcal}$, (ii) the randomness in the warm-up phase, and (iii) the randomness for splitting in the detection phase.  

\begin{condition}\label{asp:loss}
\begin{enumerate}[(i)]
\item (Bounded below) For any $\btheta\in \mathbb{R}^d$,  $f(\btheta)\geq f^\ast$.

\item (Lipschitz smooth) For any sample $\sbf \sim \mathcal{P}$, the loss function $\ell(\btheta; \sbf)$ is Lipschitz smooth with probability 1, i.e., there exists $L>0$, such that for any $\xbf, \ybf \in \mathbb{R}^d$,
\begin{equation*}
\norm{\nabla \ell(\xbf; \sbf) - \nabla \ell(\ybf; \sbf)} \leq L \norm{\xbf -\ybf}.
\end{equation*}

\item (Bounded variance) For any $\btheta\in \mathbb{R}^d$, there exists $\sigma_0>0$ such that
\begin{equation*}
\mathbb{E}_{\sbf \sim \mathcal{P}} \bigl[ \left\| \nabla \ell(\btheta; \sbf) - \mathbb{E}_{\sbf \sim \mathcal{P}} [\nabla \ell(\btheta; \sbf)] \right\|^2 \bigr] \leq \sigma_0^2.
\end{equation*}
Moreover, there exists $\sigma_i >0, i \in \Vcal$ such that
\begin{equation*}
\mathbb{E} \bigl[ \left\| \gbf_{i}(\btheta; \xi_{i, k})  -  \nabla f_i(\btheta) \right\|^2 | \mathcal{F}_{k} \bigr] \leq \sigma_i^2, \text{ for } k\in \mathbb{N}.
\end{equation*}

\end{enumerate}
\end{condition}

Condition~\ref{asp:loss}(ii) indicates that, during the detection phase, for any normal node $j \in \Gcal_i$ and $h=1,2$,
\[
\| \gbf_{j} (\btheta_{j, 0}; \xi_{j, 0}^{(h)}) - \gbf_{j}(\bar{\btheta}_{0}; \xi_{j, 0}^{(h)}) \|^2
\le L^2 \| \btheta_{j, 0} - \bar{\btheta}_{0} \|^2.
\]
Applying \eqref{eq:warmup}, we obtain
\[
\| \gbf_{j} (\btheta_{j, 0}; \xi_{j, 0}^{(h)}) - \gbf_{j}(\bar{\btheta}_{0}; \xi_{j, 0}^{(h)}) \|^2
\le L^2 \sup_{j \in \Gcal}\| \btheta_{j, 0} - \bar{\btheta}_{0} \|^2
= \mathcal{O}\!\biggl(\frac{L^2}{k_0^{1-\delta}}\biggr),
\]
with probability at least $1-\tau_w$.
This implies that as the parameters $\btheta_{j, 0}$ converge, the gradients computed at each normal node act as unbiased estimators of the gradient at $\bar{\btheta}_{0}$ with a diminishing bias term controlled by $\mathcal{O}(1/k_0^{(1-\delta)})$.
This concentration ensures that each $\mathcal{M}_i, i \in \mathcal{G}$ produces desired gradient estimates and helps the detection of the Byzantine neighbors.

Condition~\ref{asp:loss} gives that
\begin{equation*}
\mathbb{E}_{\sbf\sim \Pcal}\norm{\nabla f_i(\btheta)- \mathbb{E}_{\sbf \sim \mathcal{P}} [\nabla \ell(\btheta; \sbf)]}^2 \leq  \frac{\sigma_0^2}{N_i}.
\end{equation*}
Let $N_0:= \min\{N_i: i\in \Gcal\}$.
We have
\begin{equation*}
\mathbb{E}_{\sbf\sim \Pcal}\Bigl[\max_{i \in \Gcal} \norm{\nabla f_i(\btheta)- \mathbb{E}_{\sbf \sim \mathcal{P}} [\nabla \ell(\btheta; \sbf)]}^2\Bigr] \leq  \frac{\sigma_0^2}{N_0},
\end{equation*}
which implies that
\begin{equation}\label{eq:heter}
\mathbb{E}_{\sbf\sim \Pcal}\biggl[\max_{i \in \Gcal} \Bignorm{\nabla f_i(\btheta)- \frac{1}{m_g} \sum\limits_{i\in \Gcal} \nabla f_i(\btheta)}^2\biggr] \leq  \frac{4\sigma_0^2}{N_0}.
\end{equation}
Inequality~\eqref{eq:heter} provides an upper bound on the gradient heterogeneity among the normal nodes, measured by the maximum deviation of each local gradient from the average normal gradient.

\subsection{Statistical Guarantees for Detection Phase}\label{sec:gua_detect}

In this section, we establish the statistical guarantees for the proposed decentralized ByMI procedure.
We focus our analysis on a representative normal machine $\mathcal{M}_i$ ($i \in \Gcal$) and its neighbors $\Ncal_i$.

Crucially, our sample-splitting strategy ensures independence: $\widehat{\gbf}_{i}$ and $\bm{\Omega}_i$ are estimated using samples from $\xi_{i, 0}^{(1)}$, while $\gbf_{i} (\btheta_{i, 0}; \xi_{i, 0}^{(2)})$ is computed from $\xi_{i, 0}^{(2)}$.
Consequently, $\gbf_{i} (\btheta_{i, 0}; \xi_{i, 0}^{(2)})$ is independent of both $\widehat{\gbf}_{i}$ and $\bm{\Omega}_i$.
We first introduce the moment conditions for the local gradient estimators to ensure the symmetry of scores $\set{S_{i, j}}_{j \in \Gcal_i}$.
For notational convenience, we denote $\mathbf{g}_j^\ast := \mathbb{E}_{\sbf \sim \mathcal{P}} [\nabla \ell(\btheta_{j, 0}; \sbf)]$ for the rest of this section.

\begin{condition}[Gradient Tail] \label{cond:grad_tail}
The sample gradient $\nabla \ell(\btheta_{j, 0}; \sbf)$ with $\sbf \sim \mathcal{P}$ has bounded $q$-th centered moments for some $q > 2$.
Furthermore, it satisfies the $L_q$-$L_2$-norm equivalence condition with parameter $\gamma_q$:
  \begin{equation*}
    \max_{\vbf \in \mathbb{S}^{d - 1}} \frac{(\mathbb{E}_{\sbf \sim \mathcal{P}}[|\vbf^\top (\nabla \ell(\btheta_{j, 0}; \sbf) - \gbf_j^\ast)|^q])^{\frac{1}{q}}}{(\mathbb{E}_{\sbf \sim \mathcal{P}}[|\vbf^\top (\nabla \ell(\btheta_{j, 0}; \sbf) - \gbf_j^\ast)|^2])^{\frac{1}{2}}} \le \gamma_q.
\end{equation*}
\end{condition}

\begin{condition}[Robust Mean Estimation] \label{cond:robust_mean_est}
We assume that with probability at least $1 - \tau_{\gbf}$, the robust mean estimator satisfies
  \begin{equation*}
    \sup_{j \in \Gcal_i} \|\widehat{\gbf}_{j} - \gbf_j^\ast\| \le \delta_{\gbf},
\end{equation*}
  where $\delta_{\gbf}$ and $\tau_{\gbf}$ diminish to zero as $n, m, k_0 \to \infty$.
\end{condition}

Under Condition~\ref{cond:grad_tail} and Condition~\ref{asp:warm_up} in the warm-up phase, Condition~\ref{cond:robust_mean_est} holds for popular robust mean estimators.
This is formalized in the following proposition.

\begin{prop}\label{prop:robust_est}
Suppose the warm-up phase satisfies Condition~\ref{asp:warm_up}.
By applying a robust mean estimator (e.g., coordinate-wise median or Filtering \citep{diakonikolas2017being}) to produce gradient estimates $\{\widehat{\gbf}_j\}$, we have with probability at least $1-\tau_{\gbf}$,
$$
 \sup_{j \in \Gcal_i} \|\widehat{\gbf}_j - \gbf_j^\ast\| = \mathcal{O}(\delta_{\gbf}),
$$
where $\delta_{\gbf}$ and $\tau_{\gbf}$ diminish to zero as $n, m, k_0 \to \infty$.
\end{prop}

The justification of Proposition \ref{prop:robust_est} is provided in the Appendix \ref{appendix:Prop4.4}.
Next, we specify the signal strength required to distinguish Byzantine machines from normal ones.

\begin{condition}[Byzantine Signals]\label{cond:signals}
Assume that with probability at least $1 - \tau_{\gbf}$, the following result holds for any $j \in \mathcal{B}_i$:
    $$ \sup_{j' \in \Gcal} \|\gbf_{j'} (\btheta_{j', k_0}; \xi_{j', 0}^{(1)}) - \gbf_{j'}^\ast\|_{\bm{\Omega}} + \delta_{\gbf} \lesssim \|\gbf_{j}^\ast - \gbf_{i}^\ast\|_{\bm{\Omega}}.$$
                                                    \end{condition}

We now present the main theorem concerning the finite-sample FDP control and sure-detection capability.

\begin{theorem}[Sure detection and FDP control]\label{thm:fdp_control}
    Suppose Conditions \ref{cond:grad_tail}, \ref{cond:robust_mean_est}, and \ref{cond:signals} hold. Let $\kappa := \min(1,q - 2)$.

    (i) If either $n \ge \size{\Gcal_i}^{\frac{2}{\kappa - 2 c}}$ or $\size{\Bcal_i} \ge c \max\{1, \size{\Gcal_i} n^{-\frac{\kappa}{2}}\} \log n$ for some constant $0 < c < \frac{\kappa}{2}$, with probability at least $1 - 2\tau_{\gbf} - (C_{\kappa} + 1) \size{\Bcal_i} n^{-\frac{\kappa}{2}} - n^{-c}$, the sure-detection property holds:
    $\mathcal{B}_i \subseteq \widehat{\mathcal{B}}_i.$

    (ii) If $ n \gtrsim (\size{\Gcal_i} / \size{\Bcal_i})^{\frac{2}{\kappa}}$, with probability at least $1 - 2\tau_{\gbf} - \exp(- c \log n) - \exp(-C\size{\Bcal_i}) - 3 \exp(- C (\alpha \size{\Bcal_i})^{\frac{1}{3}})$,
    the \text{FDP} is controlled by:
    \begin{equation*}
      \text{FDP}(\widehat{\mathcal{B}}_i) \le \alpha \Bigl\{1 + \mathcal{O}\Bigl((\alpha \size{\Bcal_i})^{-1/3} + n^{-\frac{(1 - a^2) \kappa}{2}} (\log n)^{\frac{1}{2}} + \size{\Gcal_i} n^{-\frac{a^2 \kappa}{2}} + \delta_{\gbf}\Bigr) \Bigr\}:= \alpha H_{i,n},
    \end{equation*}
    where $a$ is any fixed constant in $(0, 1)$ and $H_{i,n} = 1 + o(1)$.

    (iii) If $\size{\Bcal_i} < \frac{1 - \alpha}{2\alpha} m_f$ with some sufficiently large $m_f$, with probability at least $1 - 2\tau_{\gbf} - 3 \exp(-C m_f^{\frac{1}{3}})$, the false discovery number is controlled by:
    $$\sum\limits_{j \in \Gcal_i} \id\{S_{i, j} \ge R_i\} \le 2 m_f.$$
    When $m_f = (\log m)^3$ and $m$ is sufficiently large, we have with probability at least $1 - 2\tau_{\gbf} - \exp(- C \log m)$,
    $$\sum\limits_{j \in \Gcal_i} \id\{S_{i, j} \ge R_i\} \le 2 (\log m)^3.$$
\end{theorem}

\begin{rmk}
In the probability bound of Theorem~\ref{thm:fdp_control} (ii), we require $\size{\Bcal_i}$ to be sufficiently large (e.g., $\size{\Bcal_i} \ge (\log m)^3$) to ensure that the FDP control holds uniformly for all normal nodes $i \in \Gcal$. By combining this with Theorem~\ref{thm:fdp_control} (iii), we demonstrate that with high probability, the false pruning of normal edges is strictly bounded. Specifically, when $\size{\mathcal{B}_i}\geq \frac{1-\alpha}{2 \alpha} (\log m)^3$, the false discovery proportion $\text{FDP}(\widehat{\mathcal{B}}_i)$ is well-controlled at level $\alpha(1+o(1))$. Conversely, when $\size{\mathcal{B}_i}< \frac{1-\alpha}{2 \alpha} (\log m)^3$, the absolute number of false discoveries is explicitly bounded by $2(\log m)^3$.
\end{rmk}

\subsection{Convergence Analysis for Optimization Phase}
\revise{
In this section, we discuss how the statistical guarantees in Theorem~\ref{thm:fdp_control} affect the strong connectivity of the sub-graph $(\mathcal{G}, \mathcal{E}'|_{\mathcal{G}})$ within the pruned directed graph $\mathtt{G}'$ with edges $\mathcal{E}'$.
First, we demonstrate that Theorem~\ref{thm:fdp_control} ensures the following strong connectivity property:
\begin{condition}\label{asp:sure_detection2}
\mbox{}
$(\mathcal{G}, \mathcal{E}'|_{\mathcal{G}})$ is the largest strongly connected component of $\mathtt{G}'$, and has no in-arcs from outside.
\end{condition}
Building upon this condition, we subsequently provide a non-asymptotic convergence guarantee for the DRSGD algorithm.
}

To begin with, we introduce a result about the connectivity of Erdős–Rényi random graph
 $\mathtt{G}(m,p)$ based on Theorem 4.1 of \citet{frieze2015introduction}.
The Erdős–Rényi random graph $\mathtt{G}(m,p)$ \citep{erdos1959graph} is generated by independently placing undirected edges between $m$ nodes with probability $p$, which is widely used in the topology of decentralized networks.

\begin{prop}\label{prop:connected}
Let $c(m, p):= mp-\log(m)$, $X_1$ be the number of isolated nodes in $\mathtt{G}(m,p)$.
Then it holds that
\begin{enumerate}[(i)]
\item  $\mathbb{E}(X_{1})\leq \begin{cases}
(1+o(1)) \exp(-c(m,p)), & \text{ when } \lim_{m\to \infty}\frac{c(m, p)}{\log m}< + \infty;\\
(1+p)\exp(-c(m,p)), &   \text{ when } \lim_{m\to \infty}\frac{c(m, p)}{\log m}= + \infty.
\end{cases}$
\item  \begin{equation*}
\begin{aligned}
 & \mathbb{P}(\mathtt{G}(m,p) \text{ is connected})\\
 \geq  & \begin{cases}
1- (1+o(1)) \exp(-c(m,p)) -  \mathcal{O}(m^{o(1) - 1}), & \text{ when } \lim\limits_{m\to \infty}\frac{c(m, p)}{\log m}< + \infty;\\
 1- (1+p) \exp(-c(m,p)) -  \mathcal{O}(m^{- \frac{c(m,p)}{\log m}}), &   \text{ when } \lim\limits_{m\to \infty}\frac{c(m, p)}{\log m}= + \infty.
\end{cases}
\end{aligned}
\end{equation*}
\end{enumerate}
\end{prop}

Next, we establish a finite-sample guarantee for the strong connectivity of $(\mathcal{G}, \mathcal{E}'|_{\mathcal{G}})$.
The proof of Theorem~\ref{thm:connectivity} can be found in Appendix \ref{sec:app2}.

\begin{theorem}\label{thm:connectivity}
        Let $(\mathcal{G}, \mathcal{E}|_{\mathcal{G}})$ be generated as an Erdős–Rényi random graph $\mathtt{G}(m,p)$. Suppose that, with a proper choice of $n$, the following conditions hold: (i) $\mathcal{B}_i \subseteq \widehat{\mathcal{B}}_i$ for any $i \in \Gcal$; (ii) when $\size{\mathcal{B}_i}\geq \frac{1-\alpha}{2 \alpha} (\log m)^3$, $\text{FDP}(\widehat{\Bcal}_i) \le \alpha H_{i,n}$, and when $\size{\mathcal{B}_i} < \frac{1-\alpha}{2 \alpha} (\log m)^3$, $\sum_{j \in \Gcal_i} \id\{S_{i, j} \ge R_i\} \le 2 (\log m)^3$.
    For any $\delta < 1-\beta_0$ with  $\beta_0:= \max\{\frac{4 (\log m)^3}{(m-1)p}, \alpha \max_i\{H_{i,n}\}\}$, denote
    $$c(m,p, \delta):= m[p(1-\beta_0-\delta)] -\log m.$$
Then the sub-graph $(\mathcal{G}, \mathcal{E}'|_{\mathcal{G}})$ forms a strongly connected component of the pruned graph $\mathtt{G}'$ with probability at least
\begin{equation}\label{eq-prob:connected}
\begin{aligned}
\begin{cases}
& 1- 3 \exp(-c(m,p, \delta))
-2m \exp(-\frac{(m-1)p}{8}) \\
&~~ -2m\exp(-\,\frac12 c_0 \delta^2 (m-1)p) - \mathcal{O}\!\bigl(m^{-1+o(1)}\bigr), \text{ when } \lim\limits_{m\to \infty}\frac{c(m,p, \delta)}{\log m}< + \infty;\\
& 1 - (2+p(1-\beta_0-\delta)) \exp(-c(m,p, \delta)) -2m\exp(-\frac{(m-1)p}{8})
\\
&~~ -2m\exp(-\,\frac12 c_0 \delta^2 (m-1)p) -  \mathcal{O}(m^{- \frac{c(m,p, \delta)}{\log m}}), \text{ when } \lim\limits_{m\to \infty}\frac{c(m,p, \delta)}{\log m}= + \infty.
 \end{cases}
\end{aligned}
\end{equation}
In particular, whenever $\lim_{m \to \infty} c(m,p, \delta)= + \infty,$ it holds that
\begin{equation*}
\lim\limits_{n, m\to\infty} \mathbb{P}( (\mathcal{G}, \mathcal{E}'|_{\mathcal{G}})   \text{ is a strongly connected component of } \mathtt{G}') =1.
\end{equation*}
\end{theorem}

Theorem~\ref{thm:connectivity} indicates that controlling the FDP at a low significance level $\alpha$, or ensuring that the false discovery number remains small, yields a larger value of $c(m,p,\delta)$.
This, in turn, drives the probability that $(\mathcal{G}, \mathcal{E}'|_{\mathcal{G}})$ forms a strongly connected component of the pruned graph $\mathtt{G}'$ toward one.

Combining Theorem~\ref{thm:fdp_control} and the fact $\varrho< \frac{1}{2}$, Theorem~\ref{thm:connectivity} further implies that Condition~\ref{asp:sure_detection2} holds  with high probability.
\revise{As a direct structural consequence of Condition~\ref{asp:sure_detection2}, we can rigorously characterize the mixing matrix corresponding to $(\mathcal{G}, \mathcal{E}'|_{\mathcal{G}})$.
Specifically, the overall row-stochastic mixing matrix ${\Wbf}$ becomes reducible, with its largest irreducible sub-block $\Abf \in \mathbb{R}^{m_g \times m_g}$ conforming exactly to the topology of $(\mathcal{G}, \mathcal{E}'|_{\mathcal{G}})$.
According to \citet[Perron-Frobenius Theorem]{pillai2005perron}, this property mathematically ensures that $\rho(\Abf) = 1$ and there exists a strictly positive left eigenvector ${\vbf}_1^\top$ (with
${\vbf}_1^\top \bm{1} = 1$) and constants $c > 0$, $\rho \in (|\lambda_2(\Abf)|, 1)$
such that $\|\Abf^k - \bm{1}{\vbf}_1^\top\|_2 \leq c\rho^k$. Without loss of generality, we relabel the nodes so that the first $m_g$ nodes belong to the largest strongly connected component $(\mathcal{G}, \mathcal{E}'|_{\mathcal{G}})$ associated with the irreducible sub-block $\Abf$.}

Therefore, under Condition~\ref{asp:sure_detection2}, applying the DRSGD algorithm on the entire pruned graph $\mathtt{G}'$ can be used to find the solution of the subproblem over the strongly connected sub-graph $(\mathcal{G}, \mathcal{E}'|_{\mathcal{G}})$.
Now, we establish the non-asymptotic convergence rate of DRSGD in solving problem \eqref{obj:dcentral_all} in the following theorem.

\begin{theorem}\label{thm:main0}
Suppose Condition~\ref{asp:loss} and Condition~\ref{asp:sure_detection2} hold.
Let
\begin{equation*}
\begin{aligned}
& \sigma := \max\{\sigma_i : i\in \Gcal\}, \quad \zeta := \frac{4\sigma_0^2}{N_0}, \quad D_1:= \frac{(6w^2c\zeta^2 + 8 m_g L^2 \sigma^2 w^2 c^2 + 72 m_g L^2 \zeta^2 w^2 c^2)^2}{(1-\rho)^2m_g(f(\tilde{\btheta}_{0}) -f^\ast)^2}, \\
& \quad D_2:= \biggl(\frac{54 L^2 \zeta^2 w^4 c^3+ 216L^4w^6c^5\sigma^2+ 6 L^2  \sigma^2 w^4 c^3 + 1944  L^4  \zeta^2 w^6 c^5}{\sqrt{m_g}(1-\rho)^3(f(\tilde{\btheta}_{0}) -f^\ast)}\biggr)^{2/3},\\
& D_3:= \frac{432 L^2w^4c^3}{m_g (1-\rho)(1-\rho^2)}, \quad D_4:= \frac{288 L^2 w^2 c^2}{(1-\rho)^2}.\\
\end{aligned}
\end{equation*}
Set $\Ybf_{0}:= [ \ybf_{1,0}^{\top}; \ldots; \ybf_{m_g,0}^{\top} ] = \Abf^{t_0}$, for some $t_0\geq \lceil \frac{\log(m_g /48c^2w^4\|{\vbf}_1\|^2)}{2\log(\rho)}  \rceil$.
When the number of iterations of DRSGD satisfies
\begin{equation*}
 K \geq  \max\{{m_g^{-1}}, 4m_g L^2, D_1, D_2, D_3, D_4\},
\end{equation*}
and step-size $\eta_k := \sqrt{\frac{1}{m_g K}}$ is used, then it follows that
\begin{equation*}
\frac{1}{K}\sum\limits_{k=0}^{K-1}  \mathbb{E}[\|\nabla f(\tilde{\btheta}_k)\|^2] \leq  \frac{20(f(\tilde{\bm \theta}_0) -f^\ast)+ 4 w^2 \|{\vbf}_1\|^2 \sigma^2  L}{\sqrt{Km_g}}.
\end{equation*}
where $w:= \sup_{k}\|{\mathrm{Diag}(\Abf^{k})}^{-1}\|_2 < +\infty$.
\end{theorem}

As a direct consequence of Theorem~\ref{thm:main0}, we obtain the following convergence bound about the consensus error.

\begin{coro}\label{coro:consensus}
Suppose Condition~\ref{asp:loss} and Condition~\ref{asp:sure_detection2} hold.
It holds that
\begin{equation*}
  \frac{1}{Km_g} \sum\limits_{k=0}^{K-1} \sum\limits_{i=1}^{m_g} \mathbb{E}[\|\tilde{\btheta}_{k}- {\btheta}_{i,k}\|^2] \leq \mathcal{O}\biggl(\frac{\sigma^2 + \zeta^2}{K} +\frac{1}{K^{3/2}}\biggr).
\end{equation*}

\end{coro}

\revise{In Theorem~\ref{thm:main0}, the convergence rate $\mathcal{O}(1/\sqrt{Km_g})$ of DRSGD matches the optimal convergence rate for
decentralized nonconvex stochastic first-order methods with doubly-stochastic mixing matrices in the Byzantine-free setting \citep{yuan2022revisiting}, up to universal constant factors.
Moreover, DRSGD achieves a linear speedup with respect to the number of normal agents $m_g$: compared to the centralized lower bound of $\mathcal{O}(1/\sqrt{K})$ for a single machine \citep{arjevani2023lower}, our decentralized network attains an effective rate of $\mathcal{O}(1/\sqrt{Km_g})$.
As a by-product, Corollary~\ref{coro:consensus} implies that as $K$ grows, local optimization variables ${\btheta}_{i,k}$ asymptotically reach consensus on the ${\vbf}_1$-weighted average $\tilde{\btheta}_{k}$ and hence converge to a stationary point of problem \eqref{obj:dcentral_homo} at a rate of $\mathcal{O}({1}/{\sqrt{Km_g}})$.}

\revise{In summary, the statistical guarantees of DRSGD-ByMI establish a rigorous link between Byzantine identification and exact convergence. With high probability, the ByMI procedure achieves sure-detection while controlling the false discovery proportion (FDP) at the prescribed significance level $\alpha$, or alternatively, keeping the number of falsely removed normal edges small. The parameter $\alpha$ plays a pivotal role in both the connectivity of the pruned network and the exact convergence of DRSGD. As $\alpha$ decreases from 1, the FDP can be controlled at a lower level, which leads to a larger connectivity parameter $c(m,p,\delta)$ and thereby increases the probability of preserving sufficient connectivity in the pruned network. However, by Theorem~\ref{thm:connectivity}, once $\alpha$ falls below a certain threshold, $c(m,p,\delta)$ no longer increases, whereas the probability of  FDP control continues to decrease.  Therefore, a sufficiently small  and empirically well-chosen $\alpha$ can effectively improve the probability that the pruned subgraph $(\mathcal{G}, \mathcal{E}'|_{\mathcal{G}})$ forms a strongly connected component among the normal nodes. Building upon this recovered connectivity, DRSGD on the pruned network achieves, with high probability, the order-optimal convergence rate of $\mathcal{O}(1/\sqrt{Km_g})$ in the nonconvex setting, ultimately matching the performance of Byzantine-free decentralized stochastic first-order methods.}

\section{Numerical Experiments}\label{sec:numeric}

In this section, we evaluate the numerical performance of DRSGD-ByMI on synthetic-data and real-data
applications.
We implement DRSGD-ByMI in conjunction with three warm-up algorithms, UBAR \citep{guo2021byzantine}, IOS \citep{wu-byzantine-resilient-2023} and BALANCE \citep{fang-byzantine-robust-2024}, and refer  to them as UBAR-DRSGD-ByMI, IOS-DRSGD-ByMI
and  BALANCE-DRSGD-ByMI.
The robust mean estimator is chosen as the Filtering estimator proposed in \citep{diakonikolas2017being, lai2016agnostic}.
In synthetic-data experiments, we choose $\Omega_i := \Ibf_d$, and test their performance under varying Byzantine ratios and contamination intensities.
In real-data experiments, we choose $\Omega_i$ as the projection matrix onto the principal component directions of neighbors' stochastic gradients that account for $95\%$ of the total variance, and compare DRSGD-ByMI variants with existing Byzantine-robust decentralized stochastic gradient algorithms in decentralized learning tasks, and further examine their numerical stability under different Byzantine ratios.

All algorithms are executed five times with varying random seeds.
The numerical experiments presented here are run on a platform with two Intel(R) Xeon(R) Gold 5317 CPUs ($@$ 3.00GHz and 512GB RAM) and NVIDIA GeForce RTX 4090 GPUs under Ubuntu 20.04.
We implement all algorithms with Python 3.8 and PyTorch 1.13.1.

\textbf{Decentralized network.}
The topology of the initial decentralized network is configured as an undirected Erdős–Rényi \citep{erdos1959graph} random graph, where each pair of nodes is connected by an edge with a probability of $p = 0.5$, and the mixing matrix is chosen as the Metropolis weight matrix \citep{xiao2006distributed}.

\textbf{Performance measures.}
We employ the averaged FDP and $\mathrm{P}_a$ as our performance measures for the detection phase, i.e.
\begin{equation}
\text{FDP} = \frac{1}{m_g} \sum\limits_{i\in \mathcal{G}} \text{FDP}(\widehat{\mathcal{B}}_i), \text{ and } \mathrm{P}_a
 = \frac{1}{m_g}\sum\limits_{i\in \mathcal{G}}\mathrm{P}_a(\widehat{\mathcal{B}}_i).
\end{equation}
In the optimization phase, we use the optimality gap $f(\bar{\btheta}_{K})- f(\btheta^\ast)$ as the performance measure in synthetic-data experiments, where $\tilde{\btheta}_K$ denotes the weighted average over the largest strongly connected component of the pruned graph. For the real-data experiments, since the ground truth $\btheta^\ast$ is unknown, we use (i) the norm of the global gradient evaluated at the averaged model over the normal nodes  $\norm{\nabla f(\bar{\btheta}_{K})}$ (or nodes in the largest strongly connected component of the pruned graph) and (ii) the test accuracy to characterize the numerical performance in real-data experiments.
The test accuracy can be divided into three types:
\begin{itemize}
\item ``Acc (all)'': test accuracy of the averaged model over all nodes;
\item ``Acc (normal)'': test accuracy of the averaged model over all normal nodes;
\item ``Acc (scc)'': test accuracy of the averaged model over the largest strongly connected component of the pruned graph.
\end{itemize}
\revise{For our proposed DRSGD-ByMI, we report ``Acc (scc)'', which represents the test accuracy of the averaged model over the largest strongly connected component of the pruned graph.
This metric naturally reflects our framework's true output, as potentially  malicious nodes are explicitly isolated from the optimization process after the detection phase.
In contrast, existing Byzantine-robust baselines typically lack a rigorous statistical pruning procedure and continue to aggregate information from suspicious nodes.
To ensure a fair and comprehensive comparison, we evaluate these baselines using both ``Acc (all)'' (averaged over all nodes) and ``Acc (normal)'' (averaged over the ground-truth normal nodes).}

\subsection{Synthetic-Data Experiments}

\textbf{Data generating process.}
We consider  the linear model,
\begin{equation*}
y_i = \mathbf{x}_i^\top \boldsymbol{\theta} + \varepsilon_i, \quad  i\in [m],
\end{equation*}
    where the input $\mathbf{x}_i \sim \Ncal_d(0, I_d)$, the noise $\varepsilon_i \sim \Ncal(0,1)$ and the ground-truth parameter is given by
    $\boldsymbol{\theta}^\ast := ({\bm 1}_s, 0, \cdots, 0)^\top$ with $s = \lfloor 0.1d \rfloor$.
    Moreover, the network contains $m=150$ nodes, where each node maintains a  local dataset $\{(\xbf_{i,j}, y_{i,j})\}_{j=1}^{N}$ with size $N_i \equiv N=200$ (for $d=80$) or $N_i \equiv N=100$ (for $d=30$), and holds a local risk function
    \begin{equation*}
        f_{i}(\btheta) = \mathbb{E}_{\xbf_{i}\sim \Ncal_d(0, I_d)} \frac{1}{2}(y_{i}-\mathbf{x}_i^\top \boldsymbol{\theta})^2  = \frac{1}{2N}\sum\limits_{j=1}^{N} (y_{i}-\mathbf{x}_i^\top \boldsymbol{\theta})^2.
    \end{equation*}
\revise{To balance the probability of FDP control and the connectivity
parameter $c(m, p, \delta)$, we set the target significance
level $\alpha = 0.2$ in our experiments.} We configure the total number of iterations to be $K = 3000$, and set the detection time as $k_0 = 0.1 K$.
The size of the identification set is $n = \frac{N}{2}$.
The default Byzantine ratio $\varrho$ is $0.2$.
In the following, we consider two kinds of Byzantine attacks:
    \begin{itemize}
        \item Scenario A (Parameter attack): the model is corrupted so that the parameter at Byzantine machines is
        $\boldsymbol{\theta}_c = (\mu_c{\bm 1}_{s}, 0, \cdots, 0)^\top$ with $s = \lfloor s_r d \rfloor$, $\mu_c = 5$  by default.
     \item Scenario B (Data attack): the data are contaminated so that the covariates on Byzantine machines are replaced by
        $\tilde{\mathbf{x}}_i = 0.8\mathbf{x}_i + 3\mathbf{v}_d$, where
$\mathbf{v}_d \in \mathbb{R}^d$ is a normalized vector with $d$ independent entries drawn from $U(0,1)$, and the responses $y_i$ are shifted by a constant bias $c = 1$.
    \end{itemize}

\textbf{Numerical performance under different contamination intensities and Byzantine ratios.} Figure~\ref{fig:sr-dim30-80} reports box-and-whisker plots  of FDP, $\mathrm{P}_a$ and optimality gap against $s_r$ under Scenario A.
On the axis-$s_r$, the farther $s_r$ is from $0.1$, the higher the  contamination intensity.
It can be observed that all methods achieve  FDPs less than $40\%$ in all intensities.
When $s_r$
 is close to $0.1$, the contaminated data approximately resemble the uncontaminated data, leading to lower $\mathrm{P}_a$ and higher optimality gap, which aligns with theoretical expectations since the attack signal is weak.
As $s_r$ moves further away from $0.1$, the signal of Byzantine attacks gradually strengthens. Figures~\ref{fig:sr-dim30-80}(b), (c), (e), and (f) show that $\mathrm{P}_a$ is approximately $100\%$ and the optimality gap is less than $10^{-4}$, in accordance with Theorem~\ref{thm:fdp_control}.

\begin{figure}[htbp]
  \centering
    \begin{subfigure}[b]{0.32\textwidth}
    \centering
    \includegraphics[width=\textwidth]{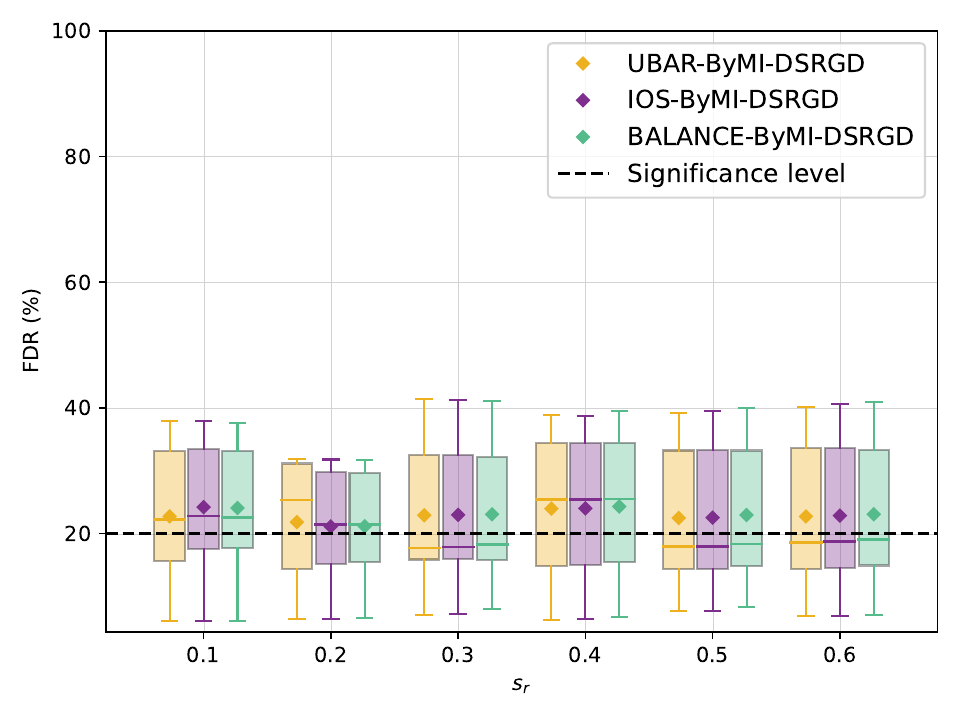}
    \caption{FDP ($d=30$)}
  \end{subfigure}
  \begin{subfigure}[b]{0.32\textwidth}
    \centering
    \includegraphics[width=\textwidth]{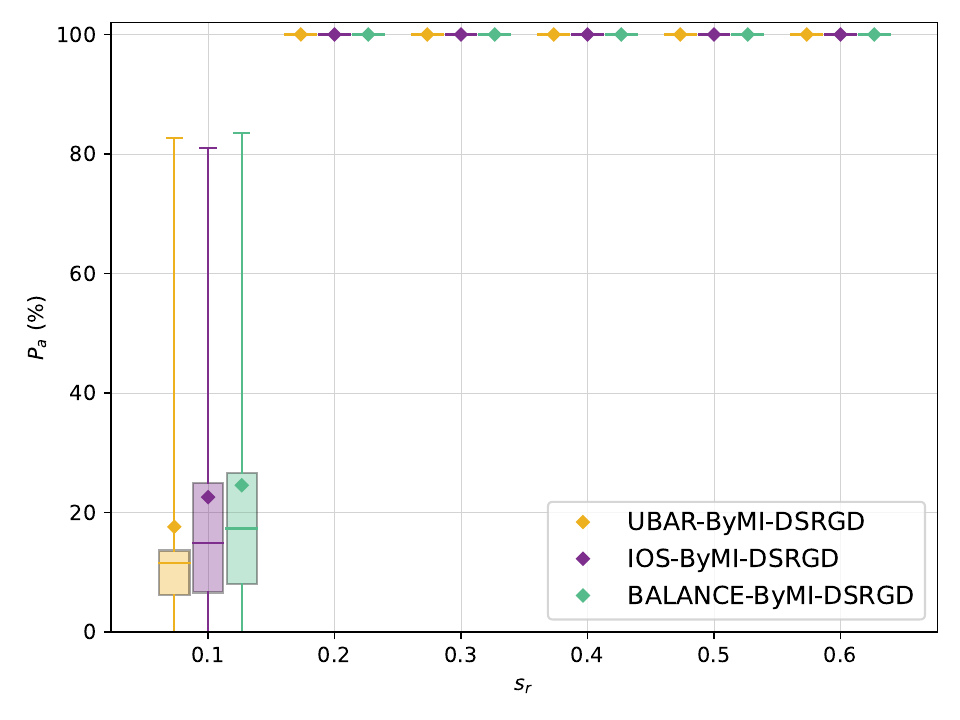}
    \caption{$\mathrm{P}_a$ ($d=30$)}
  \end{subfigure}
  \begin{subfigure}[b]{0.32\textwidth}
    \centering
    \includegraphics[width=\textwidth]{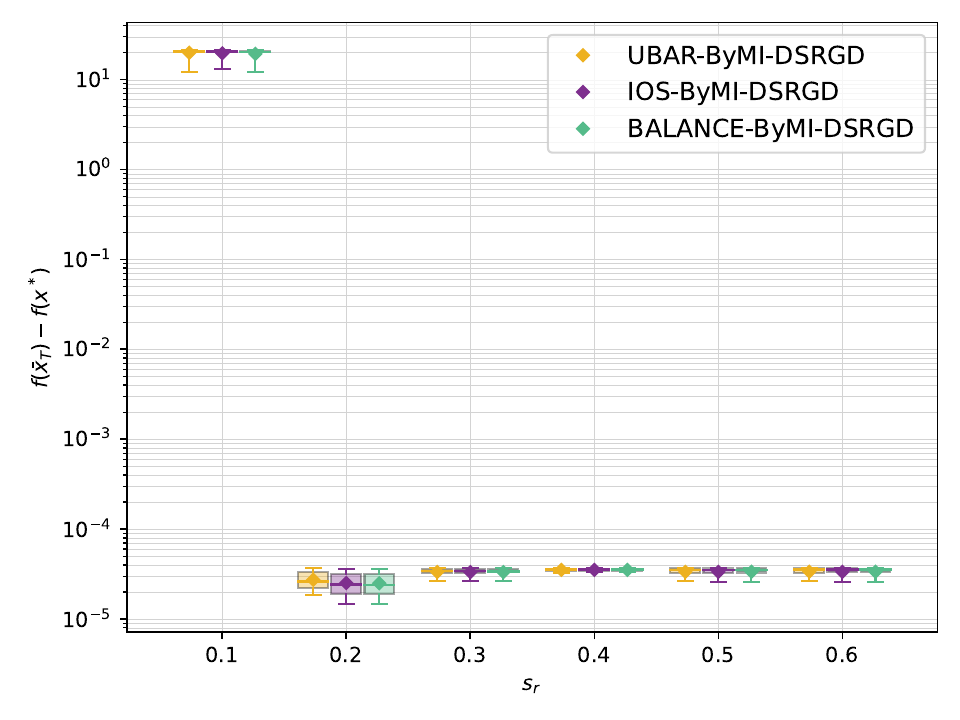}
    \caption{Optimality gap ($d=30$)}
  \end{subfigure}

  \vspace{0.5em}

    \begin{subfigure}[b]{0.32\textwidth}
    \centering
    \includegraphics[width=\textwidth]{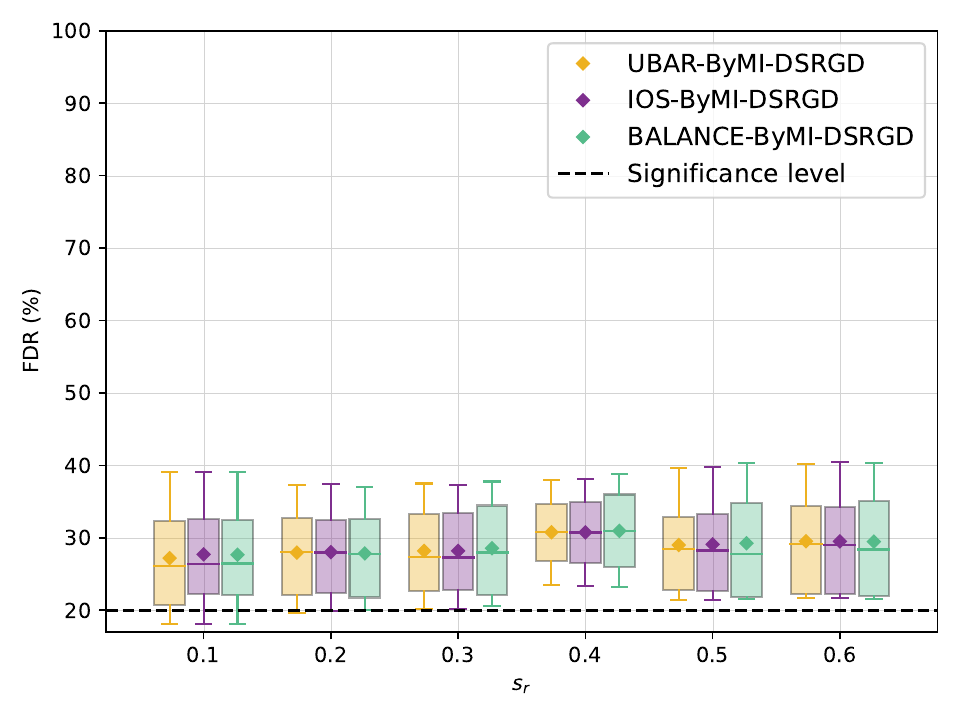}
    \caption{FDP ($d=80$)}
  \end{subfigure}
  \begin{subfigure}[b]{0.32\textwidth}
    \centering
    \includegraphics[width=\textwidth]{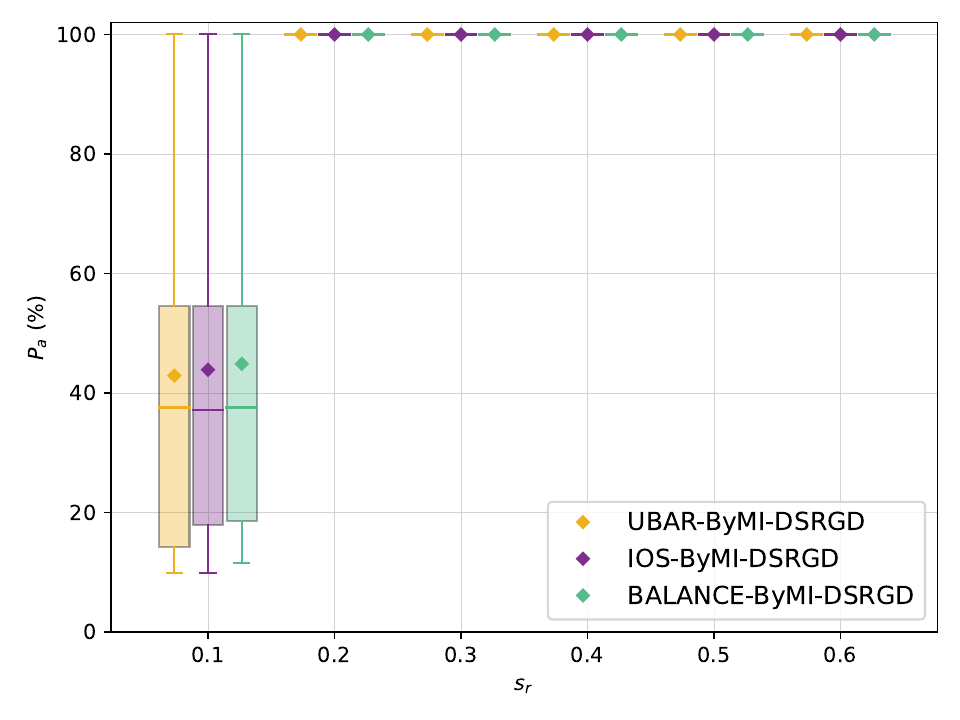}
    \caption{$\mathrm{P}_a$ ($d=80$)}
  \end{subfigure}
  \begin{subfigure}[b]{0.32\textwidth}
    \centering
    \includegraphics[width=\textwidth]{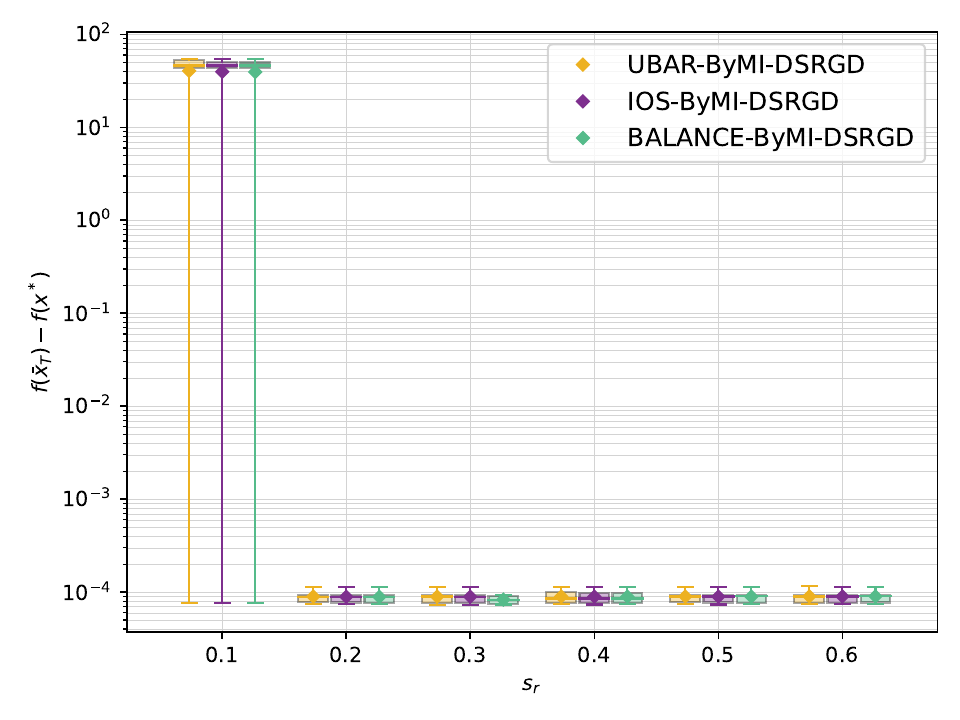}
    \caption{Optimality gap ($d=80$)}
  \end{subfigure}

  \caption{FDP, $\mathrm{P}_a$ and optimality gap of ByMI-type methods over $s_r$ when $\varrho = 0.2$ under Scenario A, with $d = 30$ (top row) and $d = 80$ (bottom row).}
  \label{fig:sr-dim30-80}
\end{figure}

Figure~\ref{fig:varrho-dim30-80} reports box-and-whisker plots  against the Byzantine ratio $\varrho$ under Scenario B.
When $\varrho < 0.3$, nearly all methods  achieve satisfactory FDPs for $d=30$, while slightly higher FDPs are observed for $d=80$.
Meanwhile, the performance of $\mathrm{P}_a$ and the optimality gap remains favorable, indicating that $(\Gcal, \Ecal'|_{\Gcal})$ is strongly connected.
When $\varrho \geq 0.3$, it can be observed that the FDP gradually increases, and $\mathrm{P}_a$ is below $100\%$ in some seeds and optimality gaps become more oscillatory.
One possible explanation is that, as the Byzantine ratio increases, the number of Byzantine neighbors surrounding certain nodes may exceed that of their normal neighbors.

Overall, DRSGD-ByMI methods are able to maintain low FDPs and achieve reliable $\mathrm{P}_a$  under relatively lower Byzantine ratios ($\varrho< 0.3$) and higher contamination levels ($s_r>0.1$), leading to a strongly connected graph and robust optimization performance.

\begin{figure}[htbp]
  \centering
    \begin{subfigure}[b]{0.32\textwidth}
    \centering
    \includegraphics[width=\textwidth]{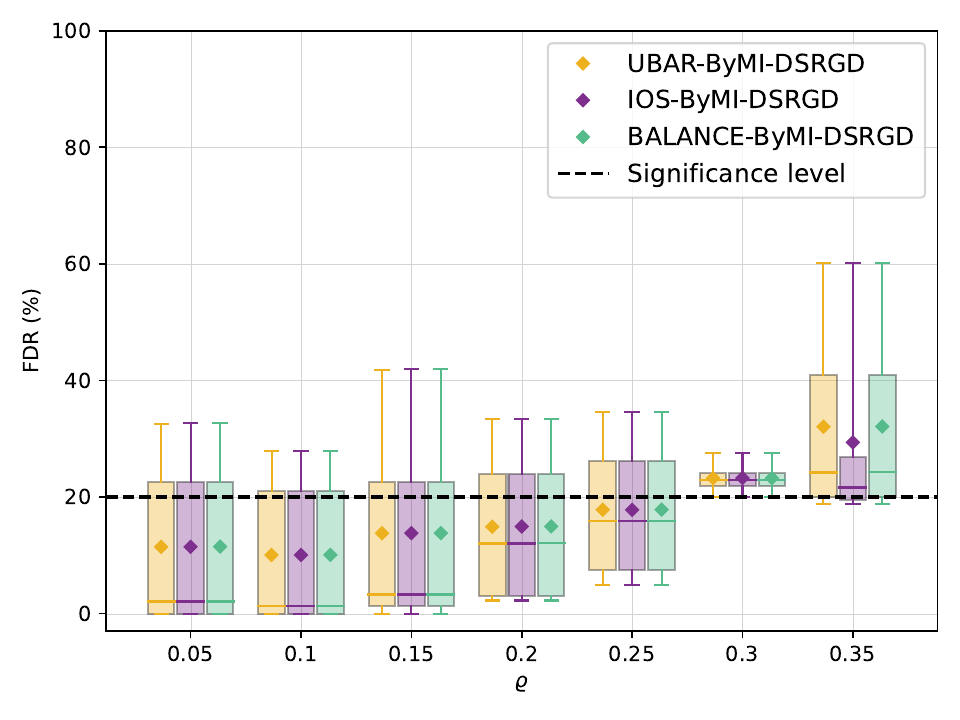}
    \caption{FDP ($d=30$)}
  \end{subfigure}
  \begin{subfigure}[b]{0.32\textwidth}
    \centering
    \includegraphics[width=\textwidth]{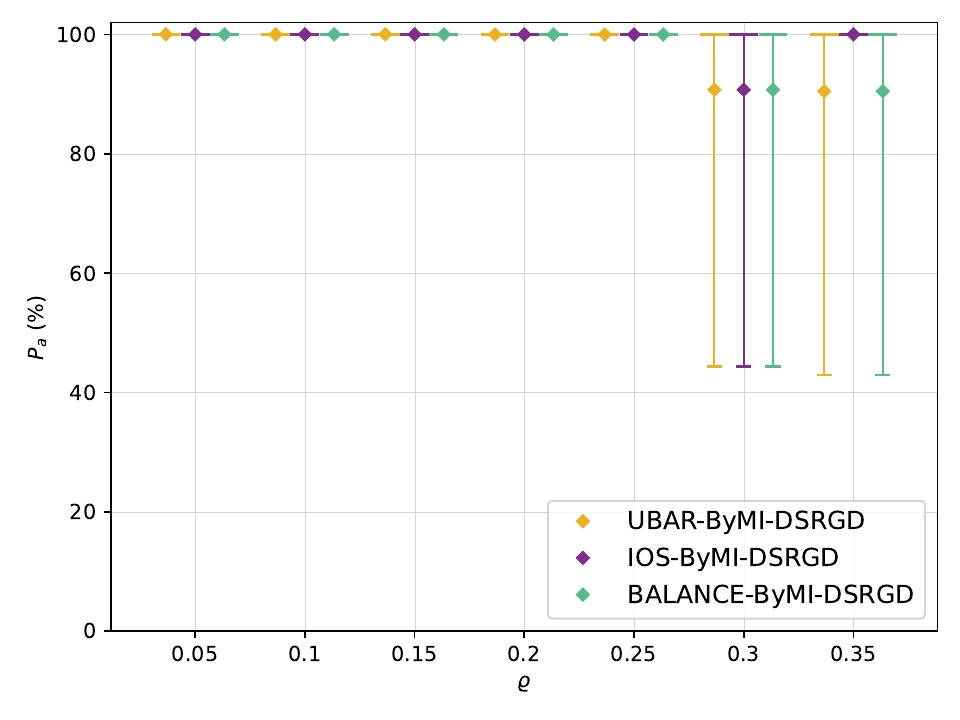}
    \caption{$\mathrm{P}_a$ ($d=30$)}
  \end{subfigure}
  \begin{subfigure}[b]{0.32\textwidth}
    \centering
    \includegraphics[width=\textwidth]{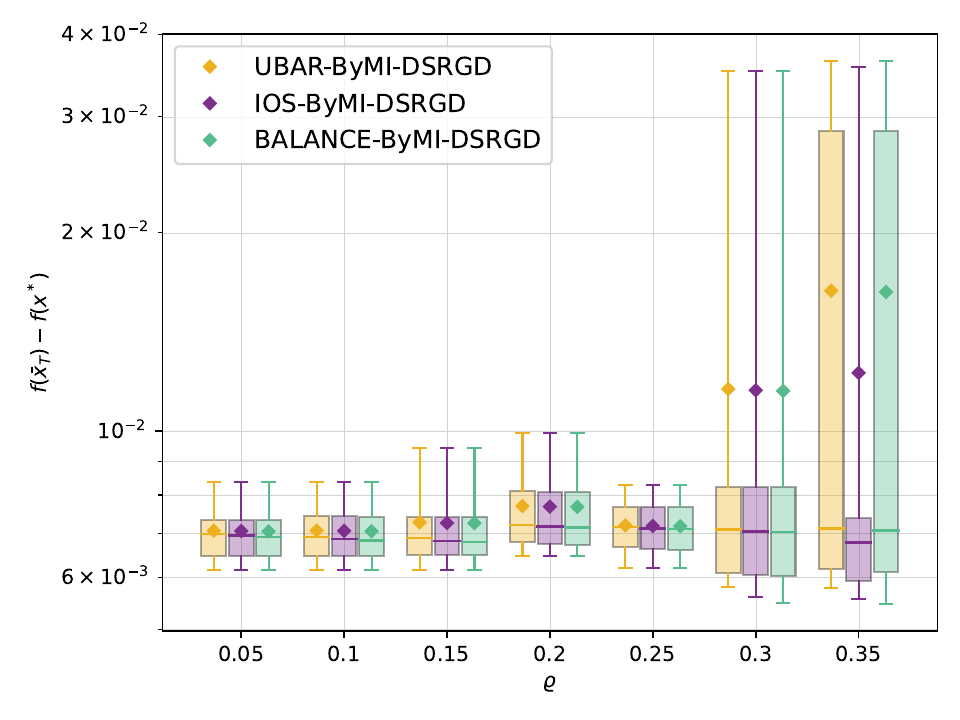}
    \caption{Optimality gap ($d=30$)}
  \end{subfigure}

  \vspace{0.5em}

    \begin{subfigure}[b]{0.32\textwidth}
    \centering
    \includegraphics[width=\textwidth]{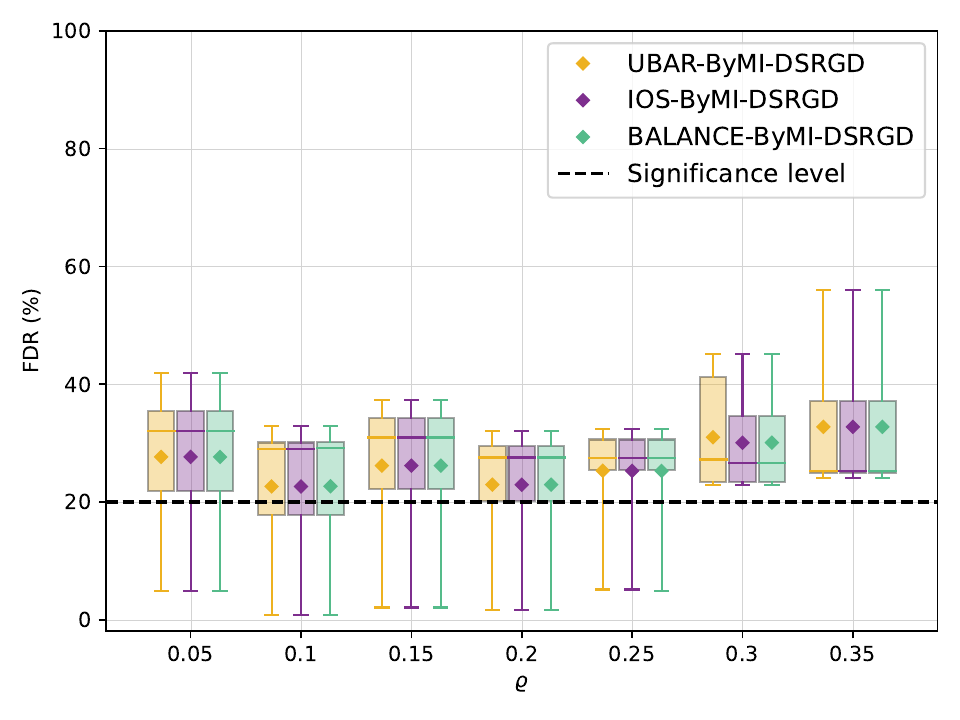}
    \caption{FDP ($d=80$)}
  \end{subfigure}
  \begin{subfigure}[b]{0.32\textwidth}
    \centering
    \includegraphics[width=\textwidth]{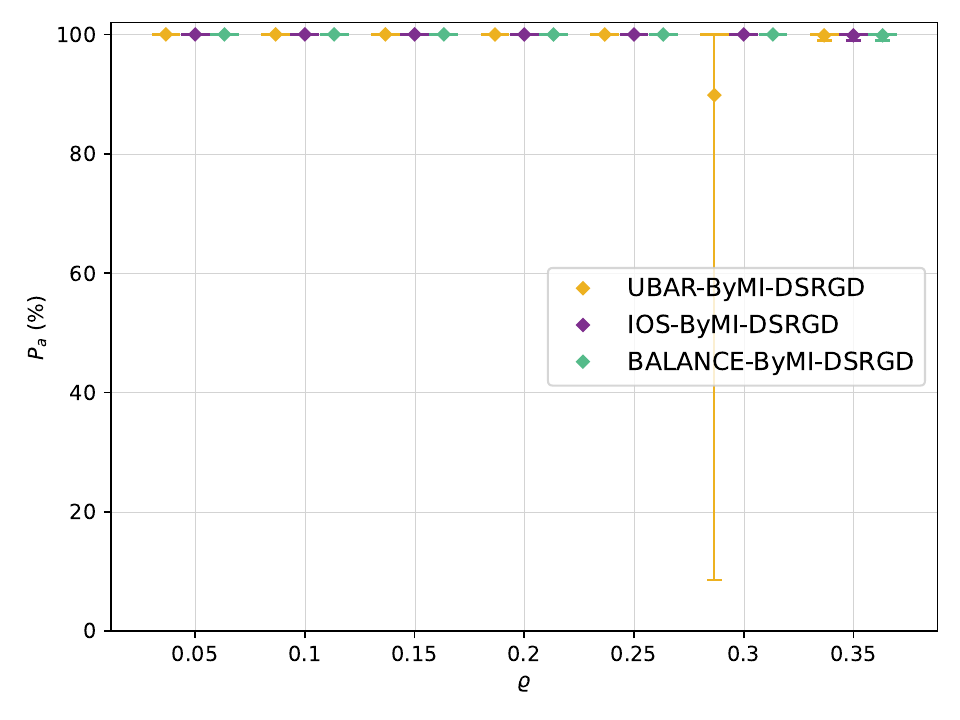}
    \caption{$\mathrm{P}_a$ ($d=80$)}
  \end{subfigure}
  \begin{subfigure}[b]{0.32\textwidth}
    \centering
    \includegraphics[width=\textwidth]{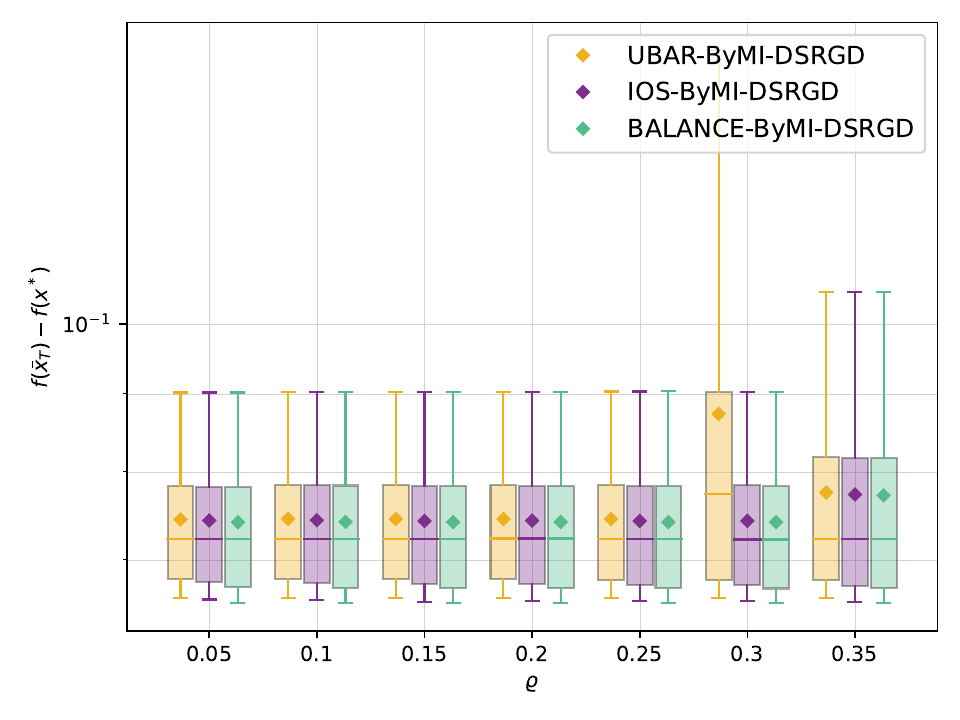}
    \caption{Optimality gap ($d=80$)}
  \end{subfigure}

  \caption{FDP, $\mathrm{P}_a$ and optimality gap of ByMI-type methods over $\varrho$ under Scenario B, with $d = 30$ (top row) and $d = 80$ (bottom row).}
  \label{fig:varrho-dim30-80}
\end{figure}

\subsection{Real-Data Experiments}

\textbf{Distributed image classification tasks.}
We conduct distributed image classification tasks on  MNIST \citep{lecun2002gradient} dataset and Fashion-MNIST \citep{xiao2017fashionmnist} dataset.
Both MNIST and Fashion-MNIST contain 60,000 training images and 10,000 test images of size $28 \times 28$ pixels.
While MNIST provides labels from 0 to 9 for handwritten digits, Fashion-MNIST assigns labels to 10 categories of clothing items.
Moreover, we adopt the LeNet neural network with the cross-entropy loss as our training model.
The weights of LeNet are initialized using Kaiming initialization, and the biases are initialized using a uniform distribution.
The decentralized network contains $m = 150$ nodes, and each node is randomly assigned a local dataset of equal size, and a local copy of LeNet neural network.
The default Byzantine ratio is set to $0.2$.
The size of the identification set is set to $n=200$. \revise{Similarly, we set the target significance
level $\alpha = 0.2$ empirically.}

To satisfy the Huber contamination model \eqref{eq:huber}, we require that the label distribution of each local dataset is identical across all nodes, and introduce three types of Byzantine attacks,
\begin{itemize}
\item Out-of-distribution (OOD) attack \citep{fort2021exploring}: we replace the sample ${ \bm s}_{i,l}$ on the Byzantine node by $\tilde{ \bm s}_{i,l}:= 0.3{ \bm s}_{i,l} + 0.7 {\bm \epsilon}_{d}$, where ${\bm \epsilon}_{d} \sim \Ncal_d({\bm \nu}_d,  I_d)$ with ${\bm \nu}_d\in \mathbb{R}^d$ randomly sampled from $\Ncal_d({\bm 0}, 20^2 I_d)$.
\item Gradient attack: each Byzantine gradient is replaced by  $\gbf_{i}:= 0.5 \bar\gbf^{\text{clean}} +  {\bm \epsilon}_{d}$, where $\bar\gbf^{\text{clean}}:= \frac{1}{m_g}\sum_{j\in \mathcal{G}} \gbf_{i}(\btheta) $,  ${\bm \epsilon}_{d} \sim \Ncal_d({\bm \nu}_d, \mathrm{std}^2(\{\gbf_{i}\}_{i \in \mathcal{G}}) I_d)$, and ${\bm \nu}_d\in \mathbb{R}^d$ is randomly sampled from the distribution $\Ncal_d({\bm 0}, 20^2 \mathrm{std}^2(\{\gbf_{i}\}_{i \in \mathcal{G}}) I_d)$.

\item Inner-product-manipulation (IPM) attack \citep{xie2020fall}: each Byzantine gradient is set to $-a \bar\Gbf^{\text{clean}}$, where $a >0$ with $a=1.0$ as default value.
\end{itemize}

\textbf{Methods for comparison.}
We conduct pairwise comparisons of our proposed methods with their corresponding decentralized Byzantine-robust algorithms in the warm-up phase.
\begin{itemize}
\item UBAR-DRSGD-ByMI versus UBAR;
\item IOS-DRSGD-ByMI versus IOS;
\item BALANCE-DRSGD-ByMI versus BALANCE.
\end{itemize}
All methods are run for $K=500$ iterations, and DRSGD-ByMI terminates the warm-up phase at $k_0 = 30$ for IPM attacks or $k_0 = 100$ for the other two attacks.
All decentralized Byzantine-robust algorithms adopt robust aggregation within decentralized SGD updates.

\textbf{Numerical comparisons.}
Tables \ref{tab:attack-results1} and \ref{tab:attack-results-fmnist} illustrate the comparison results under three different Byzantine attacks on the MNIST and Fashion-MNIST datasets, respectively.
Here, the aggregation rule represents the rule used in the corresponding Byzantine-robust algorithm or in the warm-up phase of DRSGD-ByMI. ``No robust rule'' means that the Byzantine-robust algorithm reduces to vanilla DSGD, whose results are reported to illustrate the severity of Byzantine attacks on unprotected decentralized systems.
We can observe that the DRSGD-ByMI methods achieve FDPs below $20\%$ and almost $100\%$ $\mathrm{P}_a$.
For all three attacks, the DRSGD-ByMI methods achieve  higher test accuracy  than the corresponding decentralized Byzantine-robust algorithms in terms of both ``Acc (all)'' and ``Acc (normal)'' on the MNIST dataset.

\begin{table}[htbp]
\footnotesize
    \centering
    \renewcommand{\arraystretch}{1.2}      {\fontsize{8pt}{10pt}\selectfont
    \begin{tabular}{lcccccc}
    \toprule
    \multirow{2}{*}{Attack} & \multirow{2}{*}{\makecell[c]{Aggregation  \\ rule}}
    & \multicolumn{2}{c}{Byzantine-robust algorithm}
    & \multicolumn{3}{c}{DRSGD-ByMI}  \\
    \cmidrule(lr){3-7}
    & & Acc (all) & Acc  (normal) & Acc (scc)  & FDP & $\mathrm{P}_a$ \\
    \midrule
            \multirow{5}{*}{OOD Attack}
    & No robust rule  & 57.8  & 65.0    &  &   &   \\
        & UBAR            & 92.4 & 93.3  & 97.2  & 3.6 &  100.0 \\
    & IOS & 90.2 & 95.2 & 97.4   & 39.9 &  100.0 \\
  & BALANCE      & 86.0 &  86.5 & 97.1  & 5.4  & 100.0  \\
         \midrule

    \multirow{5}{*}{Gradient Attack}
    & No robust rule     & 37.2 & 37.3   &  &  &   \\
        & UBAR      & 40.0 & 93.3  & 97.4   & 16.8 & 100.0   \\
    & IOS & 34.2 & 76.3   & 96.9   & 18.7 & 100.0  \\
    & BALANCE           & 57.8 &  97.3 & 97.3   & 16.6  & 100.0 \\
      \midrule
    \multirow{3}{*}{IPM Attack}
    &  No robust rule     &  10.4 & 10.4  &  &  &   \\
                & UBAR           & 25.1 & 93.0  & 97.6 & 21.7 & 100.0  \\
    & IOS  & 9.8 & 10.9 & 97.4  & 11.3 & 100.0 \\
 & BALANCE         & 8.5 & 97.2 & 97.3 & 16.8 & 100.0   \\
       \bottomrule
    \end{tabular}}
    \caption{Comparison of DRSGD-ByMI methods with three corresponding decentralized Byzantine-robust stochastic algorithms under different Byzantine attacks on the MNIST dataset.}
    \label{tab:attack-results1}
    \end{table}

\begin{table}[htbp]
\footnotesize
    \centering
    \renewcommand{\arraystretch}{1.2}
    {\fontsize{8pt}{10pt}\selectfont
    \begin{tabular}{lcccccc}
    \toprule
    \multirow{2}{*}{Attack} & \multirow{2}{*}{\makecell[c]{Aggregation \\ rule}}
    & \multicolumn{2}{c}{Byzantine-robust algorithm}
    & \multicolumn{3}{c}{DRSGD-ByMI}  \\
    \cmidrule(lr){3-7}
    & & Acc (all) & Acc (normal) & Acc (scc) & FDP & $\mathrm{P}_a$ \\
    \midrule

    \multirow{4}{*}{OOD Attack}
    & No robust rule    & 54.5 & 57.6 &  &  &\\
    & UBAR            & 78.0 & 80.2 & 82.1 & 27.0 & 99.3 \\
    & IOS             & 74.7 & 79.5 & 82.7 & 15.8 & 100.0 \\
    & BALANCE             & 79.1 & 83.5 & 84.0 & 12.5 & 100.0 \\
    \midrule

    \multirow{4}{*}{Gradient Attack}
    & No robust rule  & 37.8 & 38.2 &  &  &  \\
    & UBAR            & 13.3 & 80.4 & 85.1 & 23.7 & 100.0 \\
    & IOS             & 29.4 & 64.5 & 83.3 & 14.0 & 100.0 \\
    & BALANCE            & 30.8 & 83.8 & 84.1 & 14.6  & 100.0 \\
    \midrule

    \multirow{4}{*}{IPM Attack}
    & No robust rule  & 10.0 & 10.0 &  &  &  \\
    & UBAR            & 18.6 & 80.0 & 85.6 & 15.1 & 100.0 \\
    & IOS             & 8.9  & 13.2  & 84.8 & 24.3 & 100.0 \\
    & BALANCE             & 14.3  & 84.4  & 85.5 & 15.8 & 100.0 \\
    \bottomrule
    \end{tabular}}
    \caption{Comparison of DRSGD-ByMI methods under different Byzantine attacks on the Fashion-MNIST dataset.}
    \label{tab:attack-results-fmnist}
\end{table}

Figures \ref{fig:comparison1}--\ref{fig:comparison3} present the curves of the norm of the global gradient and the test accuracy over training iterations under three Byzantine attacks on the MNIST dataset. For the DRSGD-ByMI methods, both the global gradient norm and the test accuracy are evaluated on the largest strongly connected component of the pruned graph, whereas for the corresponding decentralized Byzantine-robust algorithms, they are evaluated over the normal nodes.
 As shown in all figures, the DRSGD-ByMI methods consistently achieve high ``Acc (scc)'', and their final performance is numerically better than the baselines ``Acc (all)'' and also better than, or competitive with, their ``Acc (normal)'' across all three attacks.

Moreover, the norm of the global gradient for the three DRSGD-ByMI methods exhibits asymptotic convergence and decreases to the range of $10^{-1}\sim 10^{-2}$ after 500 iterations, which is consistent with the theoretical result of high-probability exact convergence. In contrast, their Byzantine-robust counterparts remain at the scale of $10^0\sim 10^1$ or even diverge. This phenomenon can be attributed to the fact that the convergence rate of the norm of the global gradient evaluated at the averaged model over the normal nodes involves a non-vanishing steady-state error term, as pointed out in \citep{wu-byzantine-resilient-2023, fang-byzantine-robust-2024}.
\revise{Furthermore, the inferior empirical performance of IOS and UBAR compared to BALANCE can be attributed to their strict reliance on an exact prior estimate of the Byzantine neighbor count.
Any mis-specification of this parameter can cause the accumulated expected consensus error among normal nodes to become unbounded, subsequently inflating the norm of the global gradient.}

\begin{figure}[tb]
  \centering
    \begin{subfigure}[b]{0.45\textwidth}
    \centering
    \includegraphics[width=\textwidth]{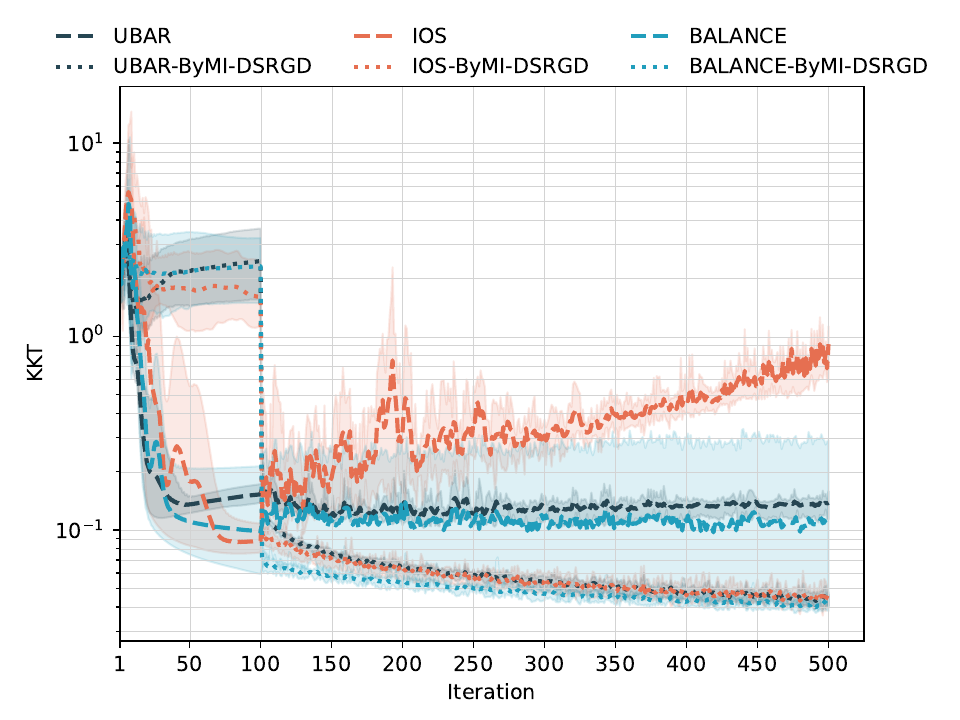}
    \caption{Norm of global gradient}
  \end{subfigure}
    \begin{subfigure}[b]{0.45\textwidth}
    \centering
    \includegraphics[width=\textwidth]{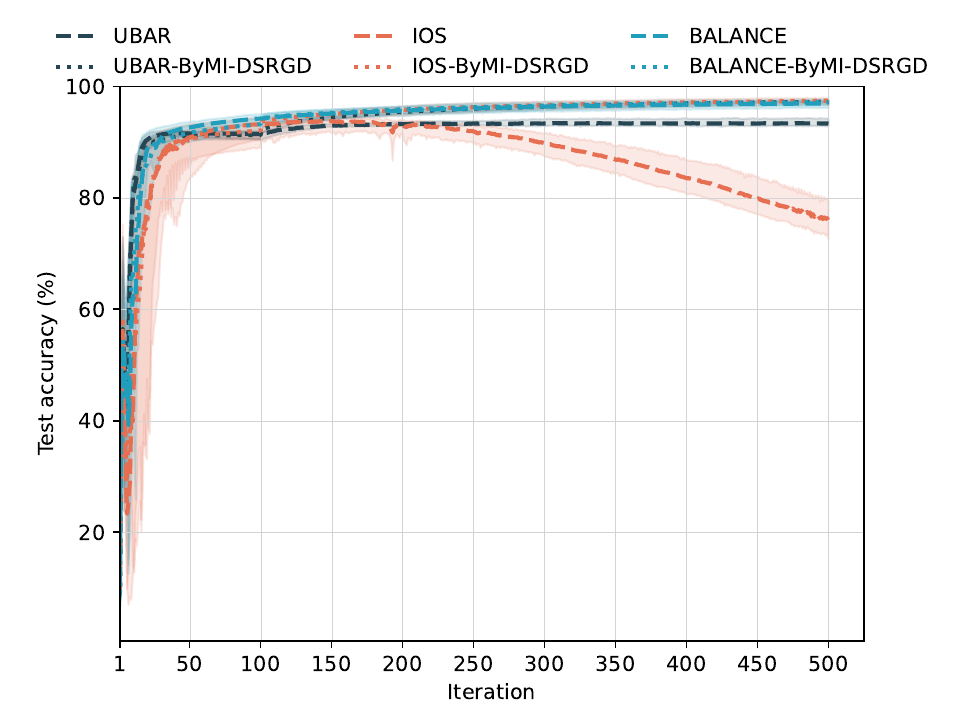}
    \caption{Test accuracy}
  \end{subfigure}
  \caption{Comparison of DRSGD-ByMI with corresponding decentralized Byzantine-robust stochastic algorithms under OOD attack on the MNIST dataset.}
  \label{fig:comparison1}
\end{figure}

\begin{figure}[htbp]
  \centering
    \begin{subfigure}[b]{0.45\textwidth}
    \centering
    \includegraphics[width=\textwidth]{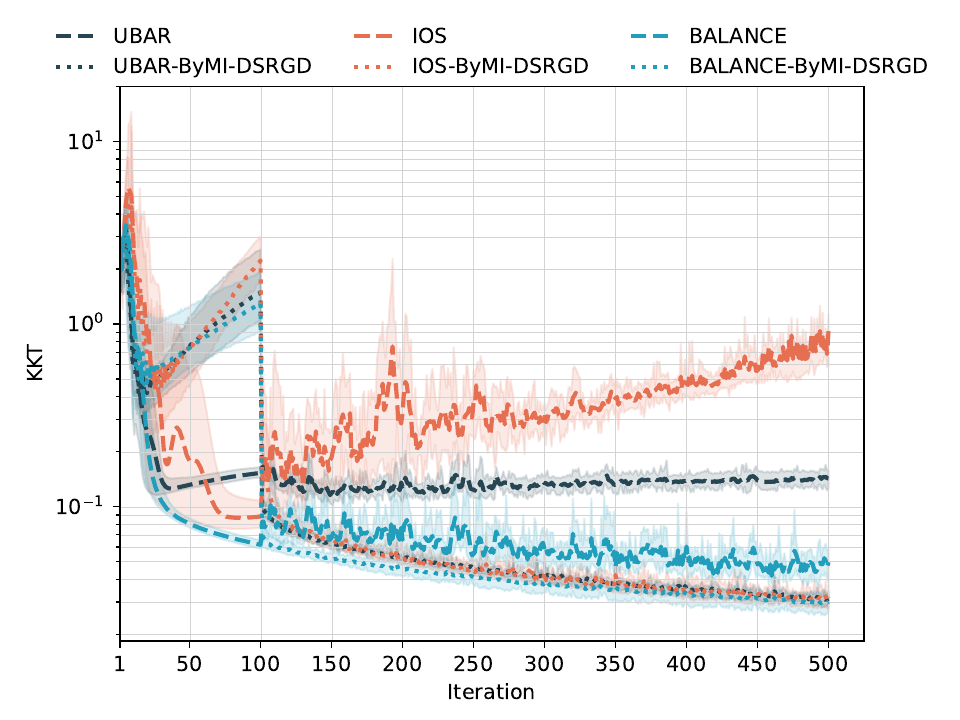}
    \caption{Norm of global gradient}
  \end{subfigure}
    \begin{subfigure}[b]{0.45\textwidth}
    \centering
    \includegraphics[width=\textwidth]{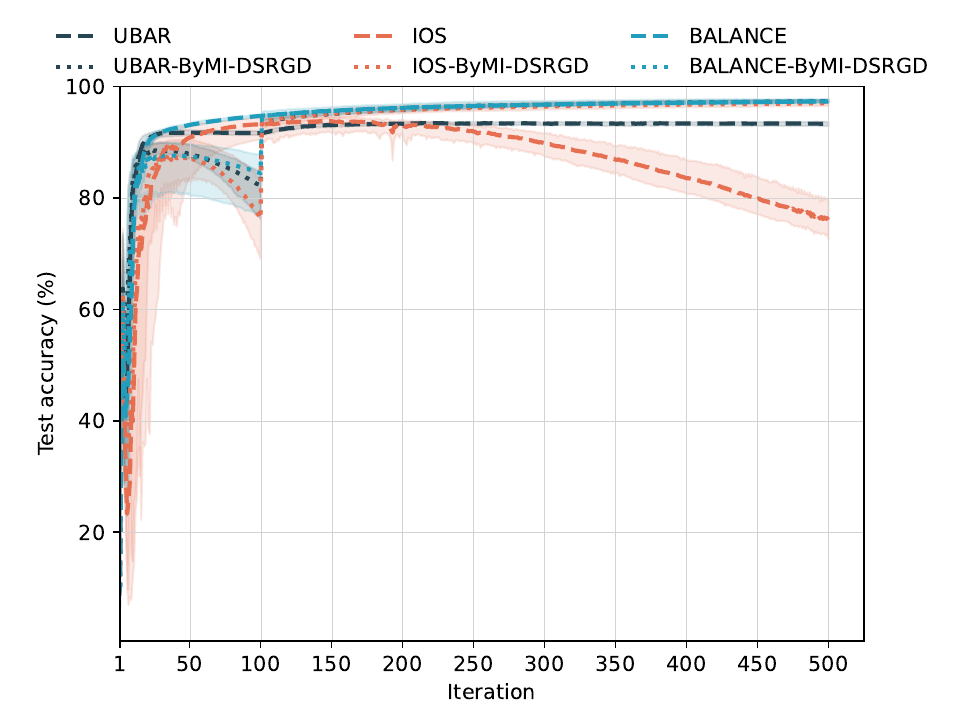}
    \caption{Test accuracy}
  \end{subfigure}

  \caption{Comparison of DRSGD-ByMI with corresponding decentralized Byzantine-robust stochastic algorithms under gradient attack on the MNIST dataset.}
  \label{fig:comparison2}
\end{figure}

\begin{figure}[htbp]
  \centering
    \begin{subfigure}[b]{0.45\textwidth}
    \centering
    \includegraphics[width=\textwidth]{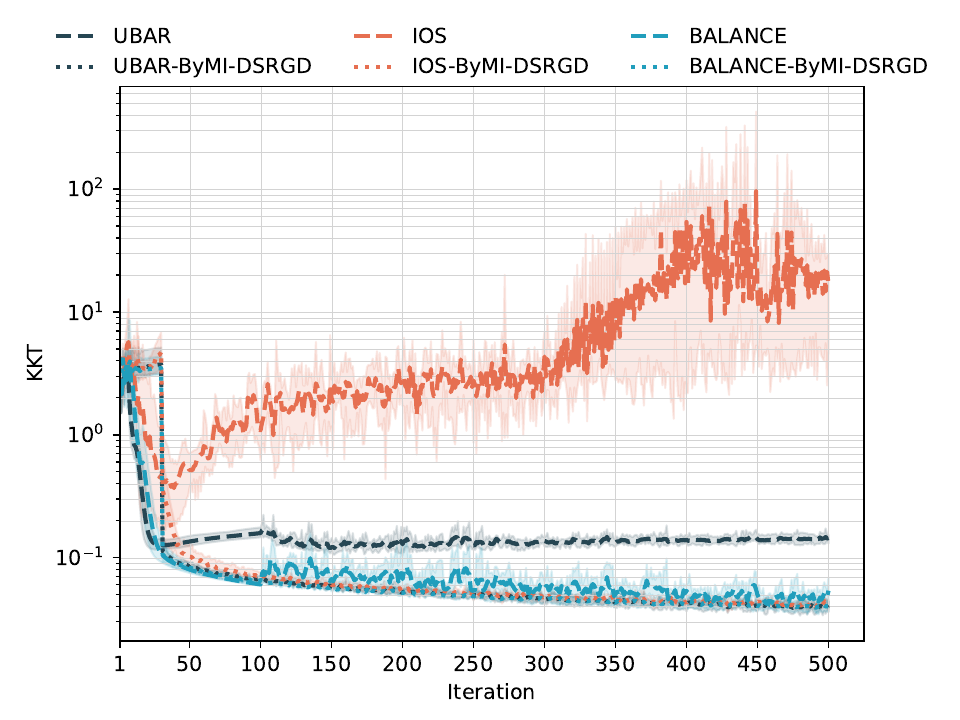}
    \caption{Norm of global gradient}
  \end{subfigure}
    \begin{subfigure}[b]{0.45\textwidth}
    \centering
    \includegraphics[width=\textwidth]{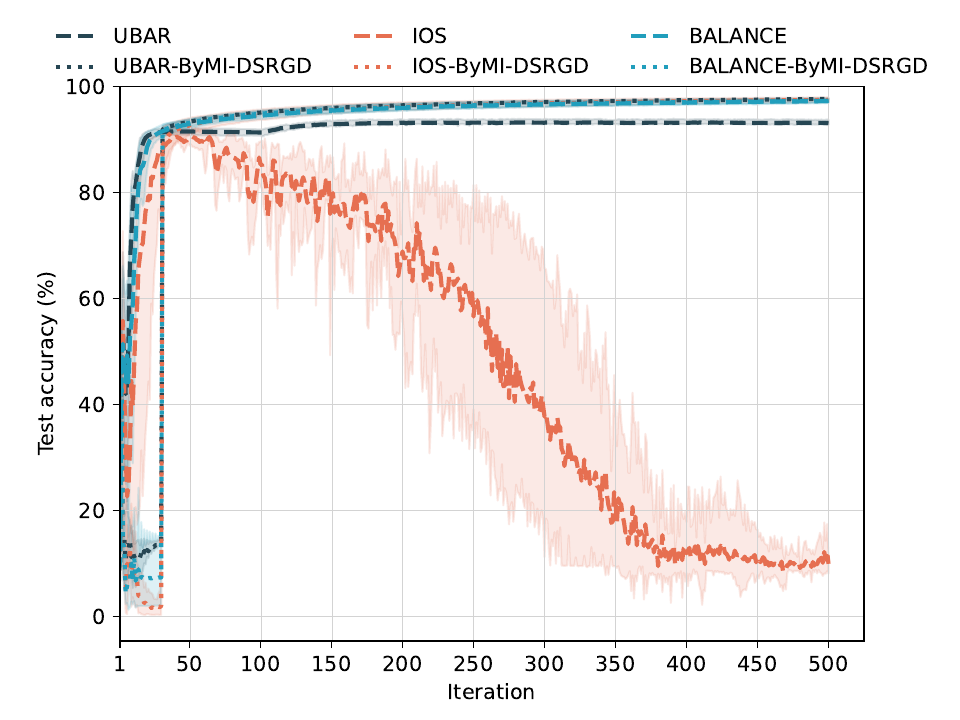}
    \caption{Test accuracy}
  \end{subfigure}

  \caption{Comparison of DRSGD-ByMI with corresponding decentralized Byzantine-robust stochastic algorithms under IPM attack on the MNIST dataset.}
  \label{fig:comparison3}
\end{figure}

\textbf{Numerical performance under different Byzantine ratios.}
Table~\ref{tab:attack-results2} and \ref{tab:attack-results3} present the FDP,  $\mathrm{P}_a$ and Acc (scc)  of our DRSGD-ByMI methods under different Byzantine attacks and Byzantine ratios on the MNIST dataset and Fashion-MNIST dataset.
It can be seen that $\mathrm{P}_a$ approaches $100\%$ in most cases and Acc (scc) shows limited sensitivity to the Byzantine ratios.
Under certain attack types and Byzantine ratios that cause $\mathrm{P}_a$  to fall below $100\%$, Acc (scc) exhibits a substantial degradation, indicating that the largest strongly connected component may contain Byzantine nodes.
Nevertheless, FDPs are controlled below $30\%$ in almost all cases.
Overall, DRSGD-ByMI methods not only enhance robustness and interpretability in decentralized learning tasks under Byzantine attacks, but also efficiently identify Byzantine machines across different Byzantine ratios.

\begin{table}[htbp]
\centering
{\fontsize{8pt}{10pt}\selectfont
\begin{tabular}{lcccccccccc}
\toprule
 \multirow{2}{*}{Attack} &  \multirow{2}{*}{Method}  & \multicolumn{3}{c}{$\varrho=0.1$} & \multicolumn{3}{c}{$\varrho=0.2$} & \multicolumn{3}{c}{$\varrho=0.3$} \\
\cmidrule(lr){3-5} \cmidrule(lr){6-8} \cmidrule(lr){9-11}
 &  & FDP & $\mathrm{P}_a$ & Acc (scc) & FDP & $\mathrm{P}_a$ & Acc (scc) & FDP & $\mathrm{P}_a$ & Acc (scc) \\
\midrule
 \multirow{3}{*}{\makecell[c]{OOD Attack}}
  & {U-B-D} & 12.7 & 100.0 & 97.2 & 3.6 & 100.0 & 97.2 & 6.5 & 100.0 & 97.3 \\
  & {I-B-D} & 28.0 & 100.0 & 97.0 & 39.9 & 100.0 & 97.4 & 42.8 & 100.0 & 97.0 \\
  & {B-B-D} & 17.4 & 100.0 & 97.3 & 5.4 & 100.0 & 97.1 & 6.8 & 100.0 & 97.2 \\
\midrule
 \multirow{3}{*}{\makecell[c]{Gradient Attack}}
  & {U-B-D}  & 23.6 & 100.0 & 97.4 & 16.8 & 100.0 & 97.4 & 16.5 & 100.0 & 97.3 \\
  & {I-B-D}   & 27.5 & 100.0 & 96.7 & 18.7 & 100.0 & 96.9 & 17.1 & 100.0 & 96.9 \\
  & {B-B-D}   & 25.2 & 100.0 & 97.2 & 16.6 & 100.0 & 97.3 & 16.6 & 100.0 & 97.3 \\
\midrule
\multirow{3}{*}{\makecell[c]{IPM Attack}}
  & {U-B-D} & 24.0 & 99.7 & 78.9 & 21.7 & 100.0 & 97.6 & 15.2 & 100.0 & 97.6 \\
  & {I-B-D} & 19.6 & 99.9 & 78.8 & 11.3 & 100.0 & 97.4 & 11.8 & 100.0 & 97.7 \\
  & {B-B-D} & 16.0 & 100.0 & 97.4 & 16.8 & 100.0 & 97.3 & 11.4 & 100.0 & 97.3 \\
\bottomrule
\end{tabular}}
\caption{Comparison of FDP, $\mathrm{P}_a$ and Acc (scc) against Byzantine ratio $\varrho$ under different Byzantine attack scenarios on the MNIST dataset.
Here, U-B-D, I-B-D, and B-B-D denote UBAR-DRSGD-ByMI, IOS-DRSGD-ByMI, and BALANCE-DRSGD-ByMI, respectively.}
\label{tab:attack-results2}
\end{table}

\begin{table}[htbp]
\centering
{\fontsize{8pt}{10pt}\selectfont
\begin{tabular}{lcccccccccc}
\toprule
 \multirow{2}{*}{Attack} &  \multirow{2}{*}{Method}
 & \multicolumn{3}{c}{$\varrho=0.1$}
 & \multicolumn{3}{c}{$\varrho=0.2$}
 & \multicolumn{3}{c}{$\varrho=0.3$} \\
\cmidrule(lr){3-5} \cmidrule(lr){6-8} \cmidrule(lr){9-11}
 &  & FDP &  $\mathrm{P}_a$ & Acc (scc)
 & FDP  & $\mathrm{P}_a$  & Acc (scc)
 & FDP &  $\mathrm{P}_a$  & Acc (scc)  \\
\midrule
 \multirow{3}{*}{\makecell[c]{OOD Attack}}
  & {U-B-D} & 38.4 & 99.1 & 83.7 & 27.0  & 99.3 & 82.1 & 17.4 &  94.3& 69.0 \\
  & {I-B-D} & 24.3  & 99.6  & 82.4  & 15.8 & 100.0 & 82.7 & 12.2 & 99.6 &  79.8\\
  & {B-B-D} & 18.9 & 99.8 & 83.2 & 12.5 & 100.0 & 84.0 & 17.8 & 89.6 & 73.6 \\
\midrule
 \multirow{3}{*}{\makecell[c]{Gradient Attack}}
  & {U-B-D}  & 35.5 & 99.0 & 79.4 & 23.7 & 100.0 & 85.1 & 14.5 & 100.0 & 85.1 \\
  & {I-B-D}  & 19.9 & 100.0 & 82.8 & 14.0 & 100.0 & 83.3 & 9.8 & 100.0 & 83.5 \\
  & {B-B-D}  & 18.8 & 100.0 & 84.1 & 14.6 & 100.0 & 84.1 & 11.5 & 100.0 & 84.4 \\
\midrule
\multirow{3}{*}{\makecell[c]{IPM Attack}}
  & {U-B-D} & 15.9 & 100.0 & 85.5 & 15.1 & 100.0 & 85.6 & 12.7 & 100.0 & 85.6 \\
  & {I-B-D} & 20.4 & 100.0 & 84.3 & 24.3 & 100.0 & 84.8 & 16.5 & 100.0 & 85.1 \\
  & {B-B-D} & 27.9 & 100.0 & 85.2 & 15.8 & 100.0 & 85.5 & 13.4 & 100.0 & 85.3 \\
\bottomrule
\end{tabular}}
\caption{Comparison of FDP, $\mathrm{P}_a$ and Acc (scc) against Byzantine ratio $\varrho$ under different Byzantine attack scenarios on Fashion-MNIST dataset.
Here, U-B-D, I-B-D, and B-B-D denote UBAR-DRSGD-ByMI, IOS-DRSGD-ByMI, and BALANCE-DRSGD-ByMI, respectively.}
\label{tab:attack-results3}
\end{table}

\section{Conclusions}

Decentralized learning systems under Byzantine attacks have received increasing attention.
However, the existing Byzantine-robust methods are less than satisfactory in terms of exact convergence.
From a detect-then-optimize perspective, this paper proposes a data-driven and p-value-free method, DRSGD-ByMI. The proposed detection procedure achieves finite-sample FDP (or false discovery number) control and sure-detection with high probability via reliable robust estimators, which guarantee that the pruned graph has a strongly connected component over normal nodes.
We introduce a rescaled stochastic gradient descent algorithm for the nonconvex setting with row-stochastic mixing matrix.
The algorithm attains a non-asymptotic convergence rate of
$\mathcal{O}(1/\sqrt{K m_g})$, and exhibits linear speedup in the asymptotic phase of the iterations as the number of nodes increases.
Therefore, it serves as a useful tool for solving robust distributed learning problems.

There remain many promising directions for future research.
First, it is worth investigating whether the detection procedure can be further improved and statistically characterized under heterogeneous (i.e., non-i.i.d.) data distributions.
Second, in dynamic Byzantine environments, it is an open question whether multiple rounds of ByMI, in conjunction with DRSGD, can effectively handle time-varying Byzantine behaviors.

\newpage
\appendix

\section{Certifying Proposition~\ref{prop:robust_est}}\label{appendix:Prop4.4}
We bound the estimation error $\norm{\widehat{\gbf}_i - \gbf_i^\ast}$ by constructing a proxy estimator and bounding the bias. Here, $\gbf_i^\ast = \mathbb{E}_{\sbf \sim \Pcal}[\nabla \ell(\btheta_{i, 0}; \sbf)]$, the robust estimator $\widehat{\gbf}_i$ is computed using $\{ \gbf_{j} (\btheta_{j, 0}; \xi_{j, 0}^{(1)})\}_{j \in \Ncal_i}$.

\noindent\textbf{Step 1: Constructing the Proxy Estimator.}
Let us define a proxy set of gradients $\mathcal{S} := \{ \gbf_{j} (\btheta_{i, 0}; \xi_{j, 0}^{(1)})\}_{j \in \Ncal_i}$.
In this proxy set, all gradients are computed at the \textit{same} parameter $\btheta_{i, 0}$ belonging to node $i$.
Let $\widehat{\gbf}_i'$ be the output of the robust mean estimator when applied to the proxy set $\mathcal{S}' := \{ \gbf_{j} (\btheta_{i, 0}; \xi_{j, 0}^{(1)})\}_{j \in \Ncal_i}$.

For any normal neighbor $j \in \Gcal_i$, the element $\gbf_{j} (\btheta_{i, 0}; \xi_{j, 0}^{(1)})$ is an unbiased estimator of $\gbf_i^\ast$, i.e.,
$$\mathbb{E}_{\xi_{j, 0}^{(1)}}[\gbf_{j} (\btheta_{i, 0}; \xi_{j, 0}^{(1)})] = \mathbb{E}_{\xi_{i, 0}^{(1)}}[\gbf_{i} (\btheta_{i, 0}; \xi_{i, 0}^{(1)})] = \gbf_i^\ast,
$$
which is due to the distributional homogeneity among normal nodes.
Furthermore, since the datasets $\xi_{j ,0}^{(1)}$ are independent across nodes, the elements in $\mathcal{S}'_i$ corresponding to normal nodes are independent and identically distributed (i.i.d.) with mean $\gbf_i^\ast$ and bounded variance (scaled by sample size $n/2$).

Standard results for robust mean estimators (e.g., coordinate-wise median, geometric median) guarantee that when a sufficient fraction of inputs are i.i.d. around a true mean $\gbf_i^\ast$, the estimation error is bounded by the statistical noise.
Thus, with high probability:
\begin{equation}\label{eq:proxy_error}
\norm{\widehat{\gbf}_i' - \gbf_i^\ast} \le \textrm{err}_{\textrm{rob}}(\epsilon) = o(1),
\end{equation}
where for the coordinate-wise robust methods like median or trimmed mean, we have $\textrm{err}_{\textrm{rob}}(\epsilon) = \mathcal{O}(\sigma_0 \sqrt{\epsilon d/n})$ and for strongly robust methods like Filtering, we have $\textrm{err}_{\textrm{rob}}(\epsilon) = \mathcal{O}(\sigma_0 \sqrt{\epsilon/n})$.
Please refer to \citep{zhu2023byzantine}.

\noindent\textbf{Step 2: Bounding the Input Perturbation (Bias).}
Now we compare the actual inputs $\{ \gbf_{j} (\btheta_{i, 0}; \xi_{j, 0}^{(1)})\}_{j \in \Ncal_i}$ with the proxy inputs $\{ \gbf_{j} (\btheta_{i, 0}; \xi_{j, 0}^{(1)})\}_{j \in \Ncal_i}$.
For any neighbor $j$, the difference is:
$$ \boldsymbol{\delta}_{g, j} = \gbf_{j} (\btheta_{j, 0}; \xi_{j, 0}^{(1)}) - \gbf_{j} (\btheta_{i, 0}; \xi_{j, 0}^{(1)}) = \frac{1}{|\xi_{j ,0}^{(1)}|} \sum\limits_{s \in \xi_{j ,0}^{(1)}} \left( \nabla \ell(\btheta_{j, 0}; \sbf) - \nabla \ell(\btheta_{i, 0}; \sbf) \right). $$
Using the $L$-smoothness of the empirical loss function (Condition~\ref{asp:loss}), we have for any sample $s$:
$$ \norm{\nabla \ell(\btheta_{j, 0}; \sbf) - \nabla \ell(\btheta_{i, 0}; \sbf)} \le L \norm{\btheta_{j, 0} - \btheta_{i, 0}}. $$
Consequently, the perturbation is bounded deterministically by:
$$ \norm{\boldsymbol{\delta}_{g, j}} \le L \norm{\btheta_{j, 0} - \btheta_{i, 0}}. $$
From the Warm-up condition (Condition~\ref{asp:warm_up}), we know that $\norm{\btheta_{j, 0} - \bar{\btheta}_{0}} = \mathcal{O}(1/\sqrt{k_0^{1-\delta}})$.
By the triangle inequality:
$$ \norm{\btheta_{j, 0} - \btheta_{i, 0}} \le \norm{\btheta_{j, 0} - \bar{\btheta}_{0}} + \norm{\bar{\btheta}_{0} - \btheta_{i, 0}} = \mathcal{O}\biggl(\frac{L}{\sqrt{k_0^{1-\delta}}}\biggr). $$
Thus, the input perturbation bias \revise{$b := \max_{j \in  \Gcal_i} \norm{\boldsymbol{\delta}_{g, j}}$ satisfies the following equality
\begin{equation}\label{eq:bias_bound}
b  = \mathcal{O}\biggl(\frac{L}{\sqrt{k_0^{1-\delta}}}\biggr) = o(1).
\end{equation}}

\noindent\textbf{Step 3: Stability of the Robust Estimator.}
The above justification shows that for every $j \in \Ncal_i$, the corresponding gradients in the two input sets $\{ \gbf_{j} (\btheta_{j, 0}; \xi_{j, 0}^{(1)})\}_{j \in \Ncal_i}$ and $\{\gbf_{j} (\btheta_{i, 0}; \xi_{j, 0}^{(1)})\}_{j \in \Ncal_i}$ differ by at most $\norm{\boldsymbol{\delta}_{g, j}}$ for each coordinate, and $\max_{j} \norm{\boldsymbol{\delta}_{g, j}} \le b$.

Robust mean estimators like the geometric median, coordinate-wise median or Filtering are stable with respect to small perturbations of the input data.
Specifically, they are Lipschitz continuous with respect to the input set (in the appropriate metric) provided the fraction of outliers is below the breakdown point.
This stability property means that if $\norm{\widehat{\gbf}_i' - {\gbf}_i^\ast} \le \delta_{\gbf}$ then
\begin{equation}\label{eq:stability}
\norm{\widehat{\gbf}_i - {\gbf}_i^\ast} \le \delta_{\gbf} + C_{\text{stab}} \cdot b = \delta_{\gbf} + \mathcal{O}\biggl(\frac{L}{\sqrt{k_0^{1-\delta}}}\biggr),
\end{equation}
for some suitable factor $C_{\text{stab}}$.

For coordinate-wise estimators like the simple mean, the geometric median and the coordinate-wise median, the above claim is easy to verify.
Here, we discuss the stability of the more advanced robust estimator, Filtering, which achieves a dimension-agnostic bias $O_P(\sqrt{\epsilon/n})$ against the corruption from the Byzantine machines.
Given samples $\set{\xbf_j}_{j \in [m]}$ and contamination level $\epsilon \in (0, 1/2]$, the Filtering estimate
$$\widehat{\gbf}_i = \sum_{j \in \Ncal_i} \widehat{w}_j \gbf_{j} (\btheta_{j, 0}; \xi_{j, 0}^{(1)})$$ is a weighted average of the samples  with the weight $\widehat{\wbf}$ defined by
\begin{equation}
    \widehat{\wbf} := \argmin_{\wbf \in \Delta_{\Ncal_i, \epsilon}} \norm{\Sigma_{\wbf}}_{\mathrm{op}},
\end{equation}
where the feasible set is
$$\Delta_{\Ncal_i, \epsilon} := \Bigset{\wbf = (w_j)_{j\in \Ncal_i}:  w_j \in [0, (1 - \epsilon)^{-1}], \sum_{j \in \Ncal_i} w_j = 1},$$
and the weighted covariance is
$$\Sigma_{\wbf} := \size{\Ncal_i}^{-1} \sum_{j \in \Ncal_i} w_j (\gbf_{j} (\btheta_{j, 0}; \xi_{j, 0}^{(1)}) - \mu_{\wbf}) (\gbf_{j} (\btheta_{j, 0}; \xi_{j, 0}^{(1)}) - \mu_{\wbf})^\top,$$
with $\mu_{\wbf} := \sum_{j \in \Ncal_i} w_j \gbf_{j} (\btheta_{j, 0}; \xi_{j, 0}^{(1)})$.

The estimation error rate of the Filtering estimate can be guaranteed by the following concept called $(\epsilon, \delta)$-stability \citep{diakonikolas2020stability}.
\begin{definition}[$(\epsilon, \delta)$-stability]\label{def:stable}
    Let $\Scal = \set{\xbf_j}_{j \in \Ncal_i}$ be $\size{\Ncal_i}$ samples.
    For $\epsilon \in [0, 1/2]$ and $\delta \ge \epsilon$, we say that $\Scal$ is $(\epsilon, \delta)$-stable w.r.t. $\mu$ and $\sigma_0^2$ if for any subset $\Scal' \subset \Scal$ with $\size{\Scal'} \ge (1 - \epsilon) \size{\Scal}$, we have (1) $\norm{\mu_{\Scal'} - \gbf_i^\ast} \le \sigma_0 \delta$, (2) $\norm{\Sigma_{\Scal'} - \sigma_0^2 \Ibf}_{\mathrm{op}} \le \sigma_0^2 \delta^2 / \epsilon$, where $\Ibf$ is the identity matrix, $\mu_{\Scal'} = \size{\Scal'}^{-1} \sum_{j \in \Scal'} \xbf_j$ and $\Sigma_{\Scal'} = \size{\Scal'}^{-1} \sum_{j \in \Scal'} (\xbf_j - \mu) (\xbf_j - \mu)^\top$.
\end{definition}

By Theorems 1.4 in \citet{diakonikolas2020stability}, we have with probability at least $1 - \tau$, $\set{\gbf_{j} (\btheta_{i, 0}; \xi_{j, 0}^{(1)})}_{j \in \Ncal_i}$ differs by at most $\epsilon \size{\Ncal_i}$ points from an $(\epsilon, \delta)$-stable set
$\set{\xbf_j'}_{j \in \Ncal_i}$ with $\delta = \mathcal{O}(\sigma_0\sqrt{(\log(\tau) + d)/(n \size{\Ncal_i})} + \sigma_0\sqrt{\epsilon / n}$, $\mu = \gbf_i^\ast$ and $\sigma_0 = \norm{\Sigma_i(\btheta_{j, 0})}_{\mathrm{op}}$.
Then by Theorem 1.3 in \citet{diakonikolas2020stability}, we have $\norm{\widehat{\gbf}_i' - \gbf_i^\ast} \le \mathcal{O}(\delta) = o(1)$.

To show the stable property of $\widehat{\gbf}_i$, we consider a perturbed set $\set{\xbf_j}$ of $\set{\xbf_j'}$ such that $\max_{j \in \Ncal_i}\|\xbf_j' - \xbf_j\| \le b$.
Note that the two stability conditions in Definition~\ref{def:stable} also hold for $\set{\xbf_j}$ with only the term $\delta$ for $\set{\xbf_j'}$ changes to $\delta + b$ for $\set{\xbf_j}$.
Since $\set{\gbf_{j} (\btheta_{i, 0}; \xi_{j, 0}^{(1)})}_{j \in \Ncal_i}$ differs from $\set{\xbf_j'}_{j \in \Ncal_i}$ by at most $\epsilon \size{\Ncal_i}$ points, we can construct the set $\set{\xbf_j}$, which differs from $\set{\gbf_{j} (\btheta_{j, 0}; \xi_{j, 0}^{(1)})}_{j \in \Ncal_i}$ by at most $\epsilon \size{\Ncal_i}$ points.
Therefore, it holds that $\norm{\widehat{\gbf}_i - \gbf_i^\ast} \le \mathcal{O}(\delta + b)$.

\noindent\textbf{Finally.}
Combining the results from Eq. \eqref{eq:proxy_error} and Eq. \eqref{eq:stability} via the triangle inequality:
\begin{align*}
\norm{\widehat{\gbf}_i - \gbf_i^\ast} &\le \norm{\widehat{\gbf}_i - \widehat{\gbf}_i'} + \norm{\widehat{\gbf}_i' - \gbf_i^\ast} \\
&\le \mathcal{O}\biggl(\frac{L}{\sqrt{k_0^{1-\delta}}}\biggr) + \mathcal{O}\left(\sqrt{\frac{\sigma_0^2}{n \cdot |\Gcal_i|}} + \textrm{err}_{\textrm{rob}}(\epsilon)\right).
\end{align*}
As $n \to \infty$ and $k_0 \to \infty$, both terms vanish.
Thus, $\norm{\widehat{\gbf}_i - \gbf_i^\ast} \le o(1)$, which satisfies Condition~\ref{cond:robust_mean_est} with $\delta_{\gbf} = o(1)$.

\section{Proof of Theorem~\ref{thm:fdp_control}}

We begin by stating some fundamental probability inequalities that will be used throughout the proof.

\begin{lemma}[Bennett's inequality]\label{lem:bennett}
  Let $\{X_i\}_{i=1}^n$ be independent random variables.
  Assume that $X_i - \Ebb X_i \le K$ almost surely for every $i$.
  Then, for any $c > 0$, we have
  \begin{equation*}
    \Pbb\biggl( \sum\limits_{i=1}^n (X_i - \Ebb X_i) \ge c \biggr) \le \exp\biggl(-\frac{\sigma_0^2}{K^2} h_{\text{b}}\biggl(\frac{Kc}{\sigma_0^2}\biggr) \biggr),
  \end{equation*}
  where $\sigma_0^2 = \sum_{i=1}^n \Var(X_i)$ and $h_{\text{b}}(u) = (1 + u) \log(1 + u) - u$.
  Additionally, if $\abs{X_i - \Ebb X_i} \le K$ almost surely for every $i$, then
  \begin{equation*}
    \Pbb\biggl( \biggabs{\sum\limits_{i=1}^n (X_i - \Ebb X_i)} \ge c \biggr) \le 2 \exp\biggl(-\frac{\sigma_0^2}{K^2} h_{\text{b}}\biggl(\frac{Kc}{\sigma_0^2}\biggr) \biggr).
  \end{equation*}
\end{lemma}

\begin{lemma}[Berry-Esseen Inequality \citep{petrov2002probabilities}]\label{lem:Berry-Esseen}
  Suppose that $\{X_i\}_{i=1}^n$ are independent random variables with mean zero, satisfying $\Ebb[\abs{X_j}^{q}] < \infty$ for some $q > 2$.
  Denote $\kappa = \min(1, q - 2)$.
  Let $B_n = \sum_{i=1}^n \Ebb X_i^2$ and
  \[
    L_n = B_n^{-1 - \frac{\kappa}{2}} \sum\limits_{i = 1}^n \Ebb \abs{X_i}^{2+\kappa}.
  \]
  There exists a universal constant $A > 0$ such that
  \begin{equation}
    \max_{-\infty < x < \infty} \left\lvert F_n(x) - \Phi(x) \right\rvert \le A L_n,
  \end{equation}
  where $\Phi(\cdot)$ is the distribution function of the standard Gaussian distribution and $F_n(x)$ is the distribution function of the normalized summation, i.e., $F_n(x) \triangleq \Pbb[B_n^{-1/2} \sum_{i=1}^n X_i \le x]$.
  When $\{X_i\}_{i=1}^n$ are identically distributed with $\Ebb X_1^2 = \sigma_0^2$ and $\Ebb\abs{X_1}^{2 + \kappa} = \gamma^{2 + \kappa}$, we have $L_n = \frac{\gamma^{2+\kappa}}{\sigma_0^{2+\kappa} n^{\kappa / 2}}$.
\end{lemma}

\begin{lemma}[Moderate deviation \citep{petrov2002probabilities}]\label{lem:mod_dev}
  Under the conditions in Lemma~\ref{lem:Berry-Esseen}, for any constant $0 < a < 1$ and $0 \le x \le a (2 \log\frac{1}{L_n})^{1/2}$,
  \begin{equation}
    \biggabs{\frac{1 - F_n(x)}{1 - \Phi(x)} - 1} \le C L_n^{1-a^2} \Bigl(\log\frac{1}{L_n}\Bigr)^{\frac{1}{2}},
  \end{equation}
  and
  \begin{equation}
    \biggabs{\frac{F_n(-x)}{\Phi(-x)} - 1} \le C L_n^{1-a^2} \Bigl(\log\frac{1}{L_n}\Bigr)^{\frac{1}{2}},
  \end{equation}
  where $C = 4 \sqrt{\pi} a A$.
\end{lemma}

Equipped with these lemmas, we proceed to the proof of Theorem~\ref{thm:fdp_control}. We focus our analysis on a fixed normal machine $\mathcal{M}_i$ ($i \in \mathcal{G}$) and its neighbors.
Define $\xi^{(h)}_{0}:= \cup_{i \in \Vcal} \{\xi_{i, 0}^{(h)}\}$ for $h=1, 2$ as the aggregate set of samples across all nodes for each split.
For notational convenience, we define the conditional tail probabilities as:
\begin{align*}
  F_{S, +}(t) &= \Ebb\biggl[\frac{1}{|\Gcal_i|} \sum\limits_{j \in \Gcal_i} \id\{S_{i, j} > t\} \biggm| \xi^{(1)}_{0}\biggr], \\
  F_{S, -}(t) &= \Ebb\biggl[\frac{1}{|\Gcal_i|} \sum\limits_{j \in \Gcal_i} \id\{S_{i, j} < -t\} \biggm| \xi^{(1)}_{0}\biggr].
\end{align*}
For a fixed $\btheta$, define $\bm{\Sigma}_{j}(\btheta) = \mathrm{Cov}_{\sbf_j \sim \mathcal{P}_j}\{\nabla \ell(\btheta; \sbf_j)\}$. For brevity, we denote $\bm{\Sigma}_{j} = \bm{\Sigma}_{j}(\btheta_{j, 0})$.
Furthermore, let $\ubf_j = {\bm \Omega} (\gbf_{j} (\btheta_{j, 0}; \xi_{j, 0}^{(1)}) - \widehat{\gbf}_j)$ and $\ubf_j^\ast = {\bm \Omega} (\gbf_{j} (\btheta_{j, 0}; \xi_{j, 0}^{(1)}) - \gbf_j^\ast)$.
We observe that for $j \in \Gcal$,
\begin{align*}
  \Ebb\bigl(\id\{S_{i,j} \le - t\} \bigm\vert \xi^{(1)}_{0}\bigr)
  &= \Ebb\bigl(\id\{\ubf_j^\top (\gbf_{j} (\btheta_{j, 0}; \xi_{j, 0}^{(2)}) - \gbf_{j}^\ast) \le - t - \ubf_j^\top (\gbf_j^\ast - \widehat{\gbf}_{i})\} \bigm\vert \xi^{(1)}_{0}\bigr) \\
  &\le C_{\kappa} n^{-\frac{\kappa}{2}} + \Phi\biggl(- \frac{\sqrt{n} \{t + \ubf_j^\top (\gbf_j^\ast - \widehat{\gbf}_{i})\}}{(\ubf_j^\top \bm{\Sigma}_j \ubf_j)^{\frac{1}{2}}}\biggr).
\end{align*}

\noindent\textbf{Part (i)}

Under Condition~\ref{cond:robust_mean_est}, we choose the threshold $t_1$ and the auxiliary term $\tilde{t}_1$ as follows:
\begin{align*}
  t_1 &= \frac{\tilde{t}_1}{\sqrt{n}} + \sup_{j \in \Gcal} \abs{\ubf_j^\top (\gbf_{j}^\ast - \widehat{\gbf}_{i})}, \\
  \tilde{t}_1 &= \sup_{j \in \Gcal} \bigl\{(\ubf_j^{\ast\top} \bm{\Sigma}_j \ubf_j^{\ast})^{\frac{1}{2}} + \delta_{\gbf}\bigr\} \{2 \kappa \log(n)\}^{\frac{1}{2}} \ge \sup_{j \in \Gcal} \{2 \kappa \ubf_j^\top \bm{\Sigma}_j \ubf_j \times \log(n)\}^{\frac{1}{2}}.
\end{align*}
With probability at least $1 - \tau_{\gbf}$, for any $j \in \Gcal$,
\begin{equation}\label{equ:single_good_lower}
  \Ebb\bigl(\id\set{S_{i,j} \le - t_1} \mid \xi^{(1)}_{0}\bigr) \le C_{\kappa} n^{-\frac{\kappa}{2}} + \exp\biggl(- \frac{\tilde{t}_1^2}{2(\ubf_j^\top \bm{\Sigma}_j \ubf_j)}\biggr) \le (C_{\kappa} + 1) n^{-\frac{\kappa}{2}},
\end{equation}
and similarly,
\begin{equation}\label{equ:single_good_upper}
  \Ebb\bigl(\id\set{S_{i,j} \ge t_1} \mid \xi^{(1)}_{0}\bigr) \le (C_{\kappa} + 1) n^{-\frac{\kappa}{2}}.
\end{equation}

Next, we analyze the properties of the Byzantine machines in $\Bcal$.
For any $j \in \Bcal$,
\begin{align*}
  \Ebb\bigl(\id\set{S_{i,j} < t_1} \mid \xi^{(1)}_{0}\bigr)
  &= \Pbb\bigl(\ubf_j^\top (\gbf_{j} (\btheta_{j, 0}; \xi_{j, 0}^{(2)}) - \gbf_{j}^\ast) < t_1 - \ubf_j^\top (\gbf_{j}^\ast - \widehat{\gbf}_{i}) \mid \xi^{(1)}_{0}\bigr) \\
  &\le C_{\kappa} n^{-\frac{\kappa}{2}} + \Phi\biggl(\frac{\sqrt{n}\{t_1 - \ubf_j^\top (\gbf_j^\ast - \widehat{\gbf}_{i})\}}{(\ubf_j^\top \bm{\Sigma}_j \ubf_j)^{\frac{1}{2}}}\biggr).
\end{align*}
By the choice of $t_1$ and Conditions~\ref{cond:robust_mean_est}--\ref{cond:signals}, with probability at least $1 - 2\tau_{\gbf}$, for any $j \in \Bcal_i$, we have $t_1 - \ubf_j^\top (\gbf_j^\ast - \widehat{\gbf}_{i}) \le - t_1$ so that:
\begin{equation}\label{equ:detectProb}
  \Ebb\bigl(\id\set{S_{i,j} < t_1} \mid \xi^{(1)}_{0}\bigr) \le (C_{\kappa} + 1) n^{-\frac{\kappa}{2}}.
\end{equation}

By summing over $j \in \Bcal_i$, with probability at least $1 - 2\tau_{\gbf} - (C_{\kappa} + 1) \size{\Bcal_i} n^{-\frac{\kappa}{2}}$,
\begin{equation*}
  \sum\limits_{j \in \Ncal_i} \id\set{S_{i,j}\ge t_1} \ge \sum\limits_{j \in \Bcal_i} \id\set{S_{i,j}\ge t_1} = \size{\Bcal_i}.
\end{equation*}

For the negative part, by Eq.~\eqref{equ:single_good_lower}, with probability at least $1 - 2\tau_{\gbf} - (C_{\kappa} + 1) \size{\Bcal_i} n^{-\frac{\kappa}{2}}$,
\begin{equation*}
  \Ebb\Bigl(\sum\limits_{j \in \Ncal_i} \id\set{S_{i,j}\le -t_1} \mid \xi^{(1)}_{0}\Bigr) = \Ebb\Bigl(\sum\limits_{j \in \Gcal_i} \id\set{S_{i,j}\le -t_1} \mid \xi^{(1)}_{0}\Bigr) \le \size{\Gcal_i} (C_{\kappa} + 1) n^{-\frac{\kappa}{2}}.
\end{equation*}
We discuss the following two cases.
If $n$ is sufficiently large so that $(1 + C_{\kappa})\size{\Gcal_i} n^{-\frac{\kappa}{2}} = n^{-c}$, we have
\begin{equation*}
  \Pbb\Bigl(\sum\limits_{j \in \Gcal_i} \id\set{S_{i,j} \le - t_1} \ge 1\Bigr) \le (1 + C_{\kappa}) \size{\Gcal_i} n^{-\frac{\kappa}{2}} \le n^{-c}.
\end{equation*}
Therefore, with probability at least $1 - 2\tau_{\gbf} - (C_{\kappa} + 1) \size{\Bcal_i} n^{-\frac{\kappa}{2}} - n^{-c}$,
\begin{equation*}
  \sum\limits_{j \in \Gcal_i} \id\set{S_{i,j} \le - t_1} = 0,
\end{equation*}
and
\begin{equation*}
  \frac{\sum\limits_{j \in \Gcal_i} \id\set{S_{i,j} \le - t_1}}{1 \vee \sum\limits_{j \in \Gcal_i} \id\set{S_{i,j} \ge t_1}} = 0 \le \alpha.
\end{equation*}
Otherwise, denote $s = \max\{1, \size{\Gcal_i} (C_{\kappa} + 1) n^{-\frac{\kappa}{2}}\}$ and $\phi_m = s \log n$.
By Lemma~\ref{lem:bennett},
\begin{equation}\label{equ:lower_tail}
  \Pbb\Bigl(\sum\limits_{j \in \Gcal_i} \id\set{S_{i,j} \le - t_1} \ge \size{\Gcal_i} (C_{\kappa} + 1) n^{-\frac{\kappa}{2}} + \phi_m \bigm\vert \xi^{(1)}_{0}\Bigr) \le \exp(- s h_{\text{b}}(\frac{\phi_m}{s})) \le n^{-c}.
\end{equation}
We need a lower bound on the number of Byzantine machines.
By assuming that $\size{\Bcal_i} \ge 2 \alpha^{-1} s \log(n)$, with probability at least $1 - 2\tau_{\gbf} - (C_{\kappa} + 1) \size{\Bcal_i} n^{-\frac{\kappa}{2}} - n^{-c}$,
\begin{equation*}
  \sum\limits_{j \in \Gcal_i} \id\set{S_{i,j} \le - t_1} \le \alpha \size{\Bcal_i},
\end{equation*}
and
\begin{equation*}
  \frac{\sum\limits_{j \in \Gcal_i} \id\set{S_{i,j} \le - t_1}}{1 \vee \sum\limits_{j \in \Gcal_i} \id\set{S_{i,j} \ge t_1}} \le \alpha.
\end{equation*}

In summary, with probability at least $1 - 2\tau_{\gbf} - (C_{\kappa} + 1) \size{\Bcal_i} n^{-\frac{\kappa}{2}} - n^{-c}$, we have the sure-detection property,
\begin{equation*}
  \Bcal_i \subseteq \widehat{\Bcal}_i.
\end{equation*}

\noindent\textbf{Part (ii)}

Next, we consider the FDP control.

By Eq.~\eqref{equ:single_good_lower},
\begin{equation}\label{equ:lower_tail_expect}
        \sum\limits_{j \in \Gcal_i} \Ebb \bigl(\id\{S_{i,j} \le - t_1\} \bigm\vert \xi^{(1)}_{0}\bigr) \le (C_{\kappa} + 1) \size{\Gcal_i} n^{-\frac{\kappa}{2}}.
\end{equation}

By Eq.~\eqref{equ:detectProb},
\begin{equation}\label{equ:lower_tail_byz}
  \sum\limits_{j \in \Bcal_i} \Ebb\bigl(\id\set{S_{i,j} < t_1} \mid \xi^{(1)}_{0}\bigr) \le (C_{\kappa} + 1) \size{\Bcal_i} n^{-\frac{\kappa}{2}}.
\end{equation}

Here we set $\phi_m = c_{b} \size{\Bcal_i}$ for some small constant $c_{b} > 0$.
By Lemma~\ref{lem:bennett}, with probability at least $1 - 2\tau_{\gbf} - \exp(- (C_{\kappa} + 1) \size{\Bcal_i} n^{-\frac{\kappa}{2}} h_{\text{b}}(\frac{\phi_m}{(C_{\kappa} + 1) \size{\Bcal_i} n^{-\frac{\kappa}{2}}}))
= 1 - 2\tau_{\gbf} - \exp(- c \log n)$ for some constant $c > 0$.
\begin{equation*}
  \sum\limits_{j \in \Bcal_i} \id\set{S_{i,j} < t_1} \le (C_{\kappa} + 1) \size{\Bcal_i} n^{-\frac{\kappa}{2}} + \phi_m = (c_{b} + (C_{\kappa} + 1) n^{-\frac{\kappa}{2}}) \size{\Bcal_i} \le 2 c_{b} \size{\Bcal_i}.
\end{equation*}

We define $t_2$ as:
\begin{equation}\label{equ:lower_tail_t2_expect}
    t_2 = F_{S,-}^{-1}\biggl(\frac{\alpha \size{\Bcal_i} - 2(1+\alpha)c_{b} \size{\Bcal_i}}{2\size{\Gcal_i}}\biggr),
\end{equation}
which means that $\Ebb\{\sum_{j \in \Gcal}\id\{S_{i,j} \le - t_2\} \mid \xi^{(1)}_{0}\} = \{\alpha \size{\Bcal_i} - 2(1+\alpha)c_{b} \size{\Bcal_i}\} / 2$.
Consequently, with probability at least $1 - \exp(- 0.193 [\alpha \size{\Bcal_i} - 2(1+\alpha)c_{b} \size{\Bcal_i}]) = 1 - \exp(-C \size{\Bcal_i})$,
\begin{equation}\label{equ:lower_tail_t2}
  \sum\limits_{j \in \Gcal} \id\{S_{i,j} \le -t_2\} \le \alpha \size{\Bcal_i} - 2(1+\alpha)c_{b} \size{\Bcal_i}.
\end{equation}
Combining Eq.~\eqref{equ:lower_tail_expect} with Eq.~\eqref{equ:lower_tail_t2_expect}, with the choice that $\size{\Gcal_i} (C_{\kappa} + 1) n^{-\frac{\kappa}{2}} \le \{\alpha \size{\Bcal_i} - 2(1+\alpha)c_{b} \size{\Bcal_i} \} / 2$ (which implies that $\size{\Bcal_i} \gtrsim \size{\Gcal_i} n^{-\frac{\kappa}{2}}$), we have $t_2 \le t_1$.
In summary, with probability at least $1 - 2\tau_{\gbf} - \exp(- c \log n)$,
\begin{equation*}
  \sum\limits_{j \in \Ncal} \id\{S_{i,j} \ge t_2\} \ge \sum\limits_{j \in \Ncal} \id\{S_{i,j} \ge t_1\} \ge (1 - 2 c_{b}) \size{\Bcal_i}.
\end{equation*}
By Eq.~\eqref{equ:lower_tail_t2}, we also have with probability at least $1 - 2\tau_{\gbf} - \exp(- c \log n) - \exp(-C\size{\Bcal_i})$,
\begin{equation*}
  \sum\limits_{j \in \Ncal} \id\{S_{i,j} \le -t_2\} \le (\alpha + 2 c_{b}) \size{\Bcal_i} - 2(1+\alpha)c_{b} \size{\Bcal_i}.
\end{equation*}
Combining the above two inequalities, we have
\begin{equation*}
  \frac{\sum\limits_{j \in \Ncal} \id\{S_{i,j} \le -t_2\}}{1 \vee \sum\limits_{j \in \Ncal} \id\{S_{i,j} \ge t_2\}} \le \alpha.
\end{equation*}
It means that $R_i \le t_2$ with probability at least $1 - 2\tau_{\gbf} - \exp(- c \log n) - \exp(-C\size{\Bcal_i})$.
Given the upper bound of $R_i$, the symmetry of the test statistics $\{S_{i,j}\}$ can be analyzed using the following two lemmas.

\begin{lemma}\label{lem:population_symmetric}
    Let $r_n = n^{-\frac{(1 - a^2) \kappa}{2}} \sqrt{\log n}$ where $0 < a < 1$.
    With probability at least $1 - \tau_{\gbf}$, uniformly for $0 \le t \le F_{S,-}^{-1}(\size{\Gcal_i}^{-1})$,
    \begin{equation*}
      \biggabs{\frac{F_{S,+}(t)}{F_{S,-}(t)} - 1} \lesssim r_n + \size{\Gcal_i} n^{-\frac{a^2 \kappa}{2}} + \delta_{\gbf}.
    \end{equation*}
    Similarly, with probability at least $1 - \tau_{\gbf}$, uniformly for $0 \le t \le F_{S,+}^{-1}(\size{\Gcal_i}^{-1})$,
    \begin{equation*}
      \biggabs{\frac{F_{S,-}(t)}{F_{S,+}(t)} - 1} \lesssim r_n + \size{\Gcal_i} n^{-\frac{a^2 \kappa}{2}} + \delta_{\gbf}.
    \end{equation*}
\end{lemma}

\begin{lemma}\label{lem:empirical_err}
    Let $1 < v < m$ be sufficiently large.
    We have with probability at least $1 - 3 \exp(- C_1 v^{\frac{1}{3}})$,
    \begin{equation*}
      \sup_{0 \le t \le F_{S,+}^{-1}(v / \size{\Gcal_i})} \biggabs{\frac{\sum\limits_{j \in \Gcal} \id\{S_{i,j} \ge t\}}{\size{\Gcal_i} F_{S,+}(t)} - 1} \le C_2 v^{-\frac{1}{3}},
    \end{equation*}
    \begin{equation*}
      \sup_{0 \le t \le F_{S,-}^{-1}(v / \size{\Gcal_i})} \biggabs{\frac{\sum\limits_{j \in \Gcal} \id\{S_{i,j} \le -t\}}{\size{\Gcal_i} F_{S,-}(t)} - 1} \le C_2 v^{-\frac{1}{3}}.
    \end{equation*}
\end{lemma}

By applying these two lemmas with $v = \frac{\alpha \size{\Bcal_i} - 2(1+\alpha)c_{b} \size{\Bcal_i}}{2} \ge \frac{\alpha \size{\Bcal_i}}{4}$ where $c_{b}$ is sufficiently small, with probability at least $1 - 2\tau_{\gbf} - \exp(- c \log n) - \exp(-C\size{\Bcal_i}) - 3 \exp(- C (\alpha \size{\Bcal_i})^{\frac{1}{3}})$, we have:
\begin{equation*}
  H_{m,n} =  \frac{\sum\limits_{j \in \Gcal_i} \id(S_{i,j} \ge R_i)}{\sum\limits_{j \in \Gcal_i} \id(S_{i,j} \le -R_i)} = 1 + \mathcal{O}((\alpha \size{\Bcal_i})^{-\frac{1}{3}} + r_n + m n^{-\frac{a^2 \kappa}{2}} + \delta_{\gbf}).
\end{equation*}
Therefore, with the same probability,
\begin{equation*}
  \text{FDP} \le \alpha \bigl(1 + \mathcal{O}((\alpha \size{\Bcal_i})^{-\frac{1}{3}}  + r_n + m n^{-\frac{a^2 \kappa}{2}} + \delta_{\gbf}) \bigr).
\end{equation*}

Note that the above FDP control result requires $\size{\Bcal_i}$ to be large.

\noindent\textbf{Part (iii)}

Here we provide a high probability control of the number of false discoveries when $\size{\Bcal_i} \le b m_f$ with some constant $0< b < \alpha^{-1} - 1$ (for example, $b = (\alpha^{-1} - 1) /2$) and the parameter $m_f$ be sufficiently large.
We will show that with high probability, the number of false discovery nodes is controlled by $\mathcal{O}(m_f)$.

If $R_i \ge F_{S, +}^{-1}(m_f / \size{\Gcal_i}) := t_3$, by Lemma~\ref{lem:empirical_err}, with probability at least $1 - 3 \exp(-C m_f^{\frac{1}{3}})$,
\begin{equation*}
\sum\limits_{j \in \Gcal_i} \id\{S_{i, j} \ge R_i\} \le \sum\limits_{j \in \Gcal_i} \id\{S_{i, j} \ge t_3\} \le m_f + m_f^{\frac{2}{3}} \le 2 m_f.
\end{equation*}

Otherwise $R_i < F_{S, +}^{-1}(m_f / \size{\Gcal_i})$ and $\size{\Gcal_i} F_{S, +}(R_i) > m_f$.
By Lemma~\ref{lem:population_symmetric}, $F_{S, -}(R_i) \ge (1 - r_n - \size{\Gcal_i} n^{-\frac{a^2 \kappa}{2}} - \delta_{\gbf}) F_{S, +}(R_i) := (1 - r_n') F_{S, +}(R_i)$ where $r_n' = r_n + \size{\Gcal_i} n^{-\frac{a^2 \kappa}{2}} + \delta_{\gbf} = o(1)$.
We apply Lemma~\ref{lem:empirical_err} again, with probability at least $1 - 3 \exp(-C_1 m_f^{\frac{1}{3}})$,
\begin{equation*}
  \sum\limits_{j \in \Ncal_i} \id\{S_{i, j} \ge R_i\} \le \size{\Bcal_i} + (m_f + C_2 m_f^{\frac{2}{3}}) \le (b + 1 + C_2 m_f^{-\frac{1}{3}}) m_f,
\end{equation*}
and
\begin{align*}
  \sum\limits_{j \in \Ncal_i} \id\{S_{i, j} \le - R_i\} \ge (1 + C_2 m_f^{-\frac{1}{3}}) \size{\Gcal_i} F_{S, -}(R_i) \ge (1 - C_2 m_f^{-\frac{1}{3}}) (1 - r_n') m_f
\end{align*}
Therefore,
\begin{equation*}
  \frac{\sum\limits_{j \in \Ncal_i} \id\{S_{i, j} \le - R_i\}}{\sum\limits_{j \in \Ncal_i} \id\{S_{i, j} \ge R_i\}} \ge \frac{(1 - C_2 m_f^{-\frac{1}{3}}) (1 - r_n')}{b + 1 + C_2 m_f^{-\frac{1}{3}}} > \alpha,
\end{equation*}
with sufficiently large $m_f$ and $n$.
It contradicts against the definition of $R_i$.

In summary, with probability at least $1 - \tau_{\gbf} - 3 \exp(-C_1 m_f^{\frac{1}{3}})$, the false discovery number $\sum_{j \in \Gcal_i} \id\{S_{i, j} \ge R_i\} \le 2 m_f$.

\begin{proof}[Proof of Lemma~\ref{lem:population_symmetric}]
  We consider the tail behavior of the statistics $\{S_{i,j}\}$ via a case-by-case analysis.

  (1) First, we apply Lemma~\ref{lem:mod_dev}.
  For $0 \le t \le a (2 n^{-1} \log(L_n^{-1}) \ubf_j^\top \Sigma_j \ubf_j)^{1/2} + \ubf_j^\top (\gbf_j^\ast - \widehat{\gbf}_{i})$, we have:
  \begin{align*}
    \Pbb(S_{i,j} \ge t \mid \xi^{(1)}_{0})
    &= \Pbb\bigl(\ubf_j^\top(\gbf_{j} (\btheta_{j, 0}; \xi_{j, 0}^{(2)}) - \gbf_j^\ast) \ge t - \ubf_j^\top (\gbf_j^\ast - \widehat{\gbf}_{i}) \mid \xi^{(1)}_{0}\bigr) \\
    &= \Pbb\Biggl(\frac{\sqrt{n}\ubf_j^\top (\gbf_{j} (\btheta_{j, 0}; \xi_{j, 0}^{(2)}) - \gbf_j^\ast)}{(\ubf_j^\top \Sigma_j \ubf_j)^{\frac{1}{2}}} \ge \frac{\sqrt{n} \{t - \ubf_j^\top (\gbf_j^\ast - \widehat{\gbf}_{i})\}}{(\ubf_j^\top \Sigma_j \ubf_j)^{\frac{1}{2}}} \Biggm\vert \xi^{(1)}_{0}\Biggr) \\
    &= \overline{\Phi}\Biggl(\frac{\sqrt{n} \{t - \ubf_j^\top (\gbf_j^\ast - \widehat{\gbf}_{i})\}}{(\ubf_j^\top \Sigma_j \ubf_j)^{\frac{1}{2}}}\Biggr) (1 + \mathcal{O}(r_n)) \\
    &= \overline{\Phi}\Biggl(\frac{\sqrt{n} t}{(\ubf_j^\top \Sigma_j \ubf_j)^{\frac{1}{2}}}\Biggr) (1 + \mathcal{O}(r_n + \sqrt{n} \delta_{\gbf})).
  \end{align*}
  Similarly, for $0 \le t \le a \{2 n^{-1} \log(L_n^{-1}) \ubf_j^\top \Sigma_j \ubf_j\}^{1/2} - \ubf_j^\top (\gbf_j^\ast - \widehat{\gbf}_{i})$, we obtain:
  \[
  \Pbb(S_{i,j} \le - t \mid \xi^{(1)}_{0}) = \Phi\Biggl(-\frac{\sqrt{n} t}{(\ubf_j^\top \Sigma_j \ubf_j)^{\frac{1}{2}}}\Biggr) (1 + \mathcal{O}(r_n + \sqrt{n} \delta_{\gbf})).
  \]

  (2) Next, for $t \ge a \{2 n^{-1} \log(L_n^{-1}) \ubf_j^\top \Sigma_j \ubf_j\}^{1/2} + \abs{\ubf_j^\top (\gbf_j^\ast - \widehat{\gbf}_{i})}$, Lemma~\ref{lem:Berry-Esseen} implies:
  \[
  \max\bigl\{\Pbb(S_{i,j} \ge t \mid \xi^{(1)}_{0}),\Pbb(S_{i,j} \le -t \mid \xi^{(1)}_{0})\bigr\}  \lesssim n^{-\frac{a^2 \kappa}{2}}.
  \]

  Now, we divide the set $\Gcal$ into two subsets for a fixed $t > 0$.
  Let
  \[
  \Gcal_1 = \set{j \in \Gcal: t \le a \{2 n^{-1} \log(L_n^{-1}) \ubf_j^\top \Sigma_j \ubf_j\}^{1/2} + \abs{\ubf_j^\top (\gbf_j^\ast - \widehat{\gbf}_{i})}}
  \]
  and $\Gcal_2 = \Gcal \setminus \Gcal_1$.
  We have for $0 \le t \le F_{S,-}^{-1}(\size{\Gcal}^{-1})$:
  \begin{align*}
    \biggabs{\frac{F_{S,+}(t)}{F_{S,-}(t)} - 1}
    &= \frac{\sum\limits_{j \in \Gcal} \{\Pbb(S_{i,j} \ge t \mid \xi^{(1)}_{0}) - \Pbb(S_{i,j} \le -t \mid \xi^{(1)}_{0})\}}{\sum\limits_{j \in \Gcal} \Pbb(S_{i,j} \le -t \mid \xi^{(1)}_{0})} \\
    &\le \frac{\sum\limits_{j \in \Gcal_1} \{\Pbb(S_{i,j} \ge t \mid \xi^{(1)}_{0}) - \Pbb(S_{i,j} \le -t \mid \xi^{(1)}_{0})\}}{\sum\limits_{j \in \Gcal} \Pbb(S_{i,j} \le -t \mid \xi^{(1)}_{0})} \\
    &\quad + \frac{\sum\limits_{j \in \Gcal_2} \{\Pbb(S_{i,j} \ge t \mid \xi^{(1)}_{0}) - \Pbb(S_{i,j} \le -t \mid \xi^{(1)}_{0})\}}{\sum\limits_{j \in \Gcal} \Pbb(S_{i,j} \le -t \mid \xi^{(1)}_{0})} \\
    &\lesssim r_n + \sqrt{n} \delta_{\gbf} + \size{\Gcal} n^{-\frac{a^2 \kappa}{2}} \\
    &\le r_n + \sqrt{n} \delta_{\gbf} + m n^{-\frac{a^2 \kappa}{2}},
  \end{align*}
  where in the last inequality, we used the property that $\sum_{j \in \Gcal} \Pbb(S_{i,j} \le -t \mid \xi^{(1)}_{0}) = \size{\Gcal} F_{S,-}(t) \ge 1$ for $0 \le t \le F_{S,-}^{-1}(\size{\Gcal}^{-1})$.
\end{proof}

\begin{proof}[Proof of Lemma~\ref{lem:empirical_err}]
To derive the convergence of the supremum, we first consider a fixed $0 \le t \le F_{S, +}^{-1}(v / \size{\Gcal_i})$.
By the definition of $F_{S, +}$, we have the conditional mean $\Ebb[\sum_{j \in \Gcal_i} \id\{S_{i, j} \ge t\} \vert \xi^{(1)}_{0}] = \size{\Gcal_i} F_{S, +}(t) \ge v$ and the conditional variance $\Var[\sum_{j \in \Gcal_i} \id\{S_{i, j} \ge t\} \vert \xi^{(1)}_{0}] \le \size{\Gcal_i} F_{S, +}(t)$.
By Lemma~\ref{lem:bennett},
\begin{equation}\label{equ:empirical_err_fixed_t}
  \Pbb\Bigl( \Bigabs{ \sum\limits_{j \in \Gcal_i} \id\{S_{i, j} \ge t\} - \size{\Gcal_i} F_{S, +}(t) } \ge u \size{\Gcal_i} F_{S, +}(t) \,\Big|\, \xi^{(1)}_{0}\Bigr) \le 2 \exp(- \size{\Gcal_i} F_{S, +}(t) h_{\text{b}}(u)).
\end{equation}

Choose $\{t_i = F_{S, +}^{-1}(\frac{v (1 + \xi)^i}{\size{\Gcal_i}})\}_{i=0}^{i_{\max}}$ with $i_{\max} = \lfloor \frac{\log(\size{\Gcal_i} / v)}{\log(1 + \xi)} \rfloor$ and some sufficiently small constant $\xi > 0$ which will be determined later.
By Eq.~(\ref{equ:empirical_err_fixed_t}), we have
\begin{equation}\label{equ:empirical_err_unibound}
  \begin{aligned}
  & \Pbb\biggl(\sup_{0 \le i \le i_{\max}} \bigl\{\size{\Gcal_i} F_{S, +}(t)\bigr\}^{-1} \biggabs{\sum\limits_{j \in \Gcal_i} \id\{S_{i, j} \ge t\} - \size{\Gcal_i} F_{S, +}(t)} \ge u \,\Big\vert\, \xi^{(1)}_{0}\biggr) \\
  \le & 2 \sum\limits_{i=0}^{i_{\max}} \exp(- v (1 + \xi)^i h_{\text{b}}(u)) \le 2 \sum\limits_{i=0}^{i_{\max}} \exp(- v (1 + \xi i) h_{\text{b}}(u)) \\
  \le & 2 \exp(- v h_{\text{b}}(u)) \times \{1 - \exp(- v \xi h_{\text{b}}(u))\}^{-1} \le 3 \exp(- v u^2 / 2).
\end{aligned}
\end{equation}
The last inequality holds because $h_{\text{b}}(u) \ge \frac{u^2}{2(1 + u / 3)}$ for $u \ge 0$, and $v \xi h_{\text{b}}(u)$ is sufficiently large so that $\exp(- v \xi h_{\text{b}}(u)) \le 1/3$.

Note that $F_{S, +}(t_i) / F_{S, +}(t_{i + 1}) = 1 / (1 + \xi) = 1 - \xi/(1+\xi)$, by choosing $u = \xi = C v^{-1/3}$, Eq.~(\ref{equ:empirical_err_unibound}) implies that with probability at least $1 - 3 \exp(- C v^{1/3} / 2)$,
\begin{equation*}
      \sup_{0 \le t \le F_{S, +}^{-1}(v / \size{\Gcal_i})} \Bigabs{\{\size{\Gcal_i} F_{S, +}(t)\}^{-1} \sum\limits_{j \in \Gcal_i} \id\{S_{i, j} \ge t\} - 1} \lesssim v^{-1/3}.
\end{equation*}
The second inequality of Lemma~\ref{lem:empirical_err} holds similarly.
\end{proof}

\section{Proof of Proposition \ref{prop:connected} and Theorem~\ref{thm:connectivity}}\label{sec:app2}

\begin{lemma}[Theorem 4.1 of \citet{frieze2015introduction}]
Let $\mathtt{G}(m,p)$ be an undirected Erdős–Rényi random graph with $m$ nodes and probability $p$, and $c(m, p):= mp-\log(m)$.
If $\lim_{m\to \infty} c(m, p) =c$, then
\begin{equation*}
\lim\limits_{m \to \infty}\mathbb{P}(\mathtt{G}(m,p) \text{ is connected.}) = \exp(-\exp(-c)).
\end{equation*}
In particular, when $c= +\infty$,  \begin{equation*}
\lim\limits_{m \to \infty}\mathbb{P}(\mathtt{G}(m,p) \text{ is connected.}) = 1.
\end{equation*}
\end{lemma}

\begin{lemma}[Theorem 11.9 of \citet{frieze2015introduction}]\label{lem:di_connected}
Let $\mathtt{D}(m,p)$ be a directed Erdős–Rényi random graph with $m$ nodes and probability $p$, and $c(m, p):= mp-\log(m)$.
If $\lim_{m\to \infty} c(m, p) =c$, then
\begin{equation*}
\lim\limits_{m \to \infty}\mathbb{P}(\mathtt{D}(m,p) \text{ is connected}) = \exp(-2\exp(-c)).
\end{equation*}
In particular, when $c= +\infty$,  \begin{equation*}
\lim\limits_{m \to \infty}\mathbb{P}(\mathtt{D}(m,p) \text{ is connected}) = 1.
\end{equation*}
\end{lemma}

\begin{proof}[\textbf{Proof of Proposition \ref{prop:connected}}]
1. Decompose $X_1$ as
\begin{equation*}
 X_1 = \sum\limits_{v \in \mathcal{V}} I_{v},
\end{equation*}
where
\begin{equation*}
I_{v} = \begin{cases}
 1, & \text{if } v \text{ is an isolated node};\\
 0, & \text{ otherwise}.
\end{cases}
\end{equation*}
Hence,
\begin{equation*}
\begin{aligned}
\mathbb{E}(X_1) & = \sum\limits_{v \in \mathcal{V}} \mathbb{E} I_{v} = m(1-p)^{m-1} = m \exp\left((m-1)\log(1-p)\right)\\
& \leq m \exp\left( -(m-1) p\right)= \exp(p) \exp(-c(m,p))\\
& \leq \begin{cases}
(1+o(1)) \exp(-c(m,p)), & \text{ when } \lim\limits_{m\to \infty}\frac{c(m, p)}{\log m}< + \infty;\\
(1+p)\exp(-c(m,p)), & \text{ when } \lim\limits_{m\to \infty}\frac{c(m, p)}{\log m}= + \infty.\\
\end{cases}
\end{aligned}
\end{equation*}

2.
Let $X_{k}$ denote the number of components with $k$ nodes in $\mathtt{G}(m,p)$.
Then,
\begin{equation*}
\begin{aligned}
\mathbb{P}(\mathtt{G}(m,p) \text{ is not connected}) & = \mathbb{P} \biggl( \bigcup_{k=1}^{m/2} \bigl\{ \mathtt{G}(m,p) \text{ has a component of order } k \bigr\} \biggr) \\
& = \mathbb{P} \biggl( \bigcup_{k=1}^{m/2} \{ X_k > 0 \} \biggr),\\
\end{aligned}
\end{equation*}
which implies that
\begin{equation*}
\mathbb{P}(\mathtt{G}(m,p) \text{ is connected}) \geq 1- \mathbb{P}(X_1>0) - \sum\limits_{k=2}^{m/2} \mathbb{P}(X_{k}>0)  \geq  1- \mathbb{E}(X_1) - \sum\limits_{k=2}^{m/2} \mathbb{P}(X_{k}>0).
\end{equation*}
Notice that
\begin{equation*}
\sum\limits_{k=2}^{m/2} \mathbb{P}(X_{k}>0) \leq \sum\limits_{k=2}^{m/2} \mathbb{E}X_k \leq \sum\limits_{k=2}^{m/2} \binom{m}{k} k^{k-2} p^{k-1} (1-p)^{k(m-k)},
\end{equation*}
where the second inequality is followed by (2.10) of \citet{frieze2015introduction}.

Denote  $u_k:=  \binom{m}{k} k^{k-2} p^{k-1} (1-p)^{k(m-k)}$.
For $2\leq k < 8$, we have
\begin{equation*}
\begin{aligned}
u_k & \leq \exp(k) m^k \left( \frac{\log m + c(m,p)}{m} \right)^{k-1} \exp\left(-k(m-8)\frac{\log m + c(m,p)}{m}\right), \\
& \leq \begin{cases}
\left(1 + o(1) \right) \exp({k(\widehat{C}-c(m,p))}) \left( \frac{\log m}{m} \right)^{k-1}, & \text{ when } \lim\limits_{m\to \infty}\frac{c(m, p)}{\log m}< + \infty;\\
\left(1 + \mathcal{O}(\frac{\log m}{m}) \right) \exp({k(1+8p-c(m,p))})  p^{k-1}, & \text{ when } \lim\limits_{m\to \infty}\frac{c(m, p)}{\log m}= + \infty.\\
\end{cases}
\end{aligned}
\end{equation*}
where $\widehat{C}>0$ is a constant.
For $k \geq 8$, we have
\begin{equation*}
\begin{aligned}
u_k  & \leq \left( \frac{m e}{k} \right)^k k^{k-2} \left( \frac{\log m + c(m,p)}{m} \right)^{k-1} \exp({-k \frac{\log m + c(m,p)}{2}}),\\
& \leq \begin{cases} m \left( \exp({1 - c(m,p)/2 + o(1)}) \widehat{C} \log m \right)^k \left( m^{-1/2} \right)^k, & \text{ when } \lim\limits_{m\to \infty}\frac{c(m, p)}{\log m}< + \infty;\\
\exp(k({1 - c(m,p)/2 })) m^{k/2} p^{k-1},& \text{ when } \lim\limits_{m\to \infty}\frac{c(m, p)}{\log m}= + \infty.
\end{cases}
\end{aligned}
\end{equation*}

As a result, whenever $\lim_{m\to \infty}\frac{c(m, p)}{\log m}< + \infty$,  it follows that
\begin{equation*}
\begin{aligned}
\sum\limits_{k=2}^{m/2} u_k  \leq &~ (1 + o(1)) \max\{\exp(2(\widehat{C}-c(m,p))), \exp(7(\widehat{C}-c(m,p)))\}\frac{\log m}{m} + \sum\limits_{k=8}^{m/2} m^{1 + o(1) - k/2}\\
 = &~ \mathcal{O}(m^{o(1) - 1});
\end{aligned}
\end{equation*}
When $\lim_{m\to \infty}\frac{c(m, p)}{\log m}= + \infty$, we obtain
\begin{equation*}
\begin{aligned}
 \sum\limits_{k=2}^{m/2} u_k  \leq &~ \left(1+\mathcal{O}(\frac{\log m}{m})\right) \max\{\exp(2(1+8p-c(m,p))), \exp(7(1+8p-c(m,p)))\}\frac{p}{1-p} \\
 &~ + \exp(8(1-\frac{c(m,p)}{2}+\frac{\log m}{2})) \frac{p}{1-p} \\
 \leq &~ \frac{p}{1-p} \mathcal{O}\left(\exp(-2c(m,p))\right) + \frac{p}{1-p} \exp(8(1-\frac{c(m,p)}{2}+\frac{\log m}{2}))\\
  = &~ \frac{p}{1-p} \mathcal{O}\left(\exp(-4 c(m,p)+4 \log m ) \right) \\
  = &~ \mathcal{O}(m^{- \frac{c(m,p)}{\log m}})
 \end{aligned}
\end{equation*}

Therefore, when $\lim_{m\to \infty}\frac{c(m, p)}{\log m}< + \infty$,
\begin{equation*}
\mathbb{P}(\mathtt{G}(m,p) \text{ is connected}) \geq 1- (1+o(1)) \exp(-c(m,p)) -  \mathcal{O}(m^{o(1) - 1}).
\end{equation*}
When $\lim_{m\to \infty}\frac{c(m, p)}{\log m}= + \infty$,
\begin{equation*}
\mathbb{P}(\mathtt{G}(m,p) \text{ is connected}) \geq 1- (1+p) \exp(-c(m,p)) -  \mathcal{O}(m^{- \frac{c(m,p)}{\log m}}).
\end{equation*}

\end{proof}

By combining Lemma~\ref{lem:di_connected} and the proof of Proposition \ref{prop:connected}, we can also obtain the following corollary, whose proof is omitted here.
\begin{coro}\label{coro:connected}
Let $\mathtt{D}(m,p)$ be a directed Erdős–Rényi random graph with $m$ nodes and probability $p$.
Suppose that $c(m, p):= mp-\log(m)\geq 0$ satisfying $\lim_{m\to \infty} c(m, p) =c$.
Then it holds that
\begin{equation*}
\begin{aligned}
 & \mathbb{P}(\mathtt{D}(m,p) \text{ is strongly connected})\\
 \geq & \begin{cases}
1- (2+o(1)) \exp(-c(m,p)) -  \mathcal{O}(m^{o(1) - 1}), & \text{ when } \lim\limits_{m\to \infty}\frac{c(m, p)}{\log m}< + \infty;\\
 1- (2+p) \exp(-c(m,p)) -  \mathcal{O}(m^{- \frac{c(m,p)}{\log m}}), & \text{ when } \lim\limits_{m\to \infty}\frac{c(m, p)}{\log m}= + \infty.
 \end{cases}\\
\end{aligned}
\end{equation*}

\end{coro}

\begin{proof}[\textbf{Proof of Theorem~\ref{thm:connectivity}}]

Replacing each undirected edge $(u,v)\in \mathcal{E}|_{\mathcal{G}}$ by a pair of opposite arcs $u\to v$ and $v\to u$, $(\mathcal{G}, \mathcal{E}|_{\mathcal{G}})$ is transformed to a bi-directed graph.
For each node $v$, let
\[
\mathrm{In}(v):=\{u\to v:(u,v)\in \mathcal{E}|_{\mathcal{G}}\},\qquad d(v) := \size{\mathrm{In}(v)}.
\]
For any deletion ratio $\beta\in(0,1)$, define the \emph{fixed-budget in-deletion} digraph $\mathtt{D}^{\mathrm{fix}}(\beta)$ by deleting \emph{exactly}
$k_v:=\lfloor \beta\, d(v)\rfloor$ in-arcs from $\mathrm{In}(v)$, sampled uniformly without replacement, independently across nodes $v$ conditional on $(\mathcal{G}, \mathcal{E}|_{\mathcal{G}})$.

Let random keys $\{U_{u,v}\}_{u\in \mathrm{In}(v), v\in \mathcal{G}}  \overset{\text{i.i.d.}}{\sim} \ U[0,1]$ (the uniform distribution on [0, 1]), independent of $(\mathcal{G}, \mathcal{E}|_{\mathcal{G}})$. $\mathtt{D}^{\mathrm{fix}}(\beta)$ can be equivalently defined by: \textit{For each $v$, order $\mathrm{In}(v)$ by $U_{u,v}$ increasingly and delete the $k_v$ in-arcs with the smallest keys.}

The equivalence is owing to that the induced ordering is a uniformly random permutation, which leads to the deleted in-arcs are sampled uniformly without replacement.
Moreover, for any $t\in[0,1]$, define the \emph{threshold in-deletion} digraph $\mathtt{D}^{\mathrm{th}}(t)$ by deleting each in-arc $u\to v$ in $(\mathcal{G}, \mathcal{E}|_{\mathcal{G}})$ if and only if $U_{u,v}\le t$.
Conditional on $(\mathcal{G}, \mathcal{E}|_{\mathcal{G}})$, the deletions are independent across arcs with $\mathbb{P}(u\to v\ \text{deleted}\mid (\mathcal{G}, \mathcal{E}|_{\mathcal{G}}))=t$.

For each $v$, let $T_v$ be the $k_v$-th order statistic among the set $\{U_{u,v}:u\to v\in \mathrm{In}(v)\}$ (with the convention $T_v:=0$ if $k_v=0$).
Then by construction, the deleted set of in-arcs at $v$ in $\mathtt{D}^{\mathrm{fix}}(\beta)$ is
\[
S_v=\{u\to v\in \mathrm{In}(v): U_{u,v}\le T_v\}.
\]

Fix $\delta\in(0,1)$ and consider the event
\[
B(\delta):=\bigcap_{v=1}^{m_g} \{\, T_v\le \beta+\delta\,\}.
\]
On $B(\delta)$ we have, for every $v$,
\[ S_v\subseteq \{U_{u,v}\le \beta+\delta\},
\]
and hence the following \emph{digraph inclusions} hold:
\begin{equation}\label{eq:sandwich}
\mathtt{D}^{\mathrm{th}}(\beta+\delta)\ \subseteq\ \mathtt{D}^{\mathrm{fix}}(\beta).
\end{equation}

Conditional on $(\mathcal{G}, \mathcal{E}|_{\mathcal{G}})$, for each $v$ and any $t\in[0,1]$, define
\[
X_v(t):=\size{\{u\to v\in \mathrm{In}(v): U_{u,v}\le t\}}\sim \mathrm{Bin}(d(v),t),
\]
where $\mathrm{Bin}(m, p)$ is the binomial distribution with parameters $m$ and $p$.
Note that $T_v\le t$ iff $X_v(t)\ge k_v$ and $T_v>t$ iff $X_v(t)\le k_v-1$.
Using Chernoff bounds for binomials, there exists an absolute constant $c_0>0$ such that for every $v$,
\[
\mathbb{P}\big(T_v\notin (0,\beta+\delta]\mid (\mathcal{G}, \mathcal{E}|_{\mathcal{G}})\big)
\ \le\ \exp\big(-c_0\,\delta^2 d(v)\big),
\]
whenever $\delta d(v)\ge 2$ (the floor in $k_v$ only affects constants).
Denote $d_{\min}:=\min_v d(v)\gtrsim mp$.
Then by a union bound,
\begin{equation}\label{eq:union}
\mathbb{P}\big(\mathcal{E}(\delta)^c\mid (\mathcal{G}, \mathcal{E}|_{\mathcal{G}})\big)\ \le\ m\exp\big(-c_0\,\delta^2 d_{\min}\big).
\end{equation}

Let $\delta < 1-\beta$.
We further get
\[
\mathbb{P}\Big( \mathtt{D}^{\mathrm{th}}(\beta+\delta)\subseteq \mathtt{D}^{\mathrm{fix}}(\beta)\Big)\ \ge\ 1-m\exp\big(-c_0\,\delta^2 d_{\min} \mid (\mathcal{G}, \mathcal{E}|_{\mathcal{G}}) \big).
\]

Following the condition (ii) in Theorem~\ref{thm:connectivity},
when $\size{\mathcal{B}_i}\geq \frac{1-\alpha}{2 \alpha} \log(m)^3$, any normal node deletes at most $\beta = \alpha \max\{H_{i,n}\}$-fraction in-arcs.
When $\size{\mathcal{B}_i} < \frac{1-\alpha}{2 \alpha} \log(m)^3$, we have
$$\sum\limits_{j \in \Gcal_i} \id\{S_{i, j} \ge R_i\} \le 2 (\log m)^3 \text{ holds for each } i \in \mathcal{G}. $$
This fact implies any normal node deletes at most $\beta = \frac{2 (\log m)^3}{d_{\min}}$-fraction in-arcs.

For any fixed node $v$ and $\gamma \in (0,1)$, using Chernoff bounds for $d(v)\sim \mathrm{Bin}(m-1,p)$, we have
\begin{equation*}
\mathbb{P}\big(d(v)\le (1-\gamma)(m-1)p\big)\ \le\ \exp\!\Big(-\frac{\gamma^{2}}{2}(m-1)p\Big).
\end{equation*}
Let $\gamma =\frac{1}{2}$.
A union bound can be derived by
\begin{equation*}
\begin{aligned}
\mathbb{P}\big(d_{\min}\le \tfrac12 (n-1)p\big)
 & =\Pbb\Big(\exists v\in[m]:\ d(v)\le \tfrac12 (m-1)p\Big)\\
& \le\ \sum\limits_{v=1}^{m}\Pbb\big(d(v)\le \tfrac12 (m-1)p \big) \\
& \le\ m\exp\!\Big(-\frac{(m-1)p}{8}\Big),
\end{aligned}
\end{equation*}
which also gives
\begin{equation*}
\mathbb{P}\big(d_{\min}>\tfrac12 (m-1)p \big)\ \ge\ 1-m\exp\!\Big(-\frac{(m-1)p}{8}\Big).
\end{equation*}

Hence, by Condition (i) in Theorem~\ref{thm:connectivity},  with probability $1-m\exp(-\frac{(m-1)p}{8})$, the graph $(\mathcal{G}, \mathcal{E}'|_{\mathcal{G}})$ can be decomposed as a $\mathtt{D}^{\mathrm{fix}}(\beta_0)$ combining with some additional arcs added, where $\beta_0:= \max\{\frac{4 (\log m)^3}{(m-1)p}, \alpha \max_i\{H_{i,n}\}\}$.
Besides, $\mathbb{P}( \mathtt{D}^{\mathrm{th}}(\beta+\delta)\subseteq \mathtt{D}^{\mathrm{fix}}(\beta))\ \ge\ 1-m\exp(-c_0\, \frac{1}{2}\delta^2 (m-1)p).$

Then we have
\begin{equation*}
\begin{aligned}
&\mathbb{P}((\mathcal{G}, \mathcal{E}'|_{\mathcal{G}}) \text{ is strongly connected }) \\
\geq &
\left[1-m\exp\!\Big(-\frac{(m-1)p}{8}\Big)\right]
\mathbb{P}\big(\mathtt{D}^{\mathrm{fix}}(\beta_0) \text{ is strongly connected }\big) \\
 \geq & \left[1-m\exp\!\Big(-\frac{(m-1)p}{8}\Big)\right]\mathbb{P}\Big( \mathtt{D}^{\mathrm{th}}(\beta_0+\delta)\subseteq \mathtt{D}^{\mathrm{fix}}(\beta_0, \text{ } \mathtt{D}^{\mathrm{th}}(\beta_0+\delta) \text{ is strongly connected}\Big)\\
\geq & \left[1-m\exp\!\Big(-\frac{(m-1)p}{8}\Big)\right]\left[1-m\exp\big(-\frac{c \delta^2}{2} (m-1)p\big)\right] \mathbb{P}(\mathtt{D}^{\mathrm{th}}(\beta_0+\delta) \text{ is s. c.} )\\
= & \left[1-m\exp\!\Big(-\frac{(m-1)p}{8}\Big)\right]\left[1-m\exp\big(-\frac{c \delta^2}{2} (m-1)p\big)\right] \mathbb{P}(\mathtt{D}(m, p(1-\beta_0-\delta)) \text{ is s. c.})\\
\geq & \begin{cases}
& 1- 3 \exp(-c(m,p, \delta))
-2m \exp(-\frac{(m-1)p}{8}) \\
&~~ -2m\exp(-\,\frac12 c_0 \delta^2 (m-1)p) - \mathcal{O}\!\bigl(m^{-1+o(1)}\bigr), \text{ when } \lim\limits_{m\to \infty}\frac{c(m,p, \delta)}{\log m}< + \infty,\\
& 1 - (2+p(1-\beta_0-\delta)) \exp(-c(m,p, \delta)) -2m\exp(-\frac{(m-1)p}{8})
\\
&~~ -2m\exp(-\,\frac12 c_0 \delta^2 (m-1)p) -  \mathcal{O}(m^{- \frac{c(m,p, \delta)}{\log m}}), \text{ when } \lim\limits_{m\to \infty}\frac{c(m,p, \delta)}{\log m}= + \infty.
 \end{cases}
\end{aligned}
\end{equation*}
Here, ``s. c.'' is the abbreviation of ``strongly connected'', and the last inequality is followed from Corollary~\ref{coro:connected}.

\end{proof}

\section{Proof of Theorem~\ref{thm:main0} and Corollary~\ref{coro:consensus}}

To facilitate analysis, we introduce some notations.
Denote the concatenation of optimization variables, auxiliary variables $\{\ybf_{i, k}\}$ and local costs in normal machines as
\begin{equation*}
{\bm \Theta}_{k}:= \begin{pmatrix}
- & {\btheta_{1, k}}^{\top} &  -\\
- & {\btheta_{2, k}}^{\top} & - \\
& \vdots  & \\
- & {\btheta_{m_g, k}}^{\top} & - \\
\end{pmatrix},
\Ybf_{k}:= \begin{pmatrix}
- & \ybf_{1, k}^{\top} &  -\\
- & \ybf_{2, k}^{\top} & - \\
& \vdots  & \\
- & \ybf_{m_g, k}^{\top} & - \\
\end{pmatrix},
{\fbf}({\bm \Theta}_{k}):= \begin{pmatrix}
f_1(\btheta_{1, k})\\
f_2(\btheta_{2, k})\\
\vdots \\
f_{m_g}(\btheta_{m_g, k}) \\
\end{pmatrix}.
\end{equation*}
Moreover, we define the ${\vbf}_1$-weighted average of optimization variables by $\tilde{\btheta}_k:={\bm \Theta}_{k}^{\top} {\vbf}_1$, and the diagonal matrix of $\Ybf_{k}$ by $\tilde\Ybf_{k}:= \mathrm{Diag}(\Ybf_{k})$.
Let the exact gradient and stochastic gradient of ${\fbf}({\bm \Theta}_k)$ be
\begin{equation*}
\nabla {\fbf}({\bm \Theta}_k):= \begin{pmatrix}
- & {\nabla f_1(\btheta_{1, k})}^{\top} &  -\\
- & \nabla f_2(\btheta_{2, k})^{\top} & - \\
& \vdots  & \\
- & \nabla f_{m_g}(\btheta_{m_g, k})^{\top} & - \\
\end{pmatrix},
\Gbf({\bm \Theta}_k; \Xi_k):= \begin{pmatrix}
- &  \gbf({\btheta}_{1, k}; {\xi}_{1, k})^{\top}  & -\\
- & \gbf({\btheta}_{2, k}; {\xi}_{2, k})^{\top}  & - \\
& \vdots  & \\
- & \gbf({\btheta}_{m_g, k}; {\xi}_{m_g, k})^{\top}  & - \\
\end{pmatrix},
\end{equation*}
where $\Xi_{k}$ represents the collection of random variables for mini-batches $\{\xi_{1,k}, \ldots, \xi_{m_g, k}\}$ over normal nodes at the $k$-th iteration.
Then the compact update scheme of DRSGD algorithm over the strongly connected component $(\mathcal{G}, \mathcal{E}'|_{\mathcal{G}})$  can be rewritten as
\begin{equation*}
\begin{cases}
\Ybf_{k+1} & = \Abf \Ybf_{k},\\
{\bm \Theta}_{k+1}  & = \Abf{\bm \Theta}_{k} - \eta_k \tilde\Ybf_{k}^{-1}  \Gbf({\bm \Theta}_k; {\Xi}_k).
\end{cases}
\end{equation*}

Lemma~\ref{lem:E} reveals the geometric convergence rate of  $\{\Ybf_{k}\}$ and $ \{\frac{{\vbf}_1^{\top} {\tilde\Ybf_k}^{-1}}{m_g}\}$, as well as the uniform boundedness of $\{\tilde\Ybf_{k}^{-1}\}$.

\begin{lemma}\label{lem:E}
Suppose Condition~\ref{asp:sure_detection2} holds.
Let $\Ybf_{\infty}:= \lim_{k\to \infty} \Ybf_k$, $\Ybf_{0}:= \Abf^{t_0}$ for some $t_0 \in \mathbb{N}$, then the following inequalities hold for any $k\in \mathbb{N}$:

\begin{enumerate}[(i)]
\item  $\|\Ybf_k - \Ybf_{\infty}\|_2 \leq c \rho^{k+t_0} \leq c \rho^{k}$.
\item $w:= \sup_{k}\|{\mathrm{Diag}(\Abf^{k})}^{-1}\|_2 < +\infty$.
In particular,  $\sup_{k}\|\tilde\Ybf_{k}^{-1}\|_2\leq w$.
\item  $\|\frac{{\bm 1}^{\top}}{m_g}- \frac{{\vbf}_1^{\top} {\tilde\Ybf_k}^{-1}}{m_g}\|\leq   \frac{ w^2 \|{\vbf}_1\|  \|\Ybf_k - \Ybf_{\infty}\|_2}{m_g}$.
\end{enumerate}

\end{lemma}

\begin{proof}[\textbf{Proof of Lemma~\ref{lem:E}}]

\begin{enumerate}[(i)]
\item Notice that $\Ybf_{k}\in \mathbb{R}^{m_g\times m_g}$ is actually updated by $\Ybf_{k} = \Abf^k \Ybf_0$.
Then, we have
\begin{equation*}
\|\Ybf_k - \Ybf_{\infty}\|_2 = \|\Abf^{k+t_0} - {\bm 1}_{m_g}{\vbf}_1^{\top}\|_2 \leq c \rho^{k+t_0}.
\end{equation*}
\item Since $\Abf$ is an irreducible row-stochastic matrix with positive diagonal elements, we can deduce that the sequence $\{\Abf^{k}\}$ is convergent.
Moreover, this fact implies each diagonal element of $\Abf^{k}$ is nonzero and bounded, which indicates that $e$ is finite.

\item
By Cauchy-Schwarz inequality,
\begin{equation*}
\begin{aligned}
\norm{\frac{{\bm 1}^{\top}}{m_g}- \frac{{\vbf}_1^{\top} {\tilde\Ybf_k}^{-1}}{m_g}} & = \norm{\frac{{\bm 1}^{\top}}{m_g}-\frac{{\vbf}_1^{\top} \tilde\Ybf_{\infty}^{-1}}{m_g} + \frac{{\vbf}_1^{\top} \tilde\Ybf_{\infty}^{-1}}{m_g} - \frac{{\vbf}_1^{\top} {\tilde\Ybf_k}^{-1}}{m_g}}  \\
& = \norm{\frac{{\vbf}_1^{\top} \tilde\Ybf_{\infty}^{-1}}{m_g} - \frac{{\vbf}_1^{\top} {\tilde\Ybf_k}^{-1}}{m_g}}\\
& = \norm{\frac{{\vbf}_1^{\top}}{m_g} \tilde\Ybf_{k}^{-1}(\tilde\Ybf_{\infty}- {\tilde\Ybf_k})\tilde\Ybf_{\infty}^{-1}}\\
& \leq   \frac{ w^2 \|{\vbf}_1\|  \|\Ybf_k - \Ybf_{\infty}\|_2}{m_g}.\\
\end{aligned}
\end{equation*}
\end{enumerate}
\end{proof}

\begin{lemma}\label{lem:square_S_control}
Under the setting in Theorem~\ref{thm:main0},
suppose Condition~\ref{asp:sure_detection2} and Condition~\ref{asp:loss} hold, we have that
\begin{equation*}
\mathbb{E} \left[\left\|\frac{{\vbf}_1^{\top}\tilde\Ybf_k^{-1}\Gbf({\bm \Theta}_{k}; \Xi_k)}{m_g}\right\|^2 \Bigg|\Fcal_k \right] \leq \frac{w^2 \|{\vbf}_1\|^2 \sigma^2}{m_g} + \norm{\frac{{\vbf}_1^{\top}\tilde\Ybf_k^{-1}}{m_g}  \nabla {\fbf}({\bm \Theta}_k)}^2.
\end{equation*}

\end{lemma}

\begin{proof}[\textbf{Proof of Lemma~\ref{lem:square_S_control}}]

Denote the $i$-th element of $\frac{{\vbf}_1^{\top}\tilde\Ybf_k^{-1}}{m_g}$ as $\pi_{i,k}$.
By adding and subtracting the exact local gradients $\sum_{i=1}^{m_g}\pi_{i,k}\nabla f_{i}(\btheta_{i,k})$ inside the norm and expanding the square, we obtain:
\begin{equation*}
\begin{aligned}
& \mathbb{E} \left[\left\|\frac{{\vbf}_1^{\top}\tilde\Ybf_k^{-1}\Gbf({\bm \Theta}_{k}; \Xi_k)}{m_g}\right\|^2 \Bigg|\Fcal_k \right]\\
=  &\mathbb{E}  \left[\left\|\sum\limits_{i=1}^{m_g}\pi_{i,k}\left(\gbf_i({\btheta}_{i, k}; {\xi}_{i, k}) - \nabla f_{i}(\btheta_{i,k})\right)+ \sum\limits_{i=1}^{m_g}\pi_{i,k}\nabla f_{i}(\btheta_{i,k})\right\|^2 \Bigg|\Fcal_k\right]    \\
= & \mathbb{E}  \left[\left\|\sum\limits_{i=1}^{m_g}\pi_{i,k}\left(\gbf_i({\btheta}_{i, k}; {\xi}_{i, k}) - \nabla f_{i}(\btheta_{i,k})\right)\right\|^2\Bigg|\Fcal_k\right] +  \norm{\sum\limits_{i=1}^{m_g}\pi_{i,k}\nabla f_{i}(\btheta_{i,k})}^2 \\
& ~ + 2 \left\langle \sum\limits_{i=1}^{m_g}\pi_{i,k}\mathbb{E}\left[\gbf_i({\btheta}_{i, k}; {\xi}_{i, k}) - \nabla f_{i}(\btheta_{i,k}) \bigg| \Fcal_k \right], \sum\limits_{i=1}^{m_g}\pi_{i,k}\nabla f_{i}(\btheta_{i,k}) \right\rangle.
\end{aligned}
\end{equation*}

Since the mini-batch stochastic gradients are unbiased estimators of the exact local gradients given $\Fcal_k$, i.e., $\mathbb{E}[\gbf_i({\btheta}_{i, k}; {\xi}_{i, k}) \mid \Fcal_k] = \nabla f_{i}(\btheta_{i,k})$, the cross-term vanishes. Thus, the equation simplifies to:
\begin{equation*}
\begin{aligned}
& \mathbb{E} \left[\left\|\frac{{\vbf}_1^{\top}\tilde\Ybf_k^{-1}\Gbf({\bm \Theta}_{k}; \Xi_k)}{m_g}\right\|^2 \Bigg|\Fcal_k \right]\\
= &  \mathbb{E}  \left[\left\|\sum\limits_{i=1}^{m_g}\pi_{i,k}\left(\gbf_i({\btheta}_{i, k}; {\xi}_{i, k}) - \nabla f_{i}(\btheta_{i,k})\right)\right\|^2\Bigg|\Fcal_k\right] +  \norm{\sum\limits_{i=1}^{m_g}\pi_{i,k}\nabla f_{i}(\btheta_{i,k})}^2.
\end{aligned}
\end{equation*}

Finally, applying the Cauchy-Schwarz inequality and the bounded variance condition (Condition~\ref{asp:loss}), we can bound the first term. Substituting the matrix norm upper bound $w$ yields the desired result:
\begin{equation*}
\begin{aligned}
& \mathbb{E} \left[\left\|\frac{{\vbf}_1^{\top}\tilde\Ybf_k^{-1}\Gbf({\bm \Theta}_{k}; \Xi_k)}{m_g}\right\|^2 \Bigg|\Fcal_k \right] \\
\leq & m_g \left\|\frac{{\vbf}_1^{\top}\tilde\Ybf_k^{-1}}{m_g}\right\|^2 \sigma^2 + \left\|\sum\limits_{i=1}^{m_g}\pi_{i,k}\nabla f_{i}(\btheta_{i,k})\right\|^2\\
\leq & \frac{w^2 \|{\vbf}_1\|^2 \sigma^2}{m_g} + \norm{\frac{{\vbf}_1^{\top}\tilde\Ybf_k^{-1}}{m_g} \nabla {\fbf}({\bm \Theta}_k)}^2,
\end{aligned}
\end{equation*}
where we define $\sigma = \max\{\sigma_i: i \in \Gcal\}$ in Theorem~\ref{thm:main0}.

\end{proof}

\begin{lemma}\label{lem:gradient_diff}
Let $\tilde{\btheta}_k:= {\vbf}_1^{\top} {\bm \Theta}_k$, $Q_{i,k}:= \mathbb{E}\left[\norm{\tilde{\btheta}_k -{\btheta}_{i,k}}^2\right]$, and $M_k:= \frac{1}{m_g}\sum_{i=1}^{m_g}Q_{i, k}$.
Under the setting in Theorem~\ref{thm:main0},
suppose Condition~\ref{asp:sure_detection2} and \ref{asp:loss}, we have that
\begin{equation*}
\begin{aligned}
 \mathbb{E}\left[\norm{\nabla f(\tilde{\btheta}_k)- \frac{{\vbf}_1^{\top}\tilde\Ybf_k^{-1}\nabla {\fbf}({\bm \Theta}_k)}{m_g}}^2\right] & \leq
2L^2 M_k + \frac{ 2w^4 \|{\vbf}_1\|^2 \|\Ybf_k - \Ybf_{\infty}\|^2_2}{m_g^2} \mathbb{E}\left[\norm{\nabla {\fbf}({\bm \Theta}_k)}_{\mathrm{F}}^2\right]. \\
\end{aligned}
\end{equation*}
\end{lemma}

\begin{proof}[\textbf{Proof of Lemma~\ref{lem:gradient_diff}}]
By adding and subtracting the average of the exact local gradients $\frac{1}{m_g}\sum_{i=1}^{m_g}\nabla f_{i}(\btheta_{i,k})$, we can decouple the estimation error into two parts.
Applying the basic inequality $\|a+b\|^2 \leq 2\|a\|^2 + 2\|b\|^2$, we have:
\begin{equation*}
\begin{aligned}
& \mathbb{E}\left[\norm{\nabla f(\tilde{\btheta}_k)- \frac{{\vbf}_1^{\top}\tilde\Ybf_k^{-1}\nabla {\fbf}({\bm \Theta}_k)}{m_g}}^2\right] \\
 = & \mathbb{E}\left[\norm{\nabla f(\tilde{\btheta}_k) -\frac{1}{m_g}\sum\limits_{i=1}^{m_g}\nabla f_{i}(\btheta_{i,k})+\frac{1}{m_g}\sum\limits_{i=1}^{m_g}\nabla f_{i}(\btheta_{i,k})-\sum\limits_{i=1}^{m_g} \pi_{i, k} \nabla f_{i}({\btheta}_{i,k})}^2\right]\\
 \leq & 2 \mathbb{E}\left[\norm{\nabla f(\tilde{\btheta}_k) -\frac{1}{m_g}\sum\limits_{i=1}^{m_g}\nabla f_{i}(\btheta_{i,k})}^2 \right] + 2\mathbb{E}\left[\norm{\frac{1}{m_g}\sum\limits_{i=1}^{m_g}\nabla f_{i}(\btheta_{i,k})-\sum\limits_{i=1}^{m_g} \pi_{i, k} \nabla f_{i}({\btheta}_{i,k})}^2\right].
\end{aligned}
\end{equation*}

For the first term, we utilize the definition $\nabla f(\tilde{\btheta}_k) = \frac{1}{m_g}\sum_{i=1}^{m_g}\nabla f_i(\tilde{\btheta}_k)$ and the $L$-smoothness of each $f_i$ owing to Condition~\ref{asp:loss}(ii).
For the second term, we apply Lemma~\ref{lem:E}.
Combining these two bounds, we obtain:
\begin{equation*}
\begin{aligned}
& \mathbb{E}\left[\norm{\nabla f(\tilde{\btheta}_k)- \frac{{\vbf}_1^{\top}\tilde\Ybf_k^{-1}\nabla {\fbf}({\bm \Theta}_k)}{m_g}}^2\right] \\
 \leq &  \frac{2L^2}{m_g}\sum\limits_{i=1}^{m_g} \mathbb{E}\left[\norm{\tilde{\btheta}_k -\btheta_{i,k}}^2 \right] + 2\mathbb{E}\left[\norm{\left(\frac{{\bm 1}_{m_g}}{m_g}- \frac{{\vbf}_1^{\top}\tilde\Ybf_k^{-1}}{m_g}\right)^{\top} \nabla {\fbf}({\bm \Theta}_{k})}^2\right] \\
 \leq & 2L^2 M_k + \frac{ 2w^4 \|{\vbf}_1\|^2 \|\Ybf_k - \Ybf_{\infty}\|^2_2}{m_g^2}\mathbb{E}\left[\norm{\nabla {\fbf}({\bm \Theta}_k)}_{\mathrm{F}}^2\right].
\end{aligned}
\end{equation*}
\end{proof}

\begin{lemma}\label{lem:nabla_f}
We have the following inequality under Condition~\ref{asp:sure_detection2} and \ref{asp:loss}:
\begin{equation}\label{eq:nabla_f}
\mathbb{E}\left[\norm{\nabla {\fbf}({\bm \Theta}_k)}_{\mathrm{F}}^2\right]\leq 3L^2m_g M_k + 3m_g \zeta^2 + 3\mathbb{E}\norm{\nabla f(\tilde{\btheta}_k){\bm 1}_{m_g}^{\top}}_{\mathrm{F}}^2.
\end{equation}

\end{lemma}

\begin{proof}[\textbf{Proof of Lemma~\ref{lem:nabla_f}}]
We first bound each $\mathbb{E}[\|\nabla f_i(\btheta_{i,k})\|^2], i \in [m_g]$.
\begin{equation}
\begin{aligned}
\mathbb{E}[\|\nabla f_i(\btheta_{i,k})\|^2] & =  \mathbb{E}[\|\nabla f_i(\btheta_{i,k})-  \nabla f_i(\tilde{\btheta}_{k}) + \nabla f_i(\tilde{\btheta}_{k}) - \nabla f( \tilde{\btheta}_{k}) +  \nabla f( \tilde{\btheta}_{k})  \|^2] \\
& \leq  3 \mathbb{E} [\|\nabla f_i(\btheta_{i,k})-  \nabla f_i(\tilde{\btheta}_{k})\|^2] + 3  \mathbb{E} [\|\nabla f_i(\tilde{\btheta}_{k}) - \nabla f( \tilde{\btheta}_{k})\|^2]\\
&~~~+ 3 \mathbb{E} [\| \nabla f( \tilde{\btheta}_{k})  \|^2]   \\
& \leq 3L^2 \mathbb{E}[\|\btheta_{i,k} - \btheta_{k} \|^2] + 3 \zeta^2  + 3 \mathbb{E} [\| \nabla f( \tilde{\btheta}_{k})  \|^2],
\end{aligned}
\end{equation}
where the second  inequality is because each $f_i, i \in \Gcal$ is  $L$-smooth  owing to Condition~\ref{asp:loss}(ii).

Then we have
\begin{equation}
\begin{aligned}
& \mathbb{E}\left[\norm{\nabla {\fbf}({\bm \Theta}_k)}_{\mathrm{F}}^2\right]  \leq  \sum\limits_{i=1}^{m_g} \mathbb{E}[\|\nabla f_i(\btheta_{i,k})\|^2]  \\
\leq & 3m_g L^2 \frac{1}{m_g}\sum\limits_{i=1}^{m_g} \mathbb{E}[\|\btheta_{i,k} - \btheta_{k} \|^2] + 3 m_g \frac{1}{m_g}\sum\limits_{i=1}^{m_g} \mathbb{E} [\|\nabla f_i(\tilde{\btheta}_{k}) - \nabla f( \tilde{\btheta}_{k})\|^2] \\
&+ 3 m_g \mathbb{E} [\| \nabla f( \tilde{\btheta}_{k})  \|^2]\\
\leq  & 3L^2m_g M_k + 3m_g \zeta^2 +  3\mathbb{E}\norm{\nabla f(\tilde{\btheta}_k){\bm 1}_{m_g}^{\top}}_{\mathrm{F}}^2.
\end{aligned}
\end{equation}
\end{proof}

\begin{lemma}\label{lem:sum_Mk}
We have the following inequality under Condition~\ref{asp:sure_detection2} and \ref{asp:loss}:
\begin{equation}\label{eq:Mkk}
\sum\limits_{k=0}^{K-1} M_{k} \leq  \frac{2 \eta^2 w^2 c^2}{1-\frac{18\eta^2L^2m_g w^2 c^2}{(1-\rho)^2}}\Biggl[\frac{9 \sum\limits_{j=0}^{K-1} \mathbb{E}\norm{\nabla f(\tilde{\btheta}_j){\bm 1}_{m_g}^{\top}}_{\mathrm{F}}^2 + 9  m_g  \zeta^2 K}{(1-\rho)^2} + \frac{m_g \sigma^2 K}{1-\rho^2}\Biggr],
\end{equation}
and
\begin{equation*}
\begin{aligned}
\sum\limits_{k=0}^{K-1} \rho^k M_{k}
& \leq   \frac{18\eta^2L^2m_g w^2 c^2}{(1-\rho^2)(1-\rho)}\sum\limits_{k=0}^{K-1} M_{k} + \frac{18\eta^2w^2 c^2 }{(1-\rho^2)(1-\rho)}\sum\limits_{k=0}^{K-1} \mathbb{E}\norm{\nabla f(\tilde{\btheta}_k){\bm 1}_{m_g}^{\top}}_{\mathrm{F}}^2\\
& ~~~ + \frac{18 \eta^2 m_g \zeta^2 w^2 c^2 }{(1-\rho)^3} + \frac{2m_g \eta^2 \sigma^2 w^2 c^2}{(1-\rho^2)(1-\rho)}.\\
\end{aligned}
\end{equation*}
\end{lemma}

\begin{proof}[\textbf{Proof of Lemma~\ref{lem:sum_Mk}}]
We bound $Q_{i,k}$ as follows:
\begin{equation}
\begin{aligned}
Q_{i,k} & = \mathbb{E}[\| {\vbf}_1^{\top}{\bm \Theta}_{k} - {\bm 1}_{\{i\}}^{\top}{\bm \Theta}_{k}  \|^2]  \\
& = \mathbb{E}[\|({\vbf}_1^{\top}-{\bm 1}_{\{i\}}^{\top}\Abf){\bm \Theta}_{k-1} - ({\vbf}_1^{\top}-{\bm 1}_{\{i\}}^{\top})\eta \mathrm{Diag}(\Ybf_{k})^{-1} \Gbf({\bm \Theta}_{k-1}, \Xi_{k-1})\|^2]\\
& = \mathbb{E}[\|({\vbf}_1^{\top}- {\bm 1}_{\{i\}}^{\top}\Abf^k){\bm \Theta}_0 - \sum\limits_{j=0}^{k-1} \eta ({\vbf}_1^{\top} - {\bm 1}_{\{i\}}^{\top}\Abf^{k-j-1})\mathrm{Diag}(\Ybf_{j})^{-1} \Gbf({\bm \Theta}_{j}, \Xi_{j})\|^2].\\
\end{aligned}
\end{equation}
Without loss of generality, we can assume ${\bm \Theta}_0={\bm 0}$, then $Q_{i,k}$ can be further bounded by
\begin{equation}
\begin{aligned}
Q_{i,k} & = \mathbb{E}[\| \sum\limits_{j=0}^{k-1} \eta ({\vbf}_1^{\top} - {\bm 1}_{\{i\}}^{\top}\Abf^{k-j-1})\mathrm{Diag}(\Ybf_{j})^{-1} \Gbf({\bm \Theta}_{j}, \Xi_{j})\|^2]\\
& = \mathbb{E}[\| \eta \sum\limits_{j=0}^{k-1}  ({\vbf}_1^{\top} - {\bm 1}_{\{i\}}^{\top}\Abf^{k-j-1})\mathrm{Diag}(\Ybf_{j})^{-1} (\Gbf({\bm \Theta}_{j}, \Xi_{j})- \nabla {\fbf} ({\bm \Theta}_{j})+ \nabla {\fbf} ({\bm \Theta}_{j}))\|^2]\\
& \leq 2\eta^2 \mathbb{E}[\|  \sum\limits_{j=0}^{k-1}  ({\vbf}_1^{\top} - {\bm 1}_{\{i\}}^{\top}\Abf^{k-j-1})\mathrm{Diag}(\Ybf_{j})^{-1} (\Gbf({\bm \Theta}_{j}, \Xi_{j})- \nabla {\fbf} ({\bm \Theta}_{j}))\|^2] \\
& ~~~ + 2\eta^2 \mathbb{E}[\| \sum\limits_{j=0}^{k-1}  ({\vbf}_1^{\top} - {\bm 1}_{\{i\}}^{\top}\Abf^{k-j-1})\mathrm{Diag}(\Ybf_{j})^{-1}  \nabla {\fbf} ({\bm \Theta}_{j})\|^2] \\
& := T_1 +T_2.
\end{aligned}
\end{equation}
For $T_1$, we have
\begin{equation*}
\begin{aligned}
T_1 & = 2\eta^2 \sum\limits_{j=0}^{k-1}\mathbb{E}[\|    ({\vbf}_1^{\top} - {\bm 1}_{\{i\}}^{\top}\Abf^{k-j-1})\mathrm{Diag}(\Ybf_{j})^{-1} (\Gbf({\bm \Theta}_{j}, \Xi_{j})- \nabla {\fbf} ({\bm \Theta}_{j}))\|^2]\\
& \leq 2\eta^2 \sum\limits_{j=0}^{k-1}\mathbb{E}\left[\mathbb{E}[\|\Gbf({\bm \Theta}_{j}, \Xi_{j})- \nabla {\fbf} ({\bm \Theta}_{j})\|_{\mathrm{F}}^2\big| \Fcal_{j}]\right]w^2 c^2 \rho^{2(k-j-1)}\\
& \leq 2m_g \eta^2 \sigma^2 w^2 c^2\sum\limits_{j=0}^{k-1} \rho^{2(k-j-1)} \leq  \frac{2m_g \eta^2 \sigma^2 w^2 c^2}{1-\rho^2}.
\end{aligned}
\end{equation*}

For $T_2$, we have
\begin{equation*}
\begin{aligned}
T_2 =& 2\eta^2 \mathbb{E}[\|  \sum\limits_{j=0}^{k-1}  ({\vbf}_1^{\top} - {\bm 1}_{\{i\}}^{\top}\Abf^{k-j-1})\mathrm{Diag}(\Ybf_{j})^{-1}  \nabla {\fbf} ({\bm \Theta}_{j})\|^2]\\
=& 2\eta^2 \sum\limits_{j=0}^{k-1} \mathbb{E} [\| ({\vbf}_1^{\top} - {\bm 1}_{\{i\}}^{\top}\Abf^{k-j-1})\mathrm{Diag}(\Ybf_{j})^{-1}  \nabla {\fbf} ({\bm \Theta}_{j}) \|^2] \\
& + 2\eta^2 \sum\limits_{j\neq j'} \mathbb{E} \Big\langle ({\vbf}_1^{\top} - {\bm 1}_{\{i\}}^{\top}\Abf^{k-j-1})\mathrm{Diag}(\Ybf_{j})^{-1}  \nabla {\fbf} ({\bm \Theta}_{j}),\\
& ~~~~~~~~~~~~~~~~~~({\vbf}_1^{\top} - {\bm 1}_{\{i\}}^{\top}\Abf^{k-j'-1})\mathrm{Diag}(\Ybf_{j'})^{-1}  \nabla {\fbf} ({\bm \Theta}_{j'})  \Big\rangle \\
:= & T_3 +T_4.
\end{aligned}
\end{equation*}

For $T_3$, by Lemma~\ref{lem:nabla_f}, we have
\begin{equation*}
\begin{aligned}
T_3 & \leq  2\eta^2 \sum\limits_{j=0}^{k-1} \mathbb{E}[\|\nabla \fbf ({\bm \Theta}_{j})\|_{\mathrm{F}}^2] \| ({\vbf}_1^{\top} - {\bm 1}_{\{i\}}^{\top}\Abf^{k-j-1})\mathrm{Diag}(\Ybf_{j})^{-1}\|^2 \\
& \leq  2\eta^2  \sum\limits_{j=0}^{k-1}  \Bigl(3L^2m_g M_j + 3m_g\zeta^2+ 3\mathbb{E}\Bignorm{\nabla f(\tilde{\btheta}_j){\bm 1}_{m_g}^{\top}}_{\mathrm{F}}^2\Bigr) \\
&~~~\times \bignorm{ ({\vbf}_1^{\top} - {\bm 1}_{\{i\}}^{\top}\Abf^{k-j-1})\mathrm{Diag}(\Ybf_{j})^{-1} }^2\\
& \leq 6\eta^2 \sum\limits_{j=0}^{k-1} L^2m_g M_j \| ({\vbf}_1^{\top} - {\bm 1}_{\{i\}}^{\top}\Abf^{k-j-1})\mathrm{Diag}(\Ybf_{j})^{-1}\|^2 + \frac{6m_g \eta^2 \zeta^2 w^2 c^2}{1-\rho^2} \\
& ~~~ + 6\eta^2 \sum\limits_{j=0}^{k-1} \mathbb{E}\norm{\nabla f(\tilde{\btheta}_j){\bm 1}_{m_g}^{\top}}_{\mathrm{F}}^2  \| ({\vbf}_1^{\top} - {\bm 1}_{\{i\}}^{\top}\Abf^{k-j-1})\mathrm{Diag}(\Ybf_{j})^{-1}\|^2 \\
& \leq 6\eta^2 \sum\limits_{j=0}^{k-1} L^2m_g M_j w^2c^2 \rho^{2(k-j-1)} + \frac{6m_g \eta^2 \zeta^2 w^2 c^2}{1-\rho^2}  \\
&~~~+ 6\eta^2 \sum\limits_{j=0}^{k-1} \mathbb{E}\Bignorm{\nabla f(\tilde{\btheta}_j){\bm 1}_{m_g}^{\top}}_{\mathrm{F}}^2  w^2 c^2 \rho^{2(k-j-1)}. \\
\end{aligned}
\end{equation*}

For $T_4$,  by Lemma~\ref{lem:nabla_f}, we have

\begin{equation*}
\begin{aligned}
T_4 & = 2\eta^2 \sum\limits_{j\neq j'}^{k-1} \mathbb{E}\Big\langle ({\vbf}_1^{\top} - {\bm 1}_{\{i\}}^{\top}\Abf^{k-j-1})\mathrm{Diag}(\Ybf_{j})^{-1}  \nabla f ({\bm \Theta}_{j}),\\
& ~~~~~~~~~~~~~~~~~~({\vbf}_1^{\top} - {\bm 1}_{\{i\}}^{\top}\Abf^{k-j'-1})\mathrm{Diag}(\Ybf_{j'})^{-1}  \nabla f ({\bm \Theta}_{j'})  \Big\rangle\\
& \leq 2\eta^2 \sum\limits_{j\neq j'}^{k-1} \mathbb{E}\left[\frac{\|\nabla {\fbf}({\bm \Theta}_{j})\|_{\mathrm{F}}^2+\|\nabla {\fbf}({\bm \Theta}_{j'})\|_{\mathrm{F}}^2}{2}\right] w^2 c^2 \rho^{k-j-1} \rho^{k-j'-1}  \\
& = 2\eta^2 \sum\limits_{j\neq j'}^{k-1} \mathbb{E}\left[\|\nabla {\fbf}({\bm \Theta}_{j})\|_{\mathrm{F}}^2\right] w^2 c^2 \rho^{2k-j-j'-2}  \\
& \leq 2\eta^2 \sum\limits_{j \neq j'}^{k-1} (3L^2m_g M_j
+ 3m_g \zeta^2  + 3\mathbb{E}\norm{\nabla f(\tilde{\btheta}_j){\bm 1}_{m_g}^{\top}}_{\mathrm{F}}^2)w^2 c^2 \rho^{2k-j-j'-2}\\
& = 12\eta^2 \sum\limits_{j =0}^{k-1} (L^2m_g M_j + \mathbb{E}\norm{\nabla f(\tilde{\btheta}_j){\bm 1}_{m_g}^{\top}}_{\mathrm{F}}^2)w^2 \sum\limits_{j'= j+1}^{k-1}\rho^{2k-j-j'-2} \\
&~~~+12 \eta^2 m_g \zeta^2 w^2 c^2 \sum\limits_{j>j'}^{k-1} \rho^{2k-j-j'-2}   \\
& \leq 12\eta^2 \sum\limits_{j =0}^{k-1} (L^2m_g M_j + \mathbb{E}\norm{\nabla f(\tilde{\btheta}_j){\bm 1}_{m_g}^{\top}}_{\mathrm{F}}^2)w^2 c^2 \frac{\rho^{k-j-1}}{1-\rho} + 12 \eta^2 m_g \zeta^2 w^2 c^2  \frac{1}{(1-\rho)^2}.\\
\end{aligned}
\end{equation*}

Plugging $T_3$ and $T_4$ into $T_2$ yields the bound for $T_2$:

\begin{equation*}
\begin{aligned}
T_2 & \leq  6\eta^2  \sum\limits_{j=0}^{k-1} L^2m_g M_jw^2 c^2 (\frac{2\rho^{k-j-1}}{1-\rho}+ \rho^{2(k-j-1)}) + \frac{6m_g\eta^2 \zeta^2w^2 c^2}{1-\rho^2} \\
& ~~~+ 6\eta^2 \sum\limits_{j =0}^{k-1}  \mathbb{E}\norm{\nabla f(\tilde{\btheta}_j){\bm 1}_{m_g}^{\top}}_{\mathrm{F}}^2w^2 c^2 (\frac{2\rho^{k-j-1}}{1-\rho}+ \rho^{2(k-j-1)}) +  \frac{12 \eta^2 m_g \zeta^2 w^2 c^2}{(1-\rho)^2}\\
& \leq  18\eta^2  \sum\limits_{j=0}^{k-1} L^2m_g M_jw^2 c^2 \frac{\rho^{k-j-1}}{1-\rho} \\
&~~~+ 18\eta^2 \sum\limits_{j =0}^{k-1}  \mathbb{E}\norm{\nabla f(\tilde{\btheta}_j){\bm 1}_{m_g}^{\top}}_{\mathrm{F}}^2 \frac{\rho^{k-j-1} w^2 c^2}{1-\rho} + \frac{18 \eta^2 m_g \zeta^2 w^2 c^2}{(1-\rho)^2}.
\end{aligned}
\end{equation*}

Plugging the bound for $T_1$ and $T_2$ into $Q_{i,k}$ yields that:

\begin{equation}\label{eq:Qik}
\begin{aligned}
Q_{i,k} & \leq  18\eta^2  \sum\limits_{j=0}^{k-1} L^2m_g M_jw^2 c^2 \frac{\rho^{k-j-1}}{1-\rho} + 18\eta^2 \sum\limits_{j =0}^{k-1}  \mathbb{E}\norm{\nabla f(\tilde{\btheta}_j){\bm 1}_{m_g}^{\top}}_{\mathrm{F}}^2w^2 c^2 \frac{\rho^{k-j-1}}{1-\rho}\\
& ~~~ + \frac{18 \eta^2 m_g \zeta^2 w^2 c^2}{(1-\rho)^2} +  \frac{2m_g \eta^2 \sigma^2 w^2 c^2}{1-\rho^2}. \\
\end{aligned}
\end{equation}

Taking the average of \eqref{eq:Qik} over $i =1, \ldots, m_g$, we obtain:
\begin{equation}\label{eq:Mk}
\begin{aligned}
M_{k} & \leq  18\eta^2L^2m_g w^2 c^2 \sum\limits_{j=0}^{k-1} M_j \frac{\rho^{k-j-1}}{1-\rho} + 18\eta^2w^2 c^2 \sum\limits_{j =0}^{k-1}  \mathbb{E}\norm{\nabla f(\tilde{\btheta}_j){\bm 1}_{m_g}^{\top}}_{\mathrm{F}}^2 \frac{\rho^{k-j-1}}{1-\rho} \\
& ~~~ +  \frac{18 \eta^2 m_g \zeta^2 w^2 c^2}{(1-\rho)^2} +  \frac{2m_g \eta^2 \sigma^2 w^2 c^2}{1-\rho^2}. \\
\end{aligned}
\end{equation}

Notice that for a non-negative sequence $\{R_{j}\}$, it holds that
\begin{equation}
\begin{aligned}
&\sum\limits_{k=0}^{K-1}\sum\limits_{j=0}^{k-1} R_j \rho^{k-j-1} \\
= & R_0( \rho^{K-1} + \ldots + \rho^{0}) + R_1( \rho^{K-2} + \ldots + \rho^{0}) + \ldots + R_{K-1} \rho^{0} \\
\leq &  \frac{1}{1-\rho}\sum\limits_{j=0}^{K-1} R_j,
\end{aligned}
\end{equation}
and
\begin{equation}
\begin{aligned}
&\sum\limits_{k=0}^{K-1} \rho^k \sum\limits_{j=0}^{k-1} R_j \rho^{k-j-1} = \sum\limits_{k=0}^{K-1}  \sum\limits_{j=0}^{k-1} R_j \rho^{2k-j-1} \\
\leq & \sum\limits_{k=0}^{K-1}  \sum\limits_{j=0}^{k-1} R_j \rho^{2(k-j-1)} \leq \frac{1}{1-\rho^2}\sum\limits_{j=0}^{K-1} R_j.
\end{aligned}
\end{equation}

Summing \eqref{eq:Mk} from $k = 0$ to $K -1$, we get:

\begin{equation*}\label{eq:sumMk}
\begin{aligned}
\sum\limits_{k=0}^{K-1} M_{k} & = 18\eta^2L^2m_g w^2 c^2 \sum\limits_{k=0}^{K-1}\sum\limits_{j=0}^{k-1} M_j \frac{\rho^{k-j-1}}{1-\rho} + 18\eta^2w^2 c^2 \sum\limits_{k=0}^{K-1}\sum\limits_{j =0}^{k-1}  \mathbb{E}\norm{\nabla f(\tilde{\btheta}_j){\bm 1}_{m_g}^{\top}}_{\mathrm{F}}^2 \frac{\rho^{k-j-1}}{1-\rho} \\
& ~~~ + \sum\limits_{k=0}^{K-1}\frac{18 \eta^2 m_g \zeta^2 w^2 c^2}{(1-\rho)^2}+  \sum\limits_{k=0}^{K-1} \frac{2m_g \eta^2 \sigma^2 w^2 c^2}{1-\rho^2}. \\
& \leq  \frac{18\eta^2L^2m_g w^2 c^2}{(1-\rho)^2} \sum\limits_{j=0}^{K-1} M_j  + \frac{18\eta^2w^2 c^2}{(1-\rho)^2} \sum\limits_{j=0}^{K-1} \mathbb{E}\norm{\nabla f(\tilde{\btheta}_j){\bm 1}_{m_g}^{\top}}_{\mathrm{F}}^2 \\
& ~~~ + \frac{18 \eta^2 m_g \zeta^2 w^2 c^2}{(1-\rho)^2}K + \frac{2m_g \eta^2 \sigma^2 w^2 c^2K}{1-\rho^2}. \\
\end{aligned}
\end{equation*}

Rearranging the terms in the above bound, we get the bound for $\sum_{k=0}^{K-1} M_{k}$:
\begin{equation}
\begin{aligned}
\sum\limits_{k=0}^{K-1} M_{k} & \leq  \frac{1}{1-\frac{18\eta^2L^2m_g w^2 c^2}{(1-\rho)^2}}\frac{18\eta^2w^2 c^2}{(1-\rho)^2} \sum\limits_{j=0}^{K-1} \mathbb{E}\norm{\nabla f(\tilde{\btheta}_j){\bm 1}_{m_g}^{\top}}_{\mathrm{F}}^2 \\
& ~~~ + \frac{1}{1-\frac{18\eta^2L^2m_g w^2 c^2}{(1-\rho)^2}} \left(\frac{18 \eta^2 m_g \zeta^2 w^2 c^2}{(1-\rho)^2}K + \frac{2m_g \eta^2 \sigma^2 w^2 c^2K}{1-\rho^2}\right). \\
\end{aligned}
\end{equation}

Furthermore, multiplying both sides of \eqref{eq:Mk}  by $\rho^k$, and summing this inequality over $k = 0, \ldots, K -1$ gives
\begin{equation}
\begin{aligned}
& \sum\limits_{k=0}^{K-1} \rho^k M_{k}\\
\leq & 18\eta^2L^2m_g w^2 c^2 \sum\limits_{k=0}^{K-1} \rho^k \sum\limits_{j=0}^{k-1} M_j \frac{\rho^{k-j-1}}{1-\rho} \\
& + 18\eta^2w^2 c^2 \sum\limits_{k=0}^{K-1} \rho^k \sum\limits_{j =0}^{k-1}  \mathbb{E}\norm{\nabla f(\tilde{\btheta}_j){\bm 1}_{m_g}^{\top}}_{\mathrm{F}}^2 \frac{\rho^{k-j-1}}{1-\rho} \\
&+\sum\limits_{k=0}^{K-1} \rho^k \frac{18 \eta^2 m_g \zeta^2 w^2 c^2}{(1-\rho)^2}  + \sum\limits_{k=0}^{K-1} \rho^k \frac{2m_g \eta^2 \sigma^2 w^2 c^2}{1-\rho^2} \\
 \leq &  \frac{18\eta^2L^2m_g w^2 c^2}{(1-\rho^2)(1-\rho)}\sum\limits_{k=0}^{K-1} M_{k} + \frac{18\eta^2w^2 c^2 }{(1-\rho^2)(1-\rho)}\sum\limits_{k=0}^{K-1} \mathbb{E}\norm{\nabla f(\tilde{\btheta}_k){\bm 1}_{m_g}^{\top}}_{\mathrm{F}}^2 \\
 & ~ + \frac{18 \eta^2 m_g \zeta^2 w^2 c^2 }{(1-\rho)^3} + \frac{2m_g \eta^2 \sigma^2 w^2 c^2 }{(1-\rho^2)(1-\rho)}.\\
\end{aligned}
\end{equation}

\end{proof}

\begin{proof}[\textbf{Proof of Theorem~\ref{thm:main0}}]
According to Condition~\ref{asp:loss}, we know each $f_i$ is Lipschitz smooth with parameter $L$, which implies $f$ is also Lipschitz smooth with parameter $L$.
Then we have
\begin{equation*}
\begin{aligned}
 f(\tilde{\btheta}_{k+1})  =  & f\left({\vbf}_1^{\top}\Abf{\bm \Theta}_k - \eta {\vbf}_1^{\top}\tilde\Ybf_k^{-1} \Gbf({\bm \Theta}_k; \Xi_k)\right)\\
  \overset{L-\text{smooth}}{\leq} & f(\tilde{\btheta}_k) - m_g \eta \left\langle \nabla f(\tilde{\btheta}_k), \frac{{\vbf}_1^{\top}\tilde\Ybf_k^{-1}}{m_g} \Gbf({\bm \Theta}_k; \Xi_k) \right\rangle + \frac{m_g^2 \eta^2 L}{2} \norm{\frac{{\vbf}_1^{\top}\tilde\Ybf_k^{-1}}{m_g} \Gbf({\bm \Theta}_k; \Xi_k)}^2.
\end{aligned}
\end{equation*}
Taking expectations of both sides conditioned on $\mathcal{F}_{k}$, we have
\begin{equation*}
\begin{aligned}
& \mathbb{E}[f(\tilde{\btheta}_{k+1})|\mathcal{F}_{k}] \\
\leq & f(\tilde{\btheta}_k) - m_g \eta \left\langle \nabla f(\tilde{\btheta}_k), \frac{{\vbf}_1^{\top}\tilde\Ybf_k^{-1}\nabla {\fbf}({\bm \Theta}_k)}{m_g}  \right\rangle + \frac{m_g^2 \eta^2 L}{2} \mathbb{E}\left[\norm{\frac{{\vbf}_1^{\top}\tilde\Ybf_k^{-1}\Gbf({\bm \Theta}_k; \Xi_k)}{m_g}}^2 \Bigg| \mathcal{F}_{k}\right] \\
= & f(\tilde{\btheta}_k)  - \frac{m_g \eta}{2}\norm{\frac{{\vbf}_1^{\top}\tilde\Ybf_k^{-1}\nabla {\fbf}({\bm \Theta}_k)}{m_g}}^2 + \frac{m_g \eta}{2}\norm{\nabla f(\tilde{\btheta}_k)- \frac{{\vbf}_1^{\top}\tilde\Ybf_k^{-1}\nabla {\fbf}({\bm \Theta}_k)}{m_g}}^2\\
& - \frac{m_g \eta}{2} \|\nabla f(\tilde{\btheta}_k)\|^2 + \frac{m_g^2 \eta^2 L}{2} \mathbb{E}\left[\norm{\frac{{\vbf}_1^{\top}\tilde\Ybf_k^{-1}\Gbf({\bm \Theta}_k; \Xi_k)}{m_g}}^2 \Bigg| \mathcal{F}_{k}\right].\\
\end{aligned}
\end{equation*}

Taking expectations with respect to $\mathcal{F}_{k}$, and employing Lemma  \ref{lem:square_S_control},  \ref{lem:gradient_diff} and tower property, we attain that
\begin{equation}\label{eq:all}
\begin{aligned}
\mathbb{E}[f(\tilde{\btheta}_{k+1})]
\leq &  \mathbb{E}\left[f(\tilde{\btheta}_k)\right]   - (\frac{m_g \eta}{2} - \frac{m_g^2\eta^2 L}{2})\mathbb{E}\left[\norm{\frac{{\vbf}_1^{\top}\tilde\Ybf_k^{-1}\nabla {\fbf}({\bm \Theta}_k)}{m_g}}^2\right] + \eta m_g L^2 M_k \\
& + \frac{ \eta w^4 \|{\vbf}_1\|^2 \|\Ybf_k - \Ybf_{\infty}\|^2_2}{m_g} \mathbb{E}\left[\norm{\nabla {\fbf}({\bm \Theta}_k)}_{\mathrm{F}}^2\right] \\
&+ \frac{m_g \eta^2 w^2\|{\vbf}_1\|^2 \sigma^2 L}{2} - \frac{m_g \eta}{2} \mathbb{E}[\|\nabla f(\tilde{\btheta}_k)\|^2].
\end{aligned}
\end{equation}

Rearranging the terms in \eqref{eq:all}, we have
\begin{equation}\label{eq:all_1}
\begin{aligned}
& \sum\limits_{k=0}^{K-1}(\frac{m_g \eta}{2} - \frac{m_g^2\eta^2 L}{2})\mathbb{E}\left[\norm{\frac{\vbf_1^{\top}(\tilde\Ybf_k)^{-1}\nabla {\fbf}({\bm \Theta}_k)}{m_g}}^2\right] + \sum\limits_{k=0}^{K-1} \frac{m_g \eta}{2} \mathbb{E}[\|\nabla f(\tilde{\btheta}_k)\|^2]\\
\leq &  f(\tilde{\bm \theta}_0) -f^\ast + \frac{m_g \eta^2 w^2 \|{\vbf}_1\|^2 \sigma^2  LK}{2} + \sum\limits_{k=0}^{K-1} \eta m_g L^2 M_k \\
&+ \sum\limits_{k=0}^{K-1}  \frac{ \eta w^4 \|{\vbf}_1\|^2 \|\Ybf_k - \Ybf_{\infty}\|^2_2}{m_g} \mathbb{E}\left[\norm{\nabla {\fbf}({\bm \Theta}_k)}_{\mathrm{F}}^2\right].
\end{aligned}
\end{equation}

Let $V_k:= \frac{ \eta w^2  \|\Ybf_k - \Ybf_{\infty}\|_2}{m_g} \mathbb{E}\left[\norm{\nabla {\fbf}({\bm \Theta}_k)}_{\mathrm{F}}^2\right]$.
Plugging \eqref{eq:nabla_f} into $V_{k}$ and combining Lemma~\ref{lem:sum_Mk}, we obtain
\begin{equation*}
\begin{aligned}
&~~~ \sum\limits_{k=0}^{K-1} V_k \\
& \leq \frac{ \eta w^2 c }{m_g} \sum\limits_{k=0}^{K-1} \rho^k [ 3m_g\zeta^2 + 3m_g L^2 M_k] +   \sum\limits_{k=0}^{K-1} \frac{ 3\eta w^4 \|{\vbf}_1\|^2 \|\Ybf_k - \Ybf_{\infty}\|^2_2}{m_g}  \mathbb{E}(\|\nabla f(\tilde{\btheta}_k){\bm 1}_{m_g}^{\top}\|_{\mathrm{F}}^2)]  \\
& \leq \frac{ \eta w^2 c }{m_g}  \left[ \frac{3m_g\zeta^2}{1-\rho} +  \frac{54 \eta^2L^4 m_g^2 w^2 c^2}{(1-\rho^2)(1-\rho)}\sum\limits_{k=0}^{K-1} M_{k} + \frac{54 m_g L^2\eta^2w^2 c^2 }{(1-\rho^2)(1-\rho)}\sum\limits_{k=0}^{K-1} \mathbb{E}\norm{\nabla f(\tilde{\btheta}_k){\bm 1}_{m_g}^{\top}}_{\mathrm{F}}^2\right. \\
& \left.~~~~~~~~~~~~+ \frac{54 \eta^2 m_g^2 L^2\zeta^2 w^2 c^2 }{(1-\rho)^3} + \frac{6m_g^2 L^2 \eta^2 \sigma^2 w^2 c^2}{(1-\rho^2)(1-\rho)} \right] \\
&~~~+ \frac{ \eta w^4 \|{\vbf}_1\|^2 \|\Ybf_k - \Ybf_{\infty}\|^2_2}{m_g} \sum\limits_{k=0}^{K-1}   3\mathbb{E}(\|\nabla f(\tilde{\btheta}_k){\bm 1}_{m_g}^{\top}\|_{\mathrm{F}}^2)]  \\
& =   \frac{3\eta w^2 c  \zeta^2}{1-\rho} + \frac{54 \eta^3 L^4 m_g w^4 c^3 }{(1-\rho^2)(1-\rho)}\sum\limits_{k=0}^{K-1} M_{k} + \frac{54  L^2\eta^3 w^4 c^3  }{(1-\rho^2)(1-\rho)}\sum\limits_{k=0}^{K-1} \mathbb{E}\norm{\nabla f(\tilde{\btheta}_k){\bm 1}_{m_g}^{\top}}_{\mathrm{F}}^2 \\
& ~~~+ \frac{54 \eta^3 m_g L^2 \zeta^2 w^4 c^3 }{(1-\rho)^3} + \frac{6m_g L^2 \eta^3 \sigma^2 w^4 c^3 }{(1-\rho^2)(1-\rho)}   \\
&~~~+  \sum\limits_{k=0}^{K-1}   \frac{ 3\eta w^4 \|{\vbf}_1\|^2 \|\Ybf_k - \Ybf_{\infty}\|^2_2}{m_g} \mathbb{E}(\|\nabla f(\tilde{\btheta}_k){\bm 1}_{m_g}^{\top}\|_{\mathrm{F}}^2)].\\
\end{aligned}
\end{equation*}

Plugging the above bound into \eqref{eq:all_1}, we have

\begin{equation*}
\begin{aligned}
& \sum\limits_{k=0}^{K-1}(\frac{m_g \eta}{2} - \frac{m_g^2\eta^2 L}{2})\mathbb{E}\left[\norm{\frac{\vbf_1^{\top}(\tilde\Ybf_k)^{-1}\nabla {\fbf}({\bm \Theta}_k)}{m_g}}^2\right] \\
&~~~+ \sum\limits_{k=0}^{K-1} \left(\frac{m_g \eta}{2}-\frac{54 m_g L^2\eta^3 w^4 c^3  }{(1-\rho^2)(1-\rho)}\right) \mathbb{E}[\|\nabla f(\tilde{\btheta}_k)\|^2]\\
 & \leq f(\tilde{\btheta}_0) -f^\ast + \frac{m_g \eta^2 w^2 \|{\vbf}_1\|^2 \sigma^2  LK}{2} + \left(\eta m_g L^2+\frac{54 \eta^3 L^4 m_g w^4 c^3 }{(1-\rho^2)(1-\rho)}\right) \sum\limits_{k=0}^{K-1} M_k  \\
& ~~~+ \frac{54 \eta^3 m_g L^2 \zeta^2 w^4 c^3 }{(1-\rho)^3}  + \frac{6m_g L^2 \eta^3 \sigma^2 w^4 c^3 }{(1-\rho^2)(1-\rho)} \\
&~~~+  \sum\limits_{k=0}^{K-1}   \frac{3 \eta w^4 \|{\vbf}_1\|^2 \|\Ybf_k - \Ybf_{\infty}\|^2_2}{m_g}\mathbb{E}(\|\nabla f(\tilde{\btheta}_k){\bm 1}_{m_g}^{\top}\|_{\mathrm{F}}^2)].\\
\end{aligned}
\end{equation*}

Plugging \eqref{eq:Mkk} and rearranging the order, we get

\begin{equation*}
\begin{aligned}
& \sum\limits_{k=0}^{K-1}(\frac{m_g \eta}{2} - \frac{m_g^2\eta^2 L}{2})\mathbb{E}\biggl[\biggnorm{\frac{\vbf_1^{\top}(\tilde\Ybf_k)^{-1}\nabla {\fbf}({\bm \Theta}_k)}{m_g}}^2\biggr] \\
&+ \sum\limits_{k=0}^{K-1} \mathbb{E}[\|\nabla f(\tilde{\btheta}_k)\|^2] \Big\{\frac{m_g \eta}{2}- 3 \eta w^4 \|{\vbf}_1\|^2 \|\Ybf_k - \Ybf_{\infty}\|^2_2 - \frac{18\eta^3 m_g^2 L^2 w^2 c^2}{(1-\rho)^2(1-\frac{18\eta^2L^2m_g w^2 c^2}{(1-\rho)^2})} \\
& ~~~~~~~~~~~~~~~~~~~~~~~~~~~ -\frac{54 m_g L^2\eta^3 w^4 c^3  }{(1-\rho^2)(1-\rho)}  - \frac{972\eta^5 L^4 m_g^2 w^6 c^5}{(1-\rho)^3(1-\rho^2)(1-\frac{18\eta^2L^2m_g w^2 c^2}{(1-\rho)^2})}  \Big\}  \\
\leq & f(\tilde{\btheta}_0) -f^\ast \\
& + \frac{m_g \eta^2 w^2 \|{\vbf}_1\|^2 \sigma^2  LK}{2}  + \frac{3\eta w^2 c  \zeta^2}{1-\rho}    + \frac{54 \eta^3 m_g L^2 \zeta^2 w^4 c^3 }{(1-\rho)^3} + \frac{6m_g L^2 \eta^3 \sigma^2 w^4 c^3 }{(1-\rho^2)(1-\rho)} \\
& +  \left(\eta m_g L^2+\frac{54 \eta^3 L^4 m_g w^4 c^3 }{(1-\rho^2)(1-\rho)}\right) \frac{1}{1-\frac{18\eta^2L^2m_g w^2 c^2}{(1-\rho)^2}} \left( \frac{18 \eta^2 m_g \zeta^2 w^2 c^2}{(1-\rho)^2}K+ \frac{2m_g \eta^2 \sigma^2 w^2 c^2 K}{1-\rho^2}\right) \\
\leq & f(\tilde{\btheta}_0) -f^\ast \\
& + \frac{m_g \eta^2 w^2 \|{\vbf}_1\|^2 \sigma^2  LK}{2}  + \frac{3\eta w^2 c  \zeta^2}{1-\rho}   + \frac{54 \eta^3 m_g L^2 \zeta^2 w^4 c^3 }{(1-\rho)^3} + \frac{6m_g L^2 \eta^3 \sigma^2 w^4 c^3 }{(1-\rho^2)(1-\rho)} \\
& + \frac{18 \eta^3 m_g^2 L^2 \zeta^2 w^2 c^2K}{(1-\rho)^2(1-\frac{18\eta^2L^2 m_g w^2 c^2}{(1-\rho)^2})}  + \frac{2\eta^3 m_g^2 L^2 \sigma^2 w^2 c^2 K}{(1-\rho^2)(1-\frac{18\eta^2L^2 m_g w^2 c^2}{(1-\rho)^2})}  \\
& + \frac{972 \eta^5 m_g^2 L^4  \zeta^2 w^6 c^5 K}{(1-\rho)^3 (1-\rho^2)(1-\frac{18\eta^2L^2 m_g w^2 c^2}{(1-\rho)^2})} + \frac{108 m_g^2\eta^5 L^4 \sigma^2 w^6 c^5  K}{(1-\rho^2)^2(1-\rho)(1-\frac{18\eta^2L^2m_g w^2 c^2}{(1-\rho)^2})}.\\
\end{aligned}
\end{equation*}

Dividing $\eta K$ by the both sides, we obtain
\begin{equation}\label{eq:all_3}
\begin{aligned}
& \frac{1}{K}\sum\limits_{k=0}^{K-1}(\frac{m_g}{2} - \frac{m_g^2\eta L}{2})\mathbb{E}\left[\norm{\frac{\vbf_1^{\top}(\tilde\Ybf_k)^{-1}\nabla {\fbf}({\bm \Theta}_k)}{m_g}}^2\right]\\
&+ \frac{1}{K}\sum\limits_{k=0}^{K-1} \bigg(\frac{m_g}{2}- 3 w^4 \|{\vbf}_1\|^2 \|\Ybf_k - \Ybf_{\infty}\|^2_2 - \frac{18\eta^2 m_g^2 L^2 w^2 c^2}{(1-\rho)^2(1-\frac{18\eta^2L^2m_g w^2 c^2}{(1-\rho)^2})}\\
&~~~~~~~~~~~~ -\frac{54 m_g L^2\eta^2 w^4 c^3  }{(1-\rho^2)(1-\rho)} - \frac{972\eta^4 L^4 m_g^2 w^6 c^5}{(1-\rho)^3(1-\rho^2)(1-\frac{18\eta^2L^2m_g w^2 c^2}{(1-\rho)^2})}  \bigg) \mathbb{E}[\|\nabla f(\tilde{\btheta}_k)\|^2] \\
\leq & \frac{f(\tilde{\btheta}_0) -f^\ast}{K\eta} + \frac{m_g \eta w^2 \|{\vbf}_1\|^2 \sigma^2  L}{2} + \frac{3 w^2 c  \zeta^2}{(1-\rho)K}   + \frac{54 \eta^2 m_g L^2 \zeta^2 w^4 c^3 }{(1-\rho)^3 K} \\
& + \frac{6m_g L^2 \eta^2 \sigma^2 w^4 c^3 }{(1-\rho^2)(1-\rho)K} +  \frac{18 \eta^2 m_g^2 L^2 \zeta^2 w^2 c^2}{(1-\rho)^2(1-\frac{18\eta^2L^2 m_g w^2 c^2}{(1-\rho)^2})}  + \frac{2\eta^2 m_g^2 L^2 \sigma^2 w^2 c^2 }{(1-\rho^2)(1-\frac{18\eta^2L^2 m_g w^2 c^2}{(1-\rho)^2})} \\
&  + \frac{972 \eta^4 m_g^2 L^4  \zeta^2 w^6 c^5 }{(1-\rho)^3 (1-\rho^2)(1-\frac{18\eta^2L^2 m_g w^2 c^2}{(1-\rho)^2})} + \frac{108 m_g^2\eta^4 L^4 \sigma^2 w^6 c^5 }{(1-\rho^2)^2(1-\rho)(1-\frac{18\eta^2L^2m_g w^2 c^2}{(1-\rho)^2})}.\\
\end{aligned}
\end{equation}

Set $t_0:= \lceil \frac{\log(m_g /48c^2w^4\|{\vbf}_1\|^2)}{2\log(\rho)}  \rceil$ and $\Ybf_{0}:= \Abf^{t_0}$.
By Lemma~\ref{lem:E}, we have
\begin{equation*}
\frac{m_g}{2}- 3  w^4 \|{\vbf}_1\|^2 \|\Ybf_k - \Ybf_{\infty}\|^2_2 \geq \frac{m_g}{2} - 3c^2w^4 \|{\vbf}_1\|^2 \rho^{2t_0} \geq \frac{7m_g}{16}.
\end{equation*}

Now we substitute $\eta = \sqrt{\frac{1}{m_g K}}$ in \eqref{eq:all_3}.
Note that $K\geq \max\{\frac{1}{m_g}, 4m_g L^2, D_3, D_4\}$, it can be inferred that $\eta \leq \min\{\frac{1}{2m_g L},  \frac{1-\rho}{12\sqrt{2} \sqrt{m_g} L e c}, \frac{\sqrt{(1-\rho)(1-\rho^2)}}{12\sqrt{3}  L w^2 c^{3/2} },  1\}$.
Then we have the followings:
\begin{equation*}
\begin{aligned}
& \frac{m_g}{2}-\frac{m_g^2 \eta L}{2}  >0, \\
& 1-\frac{18\eta^2L^2m_g w^2 c^2}{(1-\rho)^2}  > \frac{1}{2},\\
& \frac{54 m_g L^2\eta^2 w^4 c^3  }{(1-\rho^2)(1-\rho)}  \leq \frac{54 m_g L^2\eta w^4 c^3  }{(1-\rho^2)(1-\rho)} \leq \frac{m_g}{8},\\
& \frac{18\eta^2 m_g^2 L^2 w^2 c^2}{(1-\rho)^2(1-\frac{18\eta^2L^2m_g w^2 c^2}{(1-\rho)^2})}  \leq \frac{m_g}{8},\\
& \frac{972\eta^4 L^4 m_g^2 w^6 c^5}{(1-\rho)^3(1-\rho^2)(1-\frac{18\eta^2L^2m_g w^2 c^2}{(1-\rho)^2})}  \leq \frac{m_g}{16},
\end{aligned}
\end{equation*}
which implies
\begin{equation*}
\begin{aligned}
\frac{m_g}{8} & \leq  \frac{m_g}{2}-\frac{54 m_g L^2\eta^2 w^4 c^3  }{(1-\rho^2)(1-\rho)}- 3  w^4 \|{\vbf}_1\|^2 \|\Ybf_k - \Ybf_{\infty}\|^2_2- \frac{18\eta^2 m_g^2 L^2 w^2 c^2}{(1-\rho)^2(1-\frac{18\eta^2L^2m_g w^2 c^2}{(1-\rho)^2})}\\
& ~ - \frac{972\eta^4 L^4 m_g^2 w^6 c^5}{(1-\rho)^3(1-\rho^2)(1-\frac{18\eta^2L^2m_g w^2c^2}{(1-\rho)^2})}.\\
\end{aligned}
\end{equation*}
Removing the term $\mathbb{E}\left[\norm{\frac{\vbf_1^{\top}\tilde\Ybf_k^{-1}\nabla {\fbf}({\bm \Theta}_k)}{m_g}}^2\right]$ in \eqref{eq:all_3}, we have
\begin{equation*}
\begin{aligned}
 &\frac{1}{K}\sum\limits_{k=0}^{K-1} \frac{m_g}{8} \mathbb{E}[\|\nabla f(\tilde{\btheta}_k)\|^2] \\
 & \leq  \frac{\sqrt{m_g}(f(\tilde{\bm \theta}_0) -f^\ast)}{\sqrt{K}} + \frac{\sqrt{m_g} w^2 \|{\vbf}_1\|^2 \sigma^2  L}{2\sqrt{K}}+ \frac{3 w^2 c  \zeta^2}{(1-\rho)K}   + \frac{54  L^2 \zeta^2 w^4 c^3 }{(1-\rho)^3 K^{2}} \\
 & ~ + \frac{6 L^2  \sigma^2 w^4 c^3 }{(1-\rho^2)(1-\rho)K^2} + \frac{2 m_g L^2 \sigma^2 w^2 c^2 }{(1-\rho^2)(1-\frac{18\eta^2L^2m_g w^2 c^2}{(1-\rho)^2})K}  +  \frac{18  m_g L^2 \zeta^2 w^2 c^2}{(1-\rho)^2(1-\frac{18\eta^2L^2 m_g w^2 c^2}{(1-\rho)^2})K} \\
 & ~ + \frac{972  L^4  \zeta^2 w^6 c^5 }{(1-\rho)^3 (1-\rho^2)(1-\frac{18\eta^2L^2 m_g w^2 c^2}{(1-\rho)^2})K^2} + \frac{108  L^4 \sigma^2 w^6 c^5 }{(1-\rho^2)^2(1-\rho)(1-\frac{18\eta^2L^2m_g w^2 c^2}{(1-\rho)^2})K^2}\\
& \leq   \frac{2\sqrt{m_g}(f(\tilde{\bm \theta}_0) -f^\ast)+ \sqrt{m_g} w^2 \|{\vbf}_1\|^2 \sigma^2  L}{2\sqrt{K}} + \frac{3 w^2 c  \zeta^2 + 4 m_g L^2 \sigma^2 w^2 c^2 + 36  m_g L^2 \zeta^2 w^2 c^2 }{(1-\rho) K} \\
& ~ + \frac{54 L^2 \zeta^2 w^4 c^3+ 216L^4w^6c^5\sigma^2+ 6 L^2  \sigma^2 w^4 c^3 + 1944  L^4  \zeta^2 w^6 c^5}{(1-\rho)^3K^2}.
\end{aligned}
\end{equation*}

Using again the condition that $K\geq \max\{D_1, D_2\}$, we have
\begin{equation*}
\frac{1}{K}\sum\limits_{k=0}^{K-1} \frac{m_g}{8} \mathbb{E}[\|\nabla f(\tilde{\btheta}_k)\|^2] \leq  \frac{5\sqrt{m_g}(f(\tilde{\bm \theta}_0) -f^\ast)+ \sqrt{m_g} w^2 \|{\vbf}_1\|^2 \sigma^2  L}{2\sqrt{K}},
\end{equation*}
which finally leads to
\begin{equation*}
\frac{1}{K}\sum\limits_{k=0}^{K-1}  \mathbb{E}[\|\nabla f(\tilde{\btheta}_k)\|^2] \leq  \frac{20(f(\tilde{\bm \theta}_0) -f^\ast)+ 4 w^2 \|{\vbf}_1\|^2 \sigma^2  L}{\sqrt{Km_g}}.
\end{equation*}

\end{proof}

\begin{proof}[\textbf{Proof of Corollary~\ref{coro:consensus}}]
By Lemma~\ref{lem:sum_Mk}, we have
   \begin{equation*}
\frac{1}{K}\sum\limits_{k=0}^{K-1} M_{k} \leq  \frac{2}{K}\left[\frac{18\eta^2w^2 c^2}{(1-\rho)^2} \sum\limits_{j=0}^{K-1} \mathbb{E}\norm{\nabla f(\tilde{\btheta}_j){\bm 1}_{m_g}^{\top}}_{\mathrm{F}}^2 + \frac{18 \eta^2 m_g \zeta^2 w^2 c^2}{(1-\rho)^2}K  + \frac{2m_g \eta^2 \sigma^2 w^2 c^2K}{1-\rho^2}\right].
\end{equation*}

Employing Theorem~\ref{thm:main0},  we can obtain

\begin{equation*}
\begin{aligned}
\frac{1}{K}\sum\limits_{k=0}^{K-1} M_{k} & \leq \frac{36 \eta^2w^2 c^2}{(1-\rho)^2} \frac{20(f(\tilde{\bm \theta}_0) -f^\ast)+ 4 w^2 \|{\vbf}_1\|^2 \sigma^2  L}{K^{1/2} m_g^{1/2}} + \frac{36 \eta^2 m_g \zeta^2 w^2 c^2}{(1-\rho)^2} + \frac{4m_g \eta^2 \sigma^2 w^2 c^2}{1-\rho^2}\\
& \leq  \frac{36 w^2 c^2}{(1-\rho)^2} \frac{20(f(\tilde{\bm \theta}_0) -f^\ast)+ 4 w^2 \|{\vbf}_1\|^2 \sigma^2  L}{K^{3/2} m_g^{3/2}} + \frac{36  \zeta^2 w^2 c^2}{(1-\rho)^2 K} + \frac{4 \sigma^2 w^2 c^2}{(1-\rho^2) K}\\
& = \mathcal{O}\Bigl(\frac{\zeta^2 +\sigma^2}{K} +\frac{1}{K^{3/2}}\Bigr).
\end{aligned}
\end{equation*}

\end{proof}

\bibliographystyle{plain}
\bibliography{ref}

\end{document}